\documentclass[11pt]{article}
\usepackage{amsmath, amssymb}
\usepackage{mathtools}
\usepackage{algorithm}
\usepackage{algpseudocode}
\usepackage{fmtcount}
\usepackage{graphicx}
\usepackage{enumitem} 
\usepackage{subcaption}
\usepackage{afterpage}
\usepackage[noadjust]{cite}
\usepackage[font=small,labelfont=bf]{caption}
\usepackage[labelfont=bf]{caption}
\usepackage{setspace}
\allowdisplaybreaks
\usepackage[title,titletoc]{appendix}
\usepackage{epsfig}
\usepackage{graphicx}
\usepackage{hhline}
\usepackage{latexsym}
\usepackage{amssymb}
\usepackage{amsmath,amscd}
\usepackage{multirow} 
\usepackage{array}
\usepackage{caption}
\usepackage{subcaption}
\usepackage{authblk}
\usepackage{color}
\usepackage{natbib}
\usepackage{rotating}
\usepackage{graphicx,lscape}
\usepackage{amsmath, amsthm, amssymb}
\usepackage{graphicx}
\usepackage{epsfig}
\usepackage{graphics,color, graphicx}
\usepackage{latexsym}
\usepackage{fullpage}
\usepackage{booktabs,fixltx2e}
\usepackage[flushleft]{threeparttable}
\usepackage{amssymb,amsmath,amscd}
\usepackage{fancyvrb}
\usepackage{pdfpages}
\usepackage{fmtcount}
\usepackage{url}
\usepackage[margin=1in]{geometry}

\theoremstyle{plain}
\newtheorem{proposition}{Proposition}[section]
\newtheorem{theorem}[proposition]{Theorem}
\newtheorem{corollary}[proposition]{Corollary}
\newtheorem{lemma}[proposition]{Lemma}
\theoremstyle{definition}
\newtheorem{definition}[proposition]{Definition}
\newtheorem{remark}[proposition]{Remark}
\newtheorem{assumption}[proposition]{Assumption}

\title{\bf A data-based notion of quantiles on Hadamard spaces
	\medskip
}

\author[1]{Ha-Young Shin}
\author[2]{Hee-Seok Oh}
\affil[1]{
	Department of Statistics and Actuarial Science, Soongsil University
}
\affil[2]{
	Department of Statistics, Seoul National University
}
{
    \makeatletter
    \renewcommand\AB@affilsepx{: \protect\Affilfont}
    \makeatother

    \affil[ ]{Email}

    \makeatletter
    \renewcommand\AB@affilsepx{, \protect\Affilfont}
    \makeatother

    \affil[1]{hayoung.shin@gmail.com}
    \affil[2]{heeseok@stats.snu.ac.kr}
}

\begin{document}
	\maketitle
	
	\begin{abstract}
This paper defines an alternative notion, described as data-based, of geometric quantiles on Hadamard spaces, in contrast to the existing methodology, described as parameter-based. In addition to having the same desirable properties as parameter-based quantiles, these data-based quantiles are shown to have several theoretical advantages related to large-sample properties like strong consistency and asymptotic normality, breakdown points, extreme quantiles and the gradient of the loss function. Using simulations, we explore some other advantages of the data-based framework, including simpler computation and better adherence to the shape of the distribution, before performing experiments with real diffusion tensor imaging data lying on a manifold of symmetric positive definite matrices. These experiments illustrate some of the uses of these quantiles by testing the equivalence of the generating distributions of different data sets and measuring distributional characteristics.

		\noindent
		
		\vspace{\baselineskip}
		
		
	\end{abstract}
	
	\pagenumbering{arabic}

 \section{Introduction}

The mean and median of a random variable $X$ can be defined as the minimizer $p\in \mathbb{R}$ of the expected values of $\lvert X-p\rvert^2$ and $\lvert X-p\rvert$, respectively. Similarly, though quantiles were originally conceived of in relation to the ordering of data, they can also be defined as the minimizers of expected values of loss functions: $2|X-p|\{(1-\tau)I(X\leq p)+\tau I(X>p)\}$ for $\tau\in(0,1)$. When $\tau=1/2$, the loss is equal to the median loss. 

The natural generalizations of the mean and median loss functions to the multivariate case are $\lVert x-p\rVert^2$ and $\lVert x-p\rVert$, respectively. The generalization of the quantile loss function is less clear, but \cite{Chaudhuri1996} suggested
\begin{equation}\begin{aligned}\label{dbmulti}
    \lVert x-p\rVert+\langle u,x-p\rangle,
\end{aligned}\end{equation}
where $u\in B^n(1):=\{v|v\in\mathbb{R}^n,\lVert v\rVert<1\}$. This gives rise to a geometric notion of a quantile. When $u=0$, we get the geometric median loss function. By letting $u=2\tau-1$ when $M=\mathbb{R}$, this becomes $|x-p|+u(x-p)=2|x-p|\{(1-\tau)I(x\leq p)+\tau I(x>p)\}$, the univariate quantile loss function. Note that $(-1,1)=B^1(1)$.

Hadamard spaces, also called global non-positive curvature spaces or complete CAT(0) spaces, are a class of metric spaces that have recently been receiving interest among statisticians.
\cite{Sturm2002} developed a theory of non-linear martingales and \cite{Sturm2003} studied probability theory on such spaces. \cite{Yun2023} generalized the notion of the median-of-means to these spaces and derived concentration inequalities for the median-of-means as an estimator of the population Fr\'echet mean, while \cite{Kostenberger2023} proved a strong law of large numbers for random elements with independent but not necessarily identically distributed values in Hadamard space under very weak conditions and applied this to the problem of robust signal recovery. \cite{Zhang2016} analyzed first-order algorithms for convex optimization on Hadamard manifolds.

Examples of such spaces include inner product spaces, Euclidean buildings and tree spaces. Hadamard spaces that are also Riemannian manifolds are called Hadamard manifolds, and can equivalently be characterized as complete, simply connected Riemannian manifolds of non-positive sectional curvature. Common examples of Hadamard manifolds include the spaces of symmetric positive definite matrices, which are useful in covariance matrix analysis and diffusion tensor imaging (see for example \cite{Zhu2009}), and hyperbolic spaces. Hyperbolic spaces in particular are becoming increasingly popular among machine learning researchers as they are natural homes for hierarchical data; see for example \cite{Weber2020}. An Hadamard space $M$ is uniquely geodesic (that is, there exists a unique unit-speed minimal geodesic between any two points in $M$), and crucially for our purposes, it's so-called boundary at infinity $\partial M$ can be used to define a canonical notion of direction on the space. In brief, $\partial M$ is a set of equivalence classes of asymptotic unit-speed geodesic rays, and each equivalence class $\xi\in\partial M$ contains exactly one ray $\gamma_p:[0,\infty)\rightarrow M$ issuing from each point $p$ in $M$; then $\gamma_p$ can be thought of as pointing in the direction of $\xi$ from $p$. Also, angles, called Alexandrov angles, can be defined between geodesics in $M$. Refer to Section 2 of \cite{Shin2023} for the relevant geometric background on metric spaces, geodesics and Hadamard spaces and manifolds needed for the current paper. 

\cite{Shin2023} further generalized geometric quantiles to Hadamard spaces by defining loss functions indexed by $(\beta,\xi)\in[0,1)\times\partial M$: $d(p,x)+\beta d(p,x)\cos(\angle_p(x,\xi))$ when $p\neq x$ and 0 when $p=x$. Here $\angle_p(x,\xi)$ is the Alexandrov angle between the unique geodesic ray issuing from $p$ in the direction of $\xi$ and the unique unit-speed geodesic joining $p$ to $x$; the $x=p$ case needs to be defined separately because this angle is not defined in that case. On Hadamard manifolds, this loss function becomes 
\begin{equation}\begin{aligned}\label{dboriginal}
\lVert \log_p(x)\rVert+\langle\beta\xi_p,\log_p(x)\rangle,
\end{aligned}\end{equation}
where $\xi_p$ is the unique tangent vector in $T_pM$ in the direction of $\xi$. Making the natural identification between $\partial\mathbb{R}^n$ and $S^{n-1}$ and letting $u=\beta\xi$, or equivalently $\beta=\lVert u\rVert$ and $\xi=u/\lVert u\rVert$, this loss function is identical to (\ref{dbmulti}) when $M=\mathbb{R}^n$.

The loss functions in (\ref{dbmulti}) and (\ref{dboriginal}) can be thought of as being calculated using two tangent vectors at $p$: $\beta\xi_p$ and $\log_p(x)$, which becomes $u$ and $x-p$, respectively, in the Euclidean case. In the general Hadamard space case, the angle in the loss function is measured at $p$. Because the loss function is calculated at the potential quantile $p$ in this way, we may call this the \textit{parameter-based} perspective. But notice that (\ref{dbmulti}) equals
\begin{equation}\begin{aligned}\label{dbdata}
    \lVert p-x\rVert-\langle u,p-x\rangle,
\end{aligned}\end{equation}
and thus can be calculated using tangent vectors $u, p-x$ in $T_xM$ instead of $T_pM$. Then there are natural extensions to Hadamard spaces and manifolds, which will be explored throughout this paper.
  Because the tangent vectors in (\ref{dbdata}) are located at the potential data point $x$, we may call this the \textit{data-based} perspective. Though the two perspectives are of course equivalent in the Euclidean case, this is not so in general on Hadamard spaces.

Conceptually, the parameter-based approach feels natural. Indeed, the original formulation of the multivariate quantile loss function by \cite{Chaudhuri1996} uses this approach, and so did the original paper on quantiles on Hadamard spaces by \cite{Shin2023}. However, as will be seen, we are able to do much more with data-based quantiles theoretically and practically than parameter-based quantiles, resulting in a paper that goes far beyond \cite{Shin2023}.

The uses of geometric quantiles in the multivariate setting can be generalized to our data-based quantiles. Quantiles reveal the overall structure of a distribution at a granular level beyond what the mean and median, which only indicate the center, can offer. This logic can be extended to quantile regression, dealt with by \cite{Chakraborty2003} in the multivariate context, which reveals the structure of the relationship between variables in finer detail than ordinary least squares regression can. Quantiles can also be used to test whether different data sets are generated by the same distribution, for example using permutation tests, in cases where they are indistinguishable using the mean or median, as they offer more points of comparison, and in fact \cite{Koltchinskii1997} showed that multivariate distributions are uniquely defined by their geometric quantiles. One can also define a notion of ranks, closely related to quantiles, which take values in tangent spaces, and use them for tests of location. Isoquantile contours can be used for outlier detection after a transformation--retransformation procedure, as demonstrated in \cite{Chaouch2010} for multivariate geometric quantiles and \cite{Shin2023} for parameter-based quantiles on Hadamard manifolds. \cite{ShinOh2024} defined geometric quantile-based measures of skewness, dispersion, kurtosis and spherical asymmetry for multivariate distributions. See Section 6 of \cite{Shin2023} for more details of these applications as applied to parameter-based quantiles.

Section \ref{dbdef} generalizes (\ref{dbdata}) to define data-based quantiles on Hadamard spaces and gives some basic properties, including some of the benefits of the data-based perspective. In Section \ref{dbasymp}, we demonstrate the advantages of data-based quantiles in terms of the broader applicability of large-sample properties like strong consistency, and use a new method to prove a stronger theorem of joint asymptotic normality. Notwithstanding the significant differences thus far, the rest of the paper represents a complete break from \cite{Shin2023}. Section \ref{dbrobustness} and Section \ref{dbextreme} show some intuitive properties of breakdown points and extreme data-based quantiles, respectively. Properties related to the gradient of the data-based quantile loss function on Hadamard manifolds, including an explicit expression for the gradient on locally symmetric Hadamard spaces, are explored in Section \ref{dbhyp}. Section \ref{dbsims} contains experiments with simulated data in 2-dimensional hyperbolic space, in which we illustrate a method for simpler computation of data-based quantiles and demonstrate one of the applications mentioned in the previous paragraph---testing the equivalence distributions underlying different data sets with quantiles and permutations tests. Section \ref{dbreal} contains experiments with real diffusion tensor imaging data on the space of $3\times 3$ symmetric positive definite matrices, with which we demonstrate another of the applications mentioned in the previous paragraph---measuring distributional characteristics. Finally, we close the paper with a discussion in Section \ref{dbdiscussion}.


\section{Definition and basic properties} \label{dbdef}
	
	Given an Hadamard space $(M,d)$ equipped with its metric topology and Borel $\sigma$-algebra $\mathcal{B}$, let $X$ be a random element in $M$, that is, a measurable map from some probability space $(\Omega,\mathcal{F},P)$ into $(M,\mathcal{B})$. For $\beta\in[0,\infty)$ and $\xi\in\partial M$, define
 \begin{equation} \label{dbexpect}
		G^{\beta,\xi}(p)=E[\rho(X,p;\beta,\xi)],
		\end{equation}
		where
		\begin{equation} \label{dbphi}
		\rho(x,p;\beta,\xi)=d(p,x)-\beta d(p,x)\cos(\angle_x(p,\xi))
  \end{equation}
		and $\angle_x(p,\xi)$ is the Alexandrov angle at $x$ between the unique geodesic from $x$ to $p$ and the geodesic ray that is the unique member of $\xi$ issuing from $x$. We adopt the convention throughout this paper that $f(y)g(z)=0$ if a function $f$ is 0 at $y$ and a function $g$ is infinite or undefined at $z$, so that when $x=p$, $\rho(x,p;\beta,\xi)=d(p,x)=0$ even though $\angle_p(x,\xi)$ is not defined. When $M$ is an Hadamard manifold, (\ref{dbphi}) becomes
	\begin{equation}\begin{aligned} \label{dbriemquantile}
	\rho(x,p;\beta,\xi)=\lVert \log_x(p)\rVert-\langle \beta\xi_x,\log_x(p)\rangle, 
	\end{aligned}\end{equation}
	where $\xi_p\in T_pM$ is the unit vector in the tangent space at $p$ that is the velocity of the unique geodesic ray in $\xi$ issuing from $p$.
 
	\begin{definition} \label{dbquantile}
		For $\beta\in[0,1)$ and $\xi\in\partial M$, the \textit{$(\beta,\xi)$-quantile set} of $X$ is defined to be $q(\beta,\xi)=\arg\min_{p\in M}G^{\beta,\xi}(p)$. Any element of this set is called a \textit{$(\beta,\xi)$-quantile} of $X$. 
	\end{definition}

\begin{definition}
		Given data points $X_1,\ldots,X_N\in M$, the \textit{sample $(\beta,\xi)$-quantile (set)} is defined to be $\hat{q}_N(\beta,\xi)=\arg\min_{p\in M}\hat{G}^{\beta,\xi}_N(p)$ where $\hat{G}^{\beta,\xi}_N(p)=\frac{1}{N}\sum_{i=1}^N\rho(X_i,p;\beta,\xi)$.
	\end{definition}

As with parameter-based quantiles, the measurability of $\rho$ as a function of its second argument, and the equivariance of quantile sets with respect to scaled isometries is guaranteed. Compare the next two propositions to Propositions 3.1 and 3.2, respectively, in \cite{Shin2023}. The proofs of all propositions in this section are in Appendix \ref{dbproof_measurable}.

\begin{proposition} \label{dbmeasurable}
		For fixed $\beta\in[0,\infty)$, $\xi\in\partial M$, and $p\in M$, the map $x\mapsto\rho(x,p;\beta,\xi)$ is measurable. 
	\end{proposition}

We define a $\sigma$-scaled isometry for some $\sigma>0$ to be a bijection $g:M\rightarrow M'$  between metric spaces $(M,d)$ and $(M',d')$ such that $d'(g(x),g(y))=\sigma d(x,y)$ for all $x,y\in M$. For any geodesic $\gamma$ into $M$, define $g\gamma$ by $g\gamma(t)=g(\gamma(t/\sigma))$, its domain being the set of all $t$ such that $t/\sigma$ is in the domain of $\gamma$. 

Throughout this paper, for a set $A$ and a function $f$, denote $\{f(a)|a\in A\}$ by $f(A)$.

	\begin{proposition} \label{dbequivariance}
		Let $M$ and $M'$ be Hadamard spaces and $g:M\rightarrow M'$ a $\sigma$-scaled isometry. Then the $(\beta,g\xi)$-quantile set of $g(X)$, where $g\xi:=[g\gamma|\gamma\in\xi=[\gamma]]\in\partial M'$, is $g(q(\beta,\xi))$.
	\end{proposition}

This property is worth mentioning because besides the quantiles of \cite{Chaudhuri1996}, other non-geometric notions of quantiles for multivariate data have been attempted but they generally lack this equivariance; for example, coordinate-wise medians are not equivariant to rotations, which are isometries. This is an advantage of geometric quantiles.

Proposition \ref{dbnew} provides a simple example of something whose corresponding result is not true in the case of parameter-based quantiles, for which only lower semi-continuity can be guaranteed on general Hadamard spaces.

 \begin{proposition} \label{dbnew}
The map $p\mapsto\rho(x,p;\beta,\xi)$ is continuous.
 \end{proposition}

The following proposition provides another example of the advantages of the data-based perspective; the corresponding result in \cite{Shin2023}, Proposition 3.4, applies only to Hadamard manifolds. 

\begin{proposition} \label{dbbasic}
		Let $M$ be an Hadamard space. Assume that $G^{\beta^*,\xi^*}(p^*)$ is finite for some $\beta^*\in[0,1),\xi^*\in\partial M, p^*\in M$.
		\begin{itemize}
			\item[(a)] For every $(\beta,\xi)\in[0,\infty)\times \partial M$, $G^{\beta,\xi}$ is finite and continuous on all of $M$, and
			
			\item[(b)] If in addition $M$ is locally compact, for every $(\beta,\xi)\in[0,1)\times \partial M$, the $(\beta,\xi)$-quantile set is nonempty and compact.
		\end{itemize}
	\end{proposition}

  In the rest of this paper, the following condition is assumed to be true.

  \begin{assumption} \label{dbassump}
For some $\beta^*\in[0,1),\xi^*\in\partial M$ and $p^*\in M$, $G^{\beta^*,\xi^*}(p^*)<\infty$.
  \end{assumption}

   \section{Large-sample properties} \label{dbasymp}

 This section, which in large part adapts the results of Section 4 of \cite{Shin2023} to the data-based case, demonstrates some of the powerful theoretical advantages of the data-based approach. Most of the remarks in Section 4 of \cite{Shin2023} also apply here, but we avoid repeating them for the sake of brevity. Throughout this section, $X,X_1,X_2,\ldots$ are independent and identically distributed $M$-valued random elements.

\subsection{Strong consistency} \label{dbsc}

Though \cite{Shin2023} defines quantiles on Hadamard spaces, with the exception of some basic properties it focuses entirely on Hadamard manifolds, while in this section we prove the crucial statistical property of strong consistency on a much broader class of Hadamard spaces, namely locally compact Hadamard spaces. We use the fact that Proposition \ref{dbbasic} applies to all Hadamard spaces.
	
	\begin{lemma} \label{dbsulln}
		Let $M$ be an Hadamard space. For any compact $L\subset M$, $\lim_{N\rightarrow \infty}\sup_{p\in L}\Big\lvert\hat{G}^{\beta,\xi}_N(p)-G^{\beta,\xi}(p)\Big\rvert=0$ almost surely.
	\end{lemma}
	A proof is provided in Appendix~\ref{dbproof_sulln}. For the next theorem and its proof, also in Appendix~\ref{dbproof_sulln}, we adopt the notation $d(p,A):=\inf_{p'\in A}d(p,p')$ for $p\in M$, $A\in M$.
	
	\begin{theorem} \label{dbslln}
 Let $M$ be a locally compact Hadamard space.
		\begin{itemize}
			\item[(a)] For any $\epsilon>0$, there exist some $\Omega_1\in\mathcal{F}$ and $N_1(\omega)<\infty$ for all $\omega\in\Omega_1$ such that $P(\Omega_1)=1$ and the sample $(\beta,\xi)$-quantile set of $X_1,\ldots,X_N$ is contained in $C^\epsilon=\{p\in M:d(p,q(\beta,\xi))<\epsilon\}$ for all $N\geq N_1(\omega)$.
			
			\item[(b)] If $X$ has a unique $(\beta,\xi)$-quantile, then any measurable choice from the sample $(\beta,\xi)$-quantile set of $X_1,\ldots,X_N$ converges almost surely to the $(\beta,\xi)$-quantile of $X$.
		\end{itemize}
	\end{theorem}

  \subsection{Joint asymptotic normality}\label{dbasymp}

Throughout this paper, we adopt the notation $E[X;X\in A]:=E[XI(X\in A)]$. In this section, we will make use of several norms and inner products besides the Riemannian ones, so we denote the Riemannian norm and inner product on $M$ by $\lVert \cdot\rVert_g$ and $\langle\cdot,\cdot\rangle_g$, respectively. All proofs for this section are provided in Appendix~\ref{dbproof_qgrad}.

Recall that $\log_x:M\rightarrow T_xM$ is smooth, and in fact it is a diffeomorphism, so $p\mapsto\rho(x,p;\beta,\xi)$ is smooth on all of $M\backslash\{p\}$. Denote by $\nabla \rho(x,\phi(\mu))$ the Riemannian gradient at $p\neq x$ of this map. The following results, which are similar to Theorem 3.1 and Corollary 3.1 of \cite{Shin2023}, are necessary for our proofs of asymptotic normality, but they are also interesting in their own rights as necessary conditions for data-based sample quantiles. Unlike Theorem 3.1 in \cite{Shin2023}, the following theorem does not require the support to be bounded, which as a consequence will lead to simpler conditions for asymptotic normality.

 \begin{theorem} \label{dbqgrad}
		Let $M$ be an Hadamard manifold and $X$ be an $M$-valued random element. If $q\in q(\beta,\xi)$, then $\lVert E[\nabla \rho(X,q;\beta,\xi);X\neq q]-P(X=q)\beta\xi_{q}\rVert_g\leq P(X=q)$. Thus if $P(X=q)=0$, then
		$E[\nabla \rho(X,q;\beta,\xi);X\neq q]=0$. 
	\end{theorem}

	Any smooth global chart $\phi$ for an $n$-dimensional manifold $M$ is a diffeomorphism between $M$ and an open subset of $\mathbb{R}^n$. A $\phi$ always exists because $M$ is an Hadamard manifold. For example, the inverse exponential map at any $p\in M$ can define a $\phi$. In a slight abuse of notation that is standard practice, this section identifies $M$ as the image under $\phi$ so that all points $p$ in $M$ and tangent vectors in $T_pM$ are identified with their $n\times 1$ local coordinate representations in the chart $\phi$ and thus can be thought of as vectors in $\mathbb{R}^n$. This identifies the tangent bundle $TM$ with $\phi(M)\times\mathbb{R}^n$. Let $g_p$ be the $n\times n$ matrix that is the coordinate representation of the Riemannian metric $g$ at $b$ and define $\Psi(x,p;\beta,\xi)\in T_pM$ as 
	\begin{align*}
	\Psi(x,p;\beta,\xi)=(\Psi^1(x,p;\beta,\xi),\ldots,\Psi^n(x,p;\beta,\xi))^T
	=\begin{cases}
	g_p\nabla \rho(x,p;\beta,\xi)) &\text{ if $x\neq p$}, \\
	-\beta d(\log_x)_p^\dagger\xi_x &\text{ if $x=p$},
	\end{cases}
	\end{align*} 
	for $(x,p)\in M\times M$, where $d(\log_x)_p^\dagger:T_{\log_x(p)}T_xM\cong T_xM\rightarrow T_pM$ is the adjoint of the differential of $\log_x:M\rightarrow T_xM$; that is, $\langle d(\log_x)_p^\dagger u_1, u_2\rangle_g=\langle u_1, d(\log_x)_pu_2\rangle_g$ for any $u_1\in T_xM, u_2\in T_pM$. For $x\neq p$ and $v\in T_pM$, 
	\begin{equation}\begin{aligned} \label{dbequiv}
	\langle v,\Psi(x,p;\beta,\xi)\rangle_2=\langle v,\nabla \rho(x,p;\beta,\xi)\rangle_g=d(\rho(x,\cdot;\beta,\xi))_p(v),
	\end{aligned}\end{equation}
	and therefore, when $p\neq x$, $\Psi(x,p;\beta,\xi)$ is the Euclidean (not Riemannian) gradient of $\rho(x,p;\beta,\xi)$ as a function of its second argument. Again, recalling the smoothness of this function when $p\neq x$, denote the Euclidean Hessian matrix of this function when $p\neq x$ by $D\Psi(x,p;\beta,\xi)$ and, for $r,r'=1,\ldots,n$, the $(r,r')$-entry of this matrix by $D_{r'}\Psi^r(x,p;\beta,\xi)$; $D\Psi(x,p;\beta,\xi)$ is the Jacobian matrix of $\Psi(x,p;\beta,\xi)$ as a function of its second argument $p\neq x$. For convenience, let each $D\Psi(x,p;\beta,\xi)$ be 0 when $p=x$.
	
	\sloppy Given $(\beta_1,\xi_1),\ldots,(\beta_K,\xi_K)\in[0,1)\times\partial M$, denote $\rho(x,p;\beta_k,\xi_k)$, $\Psi(x,p;\beta_k,\xi_k)$, $D\Psi(x,p;\beta,\xi)$ and $D_{r'}\Psi^r(x,p;\beta_k,\xi_k)$ by the shorthands $\rho_k(x,p)$, $\Psi_k(x,p)$, $D\Psi_k(x,p)$ and $D_{r'}\Psi_k^r(x,p)$, respectively, for each $k=1,\ldots,K$. We use $q_k$ and $\hat{q}_{k,N}$ as shorthands to denote the quantile $q(\beta_k,\xi_k)$ and some measurable selection from the sample quantile set $\hat{q}_N(\beta_k,\xi_k)$, respectively, for $k=1,\ldots,K$.

 Denote the $L_2$ norm and standard Euclidean inner product on $\mathbb{R}^n$ by $\lVert \cdot\rVert_2$ and $\langle\cdot,\cdot\rangle_2$, respectively. In the proofs, we also make reference to the $\sup$, or $L_\infty$, norm of a vector with $\lVert\cdot\rVert_\infty$, and the Frobenius norm of a matrix with $\lVert\cdot\rVert_F$.

 As noted near the beginning of this section, this theorem requires simpler conditions than Theorem 4.2 in \cite{Shin2023}. We do not need to stipulate that $E[\Psi_k(X,q_k)]=0$ in (I) or that $E[\lVert\Psi(X,q_k)\rVert_2^2]<\infty$. Crucially, a $C^2$ condition on the loss functions is unnecessary, and this is carried through to Corollary \ref{dbbdd} and Proposition \ref{dbsupp}. Thus unlike the results in the corresponding section of \cite{Shin2023}, those in this section apply to all Hadamard manifolds.

	\begin{theorem} \label{dbclt1}
		Let $M$ be an $n$-dimensional Hadamard manifold, $n\geq2$. Suppose $X$ satisfies the following conditions for each $k=1,\ldots,K$:
		\begin{enumerate}[label=(\Roman*)]
			\item  the quantile $q_k$ exists uniquely and there is some neighborhood $Q_k\subset M$ around $q_k$ in which the density of $X$ exists and is bounded,
			\item there exists a positive number $d_1>0$ and neighborhood $U\subset M$ of $q_k$ that is bounded in the Riemannian metric such that $\{q:\lVert q-q_k\rVert_2\leq d_1\}\subset U$ and 
   \begin{equation*}
   E\bigg[\sup_{q:\lVert q-q_k\rVert_2\leq d_1}\lvert D_{r'}\Psi_k^r(X,q)\rvert^2;X\not\in \bar{U}\bigg]<\infty
   \end{equation*}
   for all $r,r'=1,\ldots,n$, and
			\item the $n\times n$ matrix $\Lambda^k$, defined by $\Lambda_{r,r'}^k=E[D_{r'}\Psi_k^r(X,q_k)]$ for $r,r'=1,\ldots,n$, is nonsingular.
		\end{enumerate}
		Additionally, if $n=2$, suppose the following condition is satisfied for each $k=1,\ldots,K$:
		\begin{enumerate}[label=(\Roman*)]
			\setcounter{enumi}{3}
			\item there exist $\eta_0>0$ and $\alpha>1$ such that $\sup_{w\in S^{n-1}}E[h_{r,r'}(\eta_0,w,X)^\alpha]<\infty$ for all $r,r'=1,\ldots,n$, where $h_{r,r'}$ is a map on $(0,\infty)\times S^{n-1}\times M$ defined by 
   \begin{equation*}
   h_{r,r'}(\eta,w,x)=\sup_{q:q\in q_k+[0,\eta]\times w}\lvert D_{r'}\Psi_k^r(x,q)-D_{r'}\Psi_k^r(x,q_k)\rvert I(x\not\in q_k+[0,\eta]\times w),
   \end{equation*}
   and $q_k+[0,\eta]\times w:=\{q_k+sw:0\leq s\leq\eta\}$.
		\end{enumerate}
		Then, $\sqrt{N}[(\hat{q}_{1,N}^T,\ldots,\hat{q}_{K,N}^T)^T-(q_1^T,\ldots,q_K^T)^T]\rightsquigarrow N(0,\Lambda^{-1}\Sigma\Lambda^{-1})$ as $N\rightarrow\infty$, where $\Lambda$ is the $Kn\times Kn$ block diagonal matrix whose $k$-th $n\times n$ diagonal block is $\Lambda^k$, and $\Sigma$ is the covariance matrix of $(\Psi_1(X,q_1)^T,\ldots,\Psi_K(X,q_K)^T)^T$.
	\end{theorem}

    The need for special consideration in the 2-dimensional case when it comes to asymptotic properties of geometric quantiles has Euclidean precedent; see, for example, Theorems 3.1.1 and 3.2.1 of \cite{Chaudhuri1996}.

	\begin{corollary} \label{dbbdd}
Let $M$ be an $n$-dimensional Hadamard manifold, $n\geq2$, and $X$ have bounded support according to the Riemannian metric. Suppose that for each $k=1,\ldots,K$, $q_k$ exists uniquely, there is some neighborhood around $q_k$ in which the density of $X$ exists and is bounded, (III) in Theorem \ref{dbclt1} holds and, if $n=2$, (IV) in Theorem \ref{dbclt1} holds. Then the conclusion of Theorem \ref{dbclt1} follows.
	\end{corollary}
		
	\begin{proposition} \label{dbsupp}

Let $M$ be an $n$-dimensional Hadamard manifold, $n\geq2$, and $X$ have bounded support according to the Riemannian metric. If, for any $k=1,\ldots,K$, $q_{k}$ exists uniquely and is not in the support of $X$ and (III) in Theorem \ref{dbclt1} holds, then the conditions for that $k$ in Corollary \ref{dbbdd} are satisfied.

	\end{proposition}

Theorem 2.11 of \cite{Eltzner2019}, which they call a smeary central limit theorem, proves asymptotic normality for minimizers of general loss functions under certain conditions. One may ask whether this theorem can be applied to quantiles on Hadamard manifolds. With parameter-based quantiles, it was explained in Remark 4.6 of \cite{Shin2023} that Theorem 4.2 of that paper and its corollary do not seem to follow from the smeary central limit theorem. However, for data-based quantiles, we can use the smeary central limit theorem to prove the following, which is stronger than Theorem \ref{dbclt1} when $n\geq 3$ as it weakens (II). Thus this is a new result, proven in a completely different way, without parallel in \cite{Shin2023}.

 \begin{theorem} \label{dbclt2}
		Let $M$ be an $n$-dimensional Hadamard manifold, $n\geq 3$. Suppose $X$ satisfies the following conditions for each $k=1,\ldots,K$: (I) and (III) from Theorem \ref{dbclt1}, and 
		\begin{enumerate}
			\item[(IIa)] there exists positive numbers $d_1>0$ and $\alpha>1$, and neighborhood $U\subset M$ of $q_k$ that is bounded in the Riemannian metric such that $\{q:\lVert q-q_k\rVert_2\leq d_1\}\subset U$ and 
   \begin{equation*}
   E\bigg[\sup_{q:\lVert q-q_k\rVert_2\leq d_1}\lvert D_{r'}\Psi_k^r(X,q)\rvert^\alpha;X\not\in \bar{U}\bigg]<\infty
   \end{equation*}
   for all $r,r'=1,\ldots,n$.
		\end{enumerate}
		Then the conclusion of Theorem \ref{dbclt1} follows.
	\end{theorem}

 \begin{remark}\label{dbmistake}
There is a significant typo in the statement of Theorem 2.11 of \cite{Eltzner2019}. Forgetting the notation used in the current paper for the duration of this remark and using instead the notation of that paper, the covariance of the limiting distribution should be $(1/r^2)T^{-1}R\text{Cov}[\text{grad}|_{x=0} \rho(x,X)]R^TT^{-1}$ and not $(1/r^2)T^{-1}\text{Cov}[\text{grad}|_{x=0} \rho(x,X)]T^{-1}$. This is because the proof involves showing that the sequence of random vectors is equal to $-(1/r)T^{-1}RG_n+o_p(1)$, where $G_n=\sqrt{n}((1/n)\sum_{j=1}^n\text{grad}|_{x=0} \rho(x,X_j)-E[\text{grad}|_{x=0} \rho(x,X)])$, from which our corrected limiting distribution immediately follows using the classical central limit theorem.
  \end{remark}

  Again, as mentioned at the beginning of this section, see the remarks of Section 4 of \cite{Shin2023} for more insights, including about the conditions of Theorems \ref{dbclt1} and \ref{dbclt2}.

 \section{Breakdown points}\label{dbrobustness}

Perhaps the most common global measure of robustness is the breakdown point, which indicates the smallest fraction of a sample that needs to be corrupted to produce arbitrary changes in the statistics. That is, given an ordered sample $\mathbf{X}:=(X_1,\ldots,X_N)$ and a statistic $\hat T:M^N\rightarrow M'$, where $(M',d')$ is some metric space, the breakdown point of $\hat T$ at $\mathbf{X}$ is
\begin{equation*}
    \min_{j\in\{1,\ldots,N\}} \bigg\{\frac{j}{N}:\sup_{Y_1,\ldots,Y_N:\sum_{i=1}^NI(X_i\neq Y_i)\leq j}d'(\hat T(\mathbf{X}),\hat T(Y_1,\ldots,Y_N))=\infty\bigg\},
\end{equation*}
so the $\sup$ is taken over all ordered samples $(Y_1,\ldots,Y_N)$ that differ from $\mathbf{X}$ by at most $j$ observations. 

\begin{definition}\label{quantilefunction}
A \textit{sample $(\beta,\xi)$-quantile function} is a function $\hat T:M^N\rightarrow M$ that maps $(Y_1,\ldots,Y_N)$ to an element of the sample $(\beta,\xi)$-quantile set defined by $(Y_1,\ldots,Y_N)$. A \textit{sample median function} is a sample $(0,\xi)$-quantile function for any $\xi\in\partial M$.
\end{definition}

\begin{theorem}\label{dbbdp}
Fix $\beta\in[0,1)$ and a positive integer $N$. Then for any Hadamard manifold $M$, sample $\mathbf{X}\in M^N$ and $\xi\in\partial M$, the breakdown point at $\mathbf{X}$ of any sample $(\beta,\xi)$-quantile function is

(a) $(\lceil(1-\beta)N/2\rceil)/N$ if $(1-\beta)N/2\not\in\mathbb{Z}$, and

(b) $(\lceil(1-\beta)N/2\rceil)/N$ or $(\lceil(1-\beta)N/2\rceil+1)/N$ if $(1-\beta)N/2\in\mathbb{Z}$. 

(c) For either potential breakdown point in (b), there exist some Hadamard manifold $M$, $\xi\in\partial M$ and sample $(\beta,\xi)$-quantile function for which that breakdown point is achieved everywhere in $M^N$ if $(1-\beta)N/2\in\mathbb{Z}$; therefore the statement in (b) is sharp.
\end{theorem}

Proofs for this section are in Appendix \ref{dbproof_robustness}.

\begin{remark}\label{dbelab}
    Theorem \ref{dbbdp}(c) shows that, despite the ambiguity regarding the exact breakdown point, Theorem \ref{dbbdp}(b) cannot be improved upon under the given conditions. In fact, Theorem \ref{dbbdp}(c) actually shows more than this, as the sharpness of Theorem \ref{dbbdp}(b) would only require the potential breakdown point to be achieved at some sample, not everywhere in $M^N$.

    Theorem 2 of \cite{Fletcher2009} and Theorem 2.2 of \cite{Lopuhaa1991} claim to give an exact breakdown point of $(\lfloor (N+1)/2\rfloor)/N=(\lceil N/2\rceil)/N$ for the geometric median on Riemannian manifolds and $\mathbb{R}^n$, respectively, contradicting Theorem \ref{dbbdp}(c). However, reading their proofs, one sees that neither has accounted for the possible non-uniqueness of the sample median. In the case of \cite{Lopuhaa1991}, their proof assumes that the sample median is translation equivariant but there exist selections from the sample median set that are not translation equivariant. For example, consider the sample quantile function that chooses the sample quantile $q=(q^1,\ldots,q^n)$ that minimizes the norm $\lVert q\rVert$, and in the event of multiple such quantiles, uses the coordinates in the order $j=1,\ldots, N$ as tiebreakers by choosing the sample quantile among the remaining candidates that minimizes $q^j$. Such a choice exists because quantile sets are compact by Proposition \ref{dbbasic}(b), as are their intersections with linear subspaces of $\mathbb{R}^n$. Minimizing the norm is the crucial step here as it breaks translation equivariance. In the proof of Theorem \ref{dbbdp}(c), we explicity demonstrate that the breakdown point of this sample quantile function is $(\lceil N/2\rceil+1)/N$, and not $(\lceil N/2\rceil)/N$.
    
    Translation equivariance is of course guaranteed if the sample median is unique as \cite{Lopuhaa1991} seem to assume. Thus a sufficient, though not necessary, condition under which the proof would be valid is that there is no straight line containing $N-\lceil N/2\rceil=\lfloor N/2\rfloor$ or more of the data points, which guarantees the non-collinearity of the sample and hence the uniqueness of the sample median of both the original data set and any data set in which $\lceil N/2\rceil$ or fewer data points are corrupted.
    
    \cite{Konen2024} give exact breakdown points of $(\lceil(1-\beta)N/2\rceil)/N$ in $\mathbb{R}^n$, however this does not contradict Theorem \ref{dbbdp}(c) because we treat all sample quantile functions and they specifically deal with the barycentre of the sample quantile set.
\end{remark}

For any irrational $\beta$, and therefore for almost all (according to the Lebesgue measure) values of $\beta\in[0,1)$, Theorem \ref{dbbdp}(a) fully determines the breakdown point of $\hat q_N(\beta,\xi)$ for all positive integers $N$. Theorem \ref{dbbdp} also reveals that the asymptotic breakdown point of $\hat q_N(\beta,\xi)$ as $N\rightarrow\infty$ is $(1-\beta)/2$.

Although we are most commonly concerned with robustness for estimators of location such as sample quantiles, it is also important to ensure robustness for estimators of dispersion; for example, see \cite{Rousseeuw1992} and \cite{Lopuhaa1991}.

\cite{ShinOh2024} introduced geometric quantile-based measures of dispersion for multivariate data, and \cite{Shin2023} generalized these to Hadamard manifold-valued data using parameter-based quantiles. In the same way, we can use data-based quantiles to define the following measures of dispersion on Hadamard manifolds:
\begin{equation}\begin{gathered}\label{dbdispersions}
    \delta_1(\beta):=\sup_{v\in S_{m}^{n-1}}{\lVert \log_{m}(q(\beta,f_m(v)))-\log_{m}(q(\beta,f_m(-v)))\rVert}, \\
\delta_2(\beta):=\frac{1}{SA(n-1)}\int_{S_{m}^{n-1}}{\lVert \log_{m}(q(\beta,f_m(v)))-\log_{m}(q(\beta,f_m(-v)))\rVert} dv,
\end{gathered}\end{equation}
where $m$ is the geometric median, $SA(n-1)$ is the surface area of the unit $(n-1)$-sphere and for any $y\in M$, $S_{y}^{n-1}$ is the unit sphere in $T_{y}M$ and $f_y:S_{y}^{n-1}\rightarrow \partial M$ is the bijection defined by $f(\xi_{y})=\xi$. For simplicity, we have assumed the uniqueness of the geometric median and, for fixed $\beta\in (0,1)$, all $(\beta,\xi)$-quantiles.

The quantities in (\ref{dbdispersions}) can be estimated using a sample and some appropriately chosen $\xi_1,\ldots,\xi_K\in\partial M$. Given $\beta\in(0,1)$, let $\hat m_N:M^N\rightarrow M$ be a sample median function and, for each $k=1,\ldots,K$, $\hat q_{k,N}:M^N\rightarrow M$ a sample $(\beta,\xi_k)$-quantile function and $\hat p_{k,N}:M^N\rightarrow M$ a function that maps $\mathbf{X}\in M^N$ to an element of the sample $(\beta,f_{\hat m_N(\mathbf{X})}(-f_{\hat m_N(\mathbf{X})}^{-1}(\xi_k)))$-quantile set. Define $\hat \delta_{1,N}^\beta:M^N\rightarrow \mathbb{R}$ and $\hat \delta_{2,N}^\beta:M^N\rightarrow \mathbb{R}$ by
\begin{equation}\begin{gathered}\label{dbdispersionsapprox}
    \hat \delta_{1,N}^\beta(\mathbf{X})=\max_{k=1,\ldots,K}\lVert\log_{\hat m_N(\mathbf{X})}(\hat q_{k,N}(\mathbf{X}))-\log_{\hat m_N(\mathbf{X})}(\hat p_{k,N}(\mathbf{X}))\rVert \\
    \hat \delta_{2,N}^\beta(\mathbf{X})=\frac{1}{K}\sum_{k=1}^K\lVert\log_{\hat m_N(\mathbf{X})}(\hat q_{k,N}(\mathbf{X}))-\log_{\hat m_N}(\hat p_{k,N}(\mathbf{X}))\rVert.
\end{gathered}\end{equation}
The following theorem guarantees a lower bound on the breakdown points for these two estimators of dispersion. It does not require the uniqueness of the sample median or quantiles as it holds for any choice of $\hat m_N$, $\hat q_{k,N}$ and $\hat p_{k,N}$.

\begin{theorem}\label{dbbdp2}
For any $\beta\in[0,1)$, positive integer $N$, Hadamard manifold $M$, sample $(\mathbf{X})\in M^N$, $\hat m_N:M^N\rightarrow M$ and, for each $k=1,\ldots, K$, $\hat q_{k,N}:M^N\rightarrow M$ and $\hat p_{k,N}:M^N\rightarrow M$, the breakdown points of $\hat \delta_{1,N}^\beta$ and $\hat \delta_{2,N}^\beta$ are at least $(\lceil(1-\beta)N/2\rceil)/N$ at every sample.
\end{theorem}

	\section{Extreme quantiles} \label{dbextreme}

The following result seems desirable as it has been shown to be true in the Euclidean case by \cite{Girard2017}. In particular, Theorem \ref{dbext}(c), which roughly states that from the perspective of the data cloud, quantiles end up lying in the direction of $\xi$ as $\beta\rightarrow 1$, seems intuitive.

For any $p\in M$ and $\xi\in \partial M$, define the geodesic $\gamma_p^\xi:\mathbb{R}\rightarrow M$ by $\gamma_p^\xi(t)=\exp_p(t\xi_p)$. The proof of the following theorem is in Appendix \ref{dbproof_ext}.
 
	\begin{theorem} \label{dbext}
        Let $M$ be an $n$-dimensional Hadamard manifold and fix $\xi\in\partial M$. Assume $X$ is an $M$-valued random element whose support is not contained in the image $\gamma_p^\xi(\mathbb{R})$ for any $p\in M$. Take a sequence $\beta_1,\beta_2,\ldots$ in $[0,1)$ that converges to 1, and a sequence $q_1,q_2,\ldots$ in $M$ for which $q_m$ is an element of $q(\beta_m,\xi)$.
        
        (a) The set of minimizers of (\ref{dbexpect}) is non-empty if and only if $\beta<1$.


        (b) As $m\rightarrow\infty$, $d(q_m,x)\rightarrow\infty$ for any $x\in M$,

        (c) As $m\rightarrow\infty$, $\lVert \log_X(q_m)/d(q_m,X)-\xi_X\rVert\xrightarrow{L^2} 0$.
        \end{theorem}

        \begin{remark}
            Theorem \ref{dbext}(b) might at first seem to be a disadvantage because it implies that quantiles need not be contained in the convex hull of the support of $X$ even when the support is compact; this has been noted in the Euclidean case by \cite{Girard2017}. It can be empirically verified that this phenomenon occurs for parameter-based quantiles too, even though we were unable to show the equivalent of Theorem \ref{dbext}(b). This phenomenon is counterintuitive because of our experience with univariate quantiles, but this need not negatively impact the usefulness of quantiles; see Sections 6 and 7 of \cite{Shin2023}, or Sections \ref{dbperm} and \ref{dbmeasures} of the current paper. Indeed, for some applications like outlier detection, it is actually necessary.
        \end{remark}

        \begin{remark}
            When the assumption on the support in the above theorem is violated, the situation becomes equivalent to that of ordinary univariate Euclidean quantiles, and can be dealt with accordingly.
        \end{remark}

        \begin{remark}
            When $M=\mathbb{R}^n$, making the obvious identification between $\partial M$ and $S^{n-1}$, 
            \begin{align*}
                \bigg\lVert\frac{q_m}{\lVert q_m\rVert}-\xi\bigg\rVert&\leq\bigg\lVert\frac{q_m}{\lVert q_m\rVert}-\frac{q_m-X}{\lVert q_m-X\rVert}\bigg\rVert+\bigg\lVert\frac{q_m-X}{\lVert q_m-X\rVert}-\xi\bigg\rVert \\
                &=\bigg\lVert\frac{(\lVert q_m-X\rVert-\lVert q_m\rVert) q_m}{\lVert q_m-X\rVert\lVert q_m\rVert}+\frac{X}{\lVert q_m-X\rVert}\bigg\rVert+\bigg\lVert\frac{q_m-X}{\lVert q_m-X\rVert}-\xi\bigg\rVert \\
                &\leq \frac{\lvert\lVert q_m-X\rVert-\lVert q_m\rVert\rvert/\lVert q_m\rVert}{\lVert q_m/\lVert q_m\rVert-X/\lVert q_m\rVert\rVert}+\frac{\lVert X\rVert/\lVert q_m\rVert}{\lVert q_m/\lVert q_m\rVert-X/\lVert q_m\rVert\rVert}+\bigg\lVert\frac{q_m-X}{\lVert q_m-X\rVert}-\xi\bigg\rVert\\
                &\leq \frac{\lVert X\rVert/\lVert q_m\rVert}{\lVert q_m/\lVert q_m\rVert-X/\lVert q_m\rVert\rVert}+\frac{\lVert X\rVert/\lVert q_m\rVert}{\lVert q_m/\lVert q_m\rVert-X/\lVert q_m\rVert\rVert}+\bigg\lVert\frac{q_m-X}{\lVert q_m-X\rVert}-\xi\bigg\rVert\\
                &\xrightarrow{p} 0
            \end{align*}
            since (b) implies $\lVert q_m\rVert$ converges to $\infty$ and therefore the first two summands in the second to last line converge to 0 surely, while (c) implies the last summand converges to 0 in mean square. But since $q_m/\lVert q_m\rVert-\xi$ is deterministic, the above simply means $q_m/\lVert q_m\rVert\rightarrow \xi$. So Proposition 2.1 and Theorem 2.1 in \cite{Girard2017} correspond to the above results. However, \cite{Girard2017} only deal with non-atomic distributions; no such assumption is required in the above theorem. Moreover, their proofs do not easily generalize to our Hadamard manifold case, so our proofs, apart from that of (b), are quite different.
        \end{remark}

\section{Gradient of the loss function} \label{dbhyp}
	
 We are interested in this gradient because it can be used in a gradient descent algorithm on the basis of Theorem \ref{dbqgrad} to compute quantiles. The $\mathrm{grad}$ and $d(\log_x)_p^\dagger$ notations are as defined in Section \ref{dbasymp}. The proof of the results in this section are in Appendix \ref{dbproof_hyp}. 

 \begin{proposition} \label{dbgrad}
      Let $M$ be an $n$-dimensional Hadamard manifold. Fix $x\in M$, $\beta\in[0,\infty)$ and $\xi\in\partial M$. 
      
      (a) If $p\neq x$, $1-\beta\leq\lVert\nabla \rho(x,p;\beta,\xi)\rVert\leq1+\beta$.

      (b) If $p\neq x$,
      \begin{equation}\begin{aligned} \label{dbequal}
      \nabla \rho(x,p;\beta,\xi)=-\frac{\log_p(x)}{d(p,x)}-d(\log_x)_p^\dagger \beta\xi_x=d(\log_x)_p^\dagger\bigg(\frac{\log_x(p)}{d(p,x)}-\beta\xi_x\bigg).
      \end{aligned}\end{equation}
 \end{proposition}

\begin{remark} \label{dbapprox}
Making the canonical identification between $\partial M$ and $S^{n-1}$ in the Euclidean case, $\nabla \rho(x,p;\beta,\xi)=-(x-p)/\lVert p-x\rVert-\beta\xi=(p-x)/\lVert p-x\rVert-\beta\xi$. So it seems natural that Proposition \ref{dbgrad}(a) should hold on more general Hadamard manifolds. However, Theorem 5.1 in \cite{Shin2023} shows that the size of the gradient of the parameter-based loss function on hyperbolic spaces is unbounded as noted in Remark 5.1 in that paper, marking another advantage of the data-based perspective. The na\"ive analogue on Hadamard manifolds of $(x-p)/\lVert x-p\rVert$ is $\log_p(x)/d(p,x)$, and of $(p-x)/\lVert p-x\rVert-\beta\xi$ is $\log_x(p)/d(p,x)-\beta\xi_x$. Since the gradient at $p$ must be in the tangent space at $p$, we can imagine $d(\log_x)_p^\dagger$, which is the identity in $\mathbb{R}^n$, as shifting vectors from $T_xM$ to $T_pM$, and the parallels to (\ref{dbequal}) can be clearly seen. On the other hand, because the radial field $p\mapsto\xi_p$ is also a function of $p$ and is completely independent of $x$, the gradient of the parameter-based loss function is much more complicated and is not so cleanly analogous as (\ref{dbequal}) is to the Euclidean $(p-x)/\lVert p-x\rVert-\beta\xi$.
\end{remark}

 Recall that the radial field $p\mapsto\xi_p$, though continuously differentiable, is not necessarily smooth, somewhat complicating efforts to take the gradient of the parameter-based loss function as many techniques in Riemannian geometry presuppose smoothness. That is part of why \cite{Shin2023} only calculated the gradients on hyperbolic spaces, in which they showed that smoothness is ensured in Proposition 5.1. Because there is no need to differentiate the radial field with the data-based loss function, the gradient on any Hadamard manifold can be computed if we have the differential of $\log_x$. Here we derive the gradient on Hadamard manifolds that are locally symmetric spaces. Locally symmetric spaces are Riemannian manifolds where the covariant derivative of the Riemannian curvature tensor is zero. All (globally) symmetric spaces are locally symmetric, including hyperbolic spaces, the spaces of symmetric, positive-definite matrices, Euclidean spaces, compact Lie groups, Grassmanians and spheres; the first three examples are also Hadamard manifolds. Of course, for any Hadamard manifold, the gradient could be calculated by expressing the loss function in local coordinates, obtaining the Euclidean gradient $\Psi(x,p;\beta,\xi)$ and left-multiplying it by the Riemannian metric at that point $g_p$; see (\ref{dbequiv}). The results of this section eliminate the need for that tedious process on locally symmetric Hadamard spaces.

In order to calculate the gradient, we need to make use of Jacobi fields. Given an open interval $I\subset \mathbb{R}$ containing 0 and a family of geodesics $\{\gamma_s\}_{s\in I}$ that varies smoothly with respect to $s$, a Jacobi field is a vector field along $\gamma_0$ describing how the family varies at each $t$: $J(t)=(\partial \gamma_s(t)/\partial s)|_{s=0}$. Crucially for our purposes, Jacobi fields can be explicitly calculated on locally symmetric spaces.

Let $M$ be an $n$-dimensional locally symmetric Hadamard space. Take distinct points $p,x\in M$ and $\xi\in \partial M$ and let $e_1=\log_p(x)/d(p,x)\in T_pM$. The curvature operator $CO_{e_1}:T_pM\rightarrow T_pM$ defined by 
\begin{equation}\begin{aligned}\label{dbco}
CO_{e_1}(u)=R(e_1,u)e_1,
\end{aligned}\end{equation}
where $R$ is the Riemannian curvature tensor of $M$, is self-adjoint, and hence has an orthonormal basis of eigenvectors $e_1,\ldots,e_n\in T_pM$, the first element of which is $\log_p(x)/d(p,x)$. Define $\kappa_i$ as the eigenvalue corresponding to $e_i$, that is, $\kappa_i:=\langle CO_{e_1}(e_i),e_i\rangle=\langle R(e_1,e_i)e_1,e_i\rangle$, so that $\kappa_1=0$ and $\kappa_i$ is the sectional curvature of the subspace of $T_pM$ spanned by $e_1$ and $e_i$ for $i=2,\ldots,n$. Define 
\begin{align*}
    g_i(t)=\begin{cases}
        \frac{1}{d(p,x)\sqrt{-\kappa_i}}\sinh(d(p,x)\sqrt{-\kappa_i}t) &\text{if $\kappa_i<0$,} \\
        t &\text{if $\kappa_i=0$,}
    \end{cases}
\end{align*}
for $i=2,\ldots,n$. Denote by $\Gamma_{x\rightarrow p}:T_xM\rightarrow T_pM$ the parallel transport of a vector at $x$ to $p$ along the unique (modulo speed) geodesic connecting the two.
 
	\begin{theorem}\label{dbgradient}
		Let $M$ be an $n$-dimensional locally symmetric Hadamard space. Then for $p\neq x$, 
		\begin{equation}\begin{aligned} \label{dbsymgrad}
		&\nabla \rho(x,p;\beta,\xi)   \\
  &=-\frac{\log_p(x)}{d(p,x)}-\beta\sum_{i=1}^n \frac{(d/dt)g_i(t)|_{t=0}}{g_i(1)}\langle \Gamma_{x\rightarrow p}(\xi_x),e_i\rangle e_i  \\
  &=-\frac{\log_p(x)}{d(p,x)}-\beta\bigg(\bigg\langle\Gamma_{x\rightarrow p}(\xi_x),\frac{\log_p(x)}{d(p,x)}\bigg\rangle\frac{\log_p(x)}{d(p,x)}+\sum_{i=2}^n \frac{(d/dt)g_i(t)|_{t=0}}{g_i(1)}\langle \Gamma_{x\rightarrow p}(\xi_x),e_i\rangle e_i \bigg)
		\end{aligned}\end{equation}
	\end{theorem}

	\begin{remark} \label{dbunbdd}
		\sloppy The gradient of the loss function for parameter-based quantiles is unbounded, at least on hyperbolic spaces; see Theorem 5.1 of \cite{Shin2023}. On the other hand, the norm of the above gradient  is clearly contained between $1-\beta$ and $1+\beta$, as guaranteed by Proposition \ref{dbgrad}(a), since $((d/dt)g_i(t)|_{t=0})/g_i(1)$ is 1 when $\kappa_i=0$ and $d(p,x)\sqrt{-\kappa_i}/\sinh(d(p,x)\sqrt{-\kappa_i})\in(0,1]$ when $\kappa_i<0$.
	\end{remark}

     The calculation in Theorem \ref{dbgradient} is greatly simplified if sectional curvatures are constant as the gradient can be calculated without the need for eigendecomposition of the curvature operator. If all the sectional curvatures equal 0, $M=\mathbb{R}^n$ and each $\kappa_i=0$, so (\ref{dbsymgrad}) becomes $-(x-p)/\lVert x-p\rVert-\beta\xi$, as expected. If $M$ is of constant negative sectional curvature, $M$ is a hyperbolic space; the following corollary deals with this case. Refer to Section 5 of \cite{Shin2023} for relevant background on hyperbolic spaces, including an expression for $\xi_p$ in Proposition 5.1.

     \begin{corollary} \label{dbhypgrad}
Let $M$ be an n-dimensional hyperbolic space of constant sectional curvature $\kappa<0$. Then, for $p \neq x$,
\begin{align*}
		&\nabla \rho(x,p;\beta,\xi)   \\
  &=-\frac{\log_p(x)}{d(p,x)}-\beta\bigg(\frac{d(p,x)\sqrt{-\kappa}}{\sinh(d(p,x)\sqrt{-\kappa})}\Gamma_{x\rightarrow p}(\xi_x) \\
  &\qquad+\bigg(\frac{d(p,x)\sqrt{-\kappa}}{\sinh(d(p,x)\sqrt{-\kappa})}-1\bigg)\bigg\langle\xi_x,\frac{\log_x(p)}{d(p,x)}\bigg\rangle\frac{\log_p(x)}{d(p,x)}\bigg).
		\end{align*}
     \end{corollary}

\begin{remark}\label{dbradial}
If the exponential and inverse exponential maps are known, the radial fields $x\mapsto \xi_x$ can often be calculated explicitly using Proposition 3.3(a) of \cite{Shin2023}, which states that 
\begin{equation}\begin{aligned}\label{dbxi}
\xi_x=\lim_{t\rightarrow \infty}\frac{\log_x(\gamma(t))}{d(p,\gamma(t))},
\end{aligned}\end{equation}
where $\gamma:[0,\infty)\rightarrow M$ is any geodesic ray in the equivalence class $\xi$. Then even if an exlicit expression for $\xi_x$ remains elusive, it can be approximated by choosing a suitably large $t$. This approximation is a simpler affair for data-based quantiles because one needs expressions for $\xi_x$ at only finitely many data points when calculating sample quantiles, rather than at every point in $M$ as for parameter-based quantiles.
\end{remark}

\begin{remark}\label{dbhope}
Due to Proposition \ref{dbgrad} and Remark \ref{dbapprox}, we hypothesize that 
 \begin{equation}\begin{aligned}\label{dbapproxgrad}
     -\frac{\log_p(x)}{d(p,x)}-\Gamma_{x\rightarrow p}(\beta\xi_x)
 \end{aligned}\end{equation}
 approximates the gradient well enough to be used in a descent algorithm. If this is so, we can find the quantile without the need to calculate tedious differentials or, in the case of locally symmetric spaces, computationally intensive eigendecompositions. Instead, basic operations like the exponential maps, their inverses and parallel transport, as well as radial fields which can be calculated or approximated as per Remark \ref{dbradial}, suffice. However, we also speculate that the above approximation should get worse as curvature increases in absolute value and $d(p,x)$ increases. Theorem \ref{dbgradient} supports this as $((d/dt)g_i(t)|_{t=0})/g_i(1)=d(p,x)\sqrt{-\kappa_i}/\sinh(d(p,x)\sqrt{-\kappa_i})$ when $\kappa_i<0$, and $t/\sinh(t)$ is strictly decreasing in $t$ while $\lim_{t\rightarrow 0}t/\sinh(t)=1$, $\lim_{t\rightarrow\infty}t/\sinh(t)=0$. Therefore on a given Hadamard manifold, we expect a descent algorithm which uses (\ref{dbapproxgrad}) to fare worse the more spread out the data are, and also the larger $\beta$ is thanks to Theorem \ref{dbext}(b).
 
 On the other hand, because the gradient of the parameter-based loss function can be unbounded per Remark \ref{dbunbdd}, it seems that neither (\ref{dbapproxgrad}) nor something like 
 \begin{equation}\begin{aligned} \label{dbbadgrad}
     -\frac{\log_p(x)}{d(p,x)}-\beta\xi_p
 \end{aligned}\end{equation}
 should be a good approximation. These ideas will be tested in Section \ref{dbsims}.
 \end{remark}

Here we will explain some practical details about the implementation of the descent algorithm, the pseudocode of which is contained in Algorithm \ref{dbalgo}. The $\text{step}$ in lines 4 and 8 is either the true gradient or some approximation of it such as (\ref{dbapproxgrad}) or (\ref{dbbadgrad}). When $p=x_i$, $\log_p(x_i)/d(p,x_i)$ is not defined, so we set it at 0 and let $\text{step}(x_i,p)$ be $-\beta\xi_p=-\Gamma_{x\rightarrow p}(\beta\xi_x)$. In fact, we also do this when $d(p,x_i)$ is extremely small as calculating $\log_p(x_i)/d(p,x_i)$ in that case can cause problems. In line 6, $\text{shift}$ is divided by its norm because we have observed that not doing so sometimes causes issues with convergence of the algorithm when $p$ is near a singularity. 
The size of the update is therefore controlled entirely by the learning rate $\text{lr}$ in this algorithm.

The loss function in line 7 can be either the parameter-based or data-based loss function. If the comparison in line 7 is true, the learning rate $\text{lr}$ is increased in line 8 to speed up convergence, and if it is false, we deem $\text{lr}$ to be too large and lower it. The counter is necessary because occasionally the learning rate $\text{lr}$ oscillates between bounds that are small but greater than the tolerance $\text{tol}$; in lines 5 and 10, we have set up the algorithm to terminate if this happens too often.
 
\begin{algorithm}[h] 
	\caption{Descent algorithm} \label{dbalgo}
	\begin{algorithmic}[1]
		\State Input: $x_1,\ldots,x_N\in M$, $\beta\in[0,1)$, $\xi\in\partial M$, tolerance $\text{tol}>0$ and $\text{maxcount}\in\mathbb{Z}^+$.
		\State Output: $q\in M$
		\State Initialize $p$, learning rate $\text{lr}>\text{tol}$ and counter $\text{count}\leftarrow 0$.
            \State $\text{shift}\leftarrow (1/N)\sum\text{step}(x_i,p)$
            \While{$\text{lr}>\text{tol}$ \& $\text{count}<\text{maxcount}$}
            \State $p_{\text{new}} \leftarrow \mathrm{exp}_p(-\text{lr}\cdot\text{shift}/\lVert\text{shift}\rVert)$
		\If {$(1/N)\sum\text{loss}(x_i,p_{\text{new}};\beta,\xi)\leq(1/N)\sum\text{loss}(x_i,p;\beta,\xi)$}
		\State $p\leftarrow p_{\text{new}}$, $\text{lr}\leftarrow 1.1\cdot\text{lr}$, $\text{shift}\leftarrow (1/N)\sum\text{step}(x_i,p)$
		\Else
		\State $\text{lr}\leftarrow \text{lr}/2$, $\text{count}\leftarrow \text{count}+1$
		\EndIf
		\EndWhile
            \State $q\leftarrow p$
	\end{algorithmic}
\end{algorithm}

 \section{Experiments}\label{dbexps}
	
\sloppy The code for the experiments in this section, implemented in Python with PyTorch, can be found at \url{https://github.com/hayoungshin1/Data-based-quantiles-on-Hadamard-spaces}.

 \subsection{Simulations in hyperbolic space} \label{dbsims}

Let $M$ be $2$-dimensional hyperbolic space of constant sectional curvature $-1$. We generate 100 data points from the bivariate normal distribution $\mathcal{N}(0,0.3I_2)$ truncated to ensure that each data point is contained in the standard open unit ball in $\mathbb{R}^2$. These data points are treated as points in the Poincar\'e ball model of hyperbolic space and projected onto the hyperboloid model. In this section, we compute the quantiles with Algorithm \ref{dbalgo} using true and approximate gradients, compare the performances of data- and parameter-based quantiles in this respect, and illustrate one of the potential applications of quantiles mentioned in the introduction---determining whether two data sets are generated by the same distribution.

We compute the parameter-based and data-based quantiles with Algorithm \ref{dbalgo} using the true gradient, (\ref{dbapproxgrad}) and (\ref{dbbadgrad}); the visualization of these results in the Poincar\'e ball is contained in Figure \ref{dbfig:imageone}. The values of $\beta$ we consider are $0,0.2,0.4,0.6,0.8$ and $0.98$, and making the natural identification between the boundary of the Poincar\'e ball $S^1$ and the boundary at infinity, the values of $\xi$ we consider are 
\begin{equation}\label{dbxil}
\xi_{l,L}=(\cos(2\pi l/L),\sin(2\pi l/L))
\end{equation}
for $l=1,\ldots,L$, where $L=64$.

\begin{figure}[!t]
	\centering
	\begin{subfigure}[t]{0.32\linewidth}
		\includegraphics[width=\linewidth]{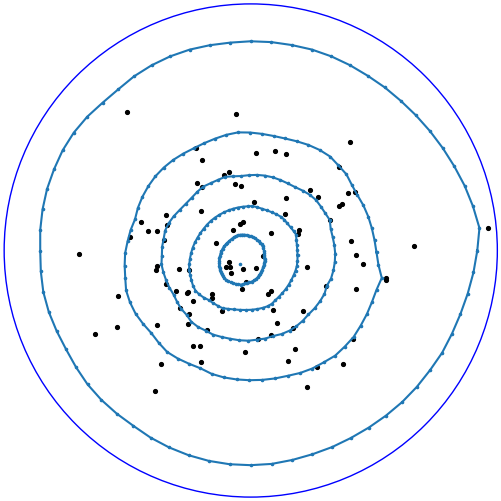}
		\caption{Parameter-based quantiles calculated with the true gradient.}
		\label{dbfig:dq1}
	\end{subfigure}
	\begin{subfigure}[t]{0.32\linewidth}
		\includegraphics[width=\linewidth]{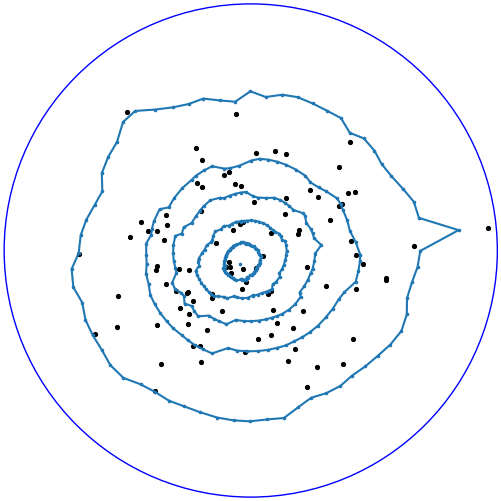}
		\caption{Parameter-based quantiles calculated with (\ref{dbapproxgrad}).}
		\label{dbfig:dq2}
	\end{subfigure}
	\begin{subfigure}[t]{0.32\linewidth}
		\includegraphics[width=\linewidth]{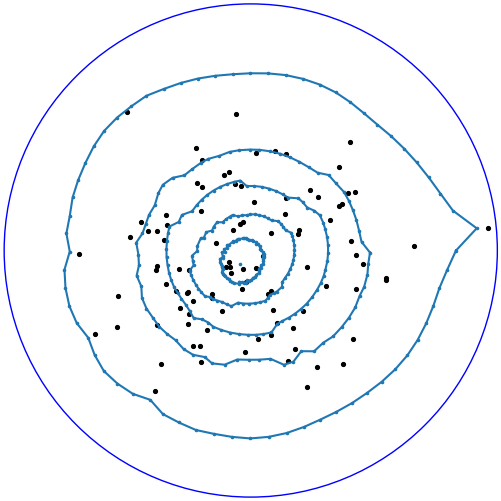}
		\caption{Parameter-based quantiles calculated with (\ref{dbbadgrad}).}
		\label{dbfig:dq3}
	\end{subfigure}

	\begin{subfigure}[t]{0.32\linewidth}
	    \includegraphics[width=\linewidth]{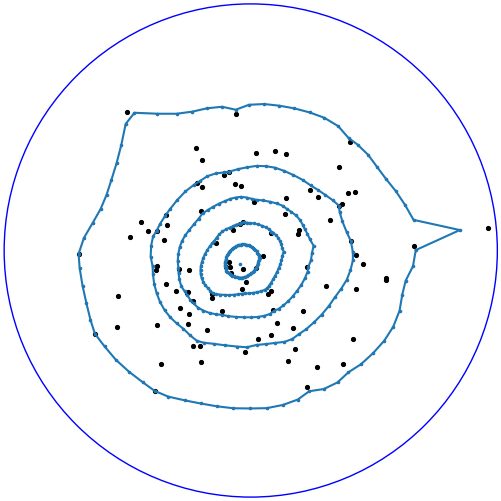}
	    \caption{Data-based quantiles calculated with the true gradient.}
	    \label{dbfig:dq4}
    \end{subfigure}
	    \begin{subfigure}[t]{0.32\linewidth}
		\includegraphics[width=\linewidth]{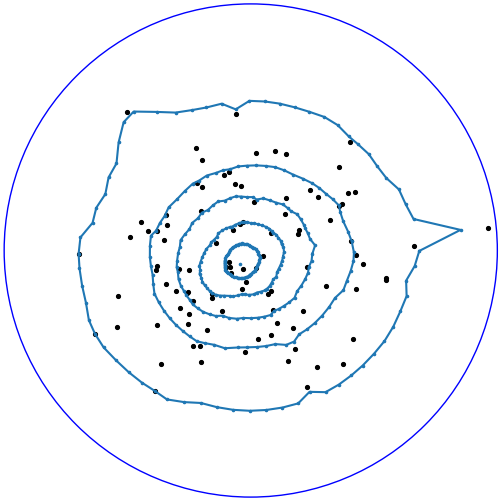}
		\caption{Data-based quantiles calculated with (\ref{dbapproxgrad}).}
		\label{dbfig:dq5}
	\end{subfigure}
	\begin{subfigure}[t]{0.32\linewidth}
		\includegraphics[width=\linewidth]{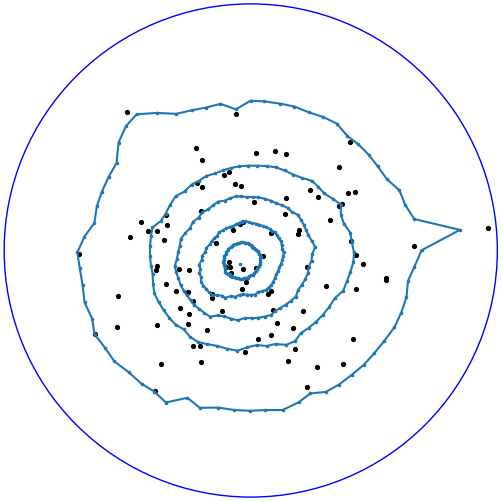}
		\caption{Data-based quantiles calculated with (\ref{dbbadgrad}).}
		\label{dbfig:dq6}
	\end{subfigure}
	\caption{The black points are the data points from the first simulated data set, the smaller blue points the computed quantiles, and the blue curves connecting them the isoquantile curves for $\beta\in\{0.2,0.4,0.6,0.8,0.98\}$. The top row shows attempts at computing parameter-based quantiles and the bottom row data-based quantiles. The results in the first column were computed using the true gradients, the second column using (\ref{dbapproxgrad}) and the third column (\ref{dbbadgrad}).}
	\label{dbfig:imageone}
\end{figure}

Focusing on Figures \ref{dbfig:dq1} and \ref{dbfig:dq4}, which show the results of using true gradient descent, the parameter-based quantiles seem to spread out more quickly as $\beta$ increases and have smoother, more spherically symmetric isoquantile contours. In contrast, the data-based quantiles may better adhere to the shape of the data cloud; we will investigate this further. 

The other four subfigures represent approximations of the results in Figures \ref{dbfig:dq1} and \ref{dbfig:dq4}. Figures \ref{dbfig:dq4} and \ref{dbfig:dq5} are almost identical, providing evidence for our hunch in Remark \ref{dbhope}. This is true despite the fact that the data are quite spread out, and even when $\beta$ is extremely large; for a given manifold, these are precisely the circumstances under which we worried in Remark \ref{dbhope} that the approximate gradient might fare poorly.

It is clear visually that the approximation in Figure \ref{dbfig:dq5} is more successful than the other three. Both approximations in the parameter-based case fare quite poorly, and interestingly Figure \ref{dbfig:dq2} is actually quite similar to Figure \ref{dbfig:dq5}, despite different the loss functions being minimized; this makes sense because the updates in both cases are computed using the same formula, (\ref{dbapproxgrad}), while the loss function only comes into play in line 7 of the algorithm.

Table \ref{dbtable1} makes the comparisons between the four approximations more precise. When $\beta=0$, the gradients and their approximations are all equivalent, so we only consider the other values of $\beta$. The approximation of the data-based quantiles using (\ref{dbapproxgrad}) has by far the smallest mean error for each $\beta$; in fact, even the best approximation for the parameter-based quantiles when $\beta=0.2$ is barely better than the approximation for the data-based quantiles using (\ref{dbapproxgrad}) when $\beta=0.98$. We also see that in all cases (\ref{dbbadgrad}) leads to better approximations of the parameter-based quantiles than (\ref{dbapproxgrad}). This makes sense as there is no reason to expect that (\ref{dbapproxgrad}) would better approximate the true parameter-based gradient than (\ref{dbbadgrad}), or the reverse in the data-based case; those two cases were included for the sake of completeness. Table \ref{dbtable1} also strongly supports the idea, mentioned in Remark \ref{dbhope}, that approximations of data-based quantiles using (\ref{dbapproxgrad}) should get worse as $\beta$ increases; this trend is seen for all four of the approximations.

\begin{table}[!h]
    \centering
    \caption{For given $\beta$, means over 64 values of $\xi_{l,L}$ of the distance between the quantiles of the first simulated data set computed using the true gradient and each of (\ref{dbapproxgrad}) and (\ref{dbbadgrad}). The smallest value in each column is highlighted.}
    {\small
        \begin{tabular}{|c|c|c|c|c|c|c|}\cline{3-7}
            \multicolumn{2}{c|}{\multirow{2}{*}{}} & \multicolumn{5}{c|}{$\beta$} \\\cline{3-7}
            \multicolumn{2}{c|}{} & $0.2$ & $0.4$ & $0.6$ & $0.8$ & $0.98$ \\\hhline {|=|=|=|=|=|=|=|}
            \multirow{2}*{Parameter-based} & (\ref{dbapproxgrad}) & 0.0511 & 0.1031 & 0.1502 & 0.2615 & 0.9801\\\cline{2-7}
            & (\ref{dbbadgrad}) & 0.0207 & 0.0451 & 0.0686 & 0.1516 & 0.6556 \\ \hhline {|=|=|=|=|=|=|=|}
            \multirow{2}*{Data-based} & (\ref{dbapproxgrad}) & \textbf{0.0044} & \textbf{0.0064} & \textbf{0.0090} & \textbf{0.0122} & \textbf{0.0233} \\\cline{2-7}
            & (\ref{dbbadgrad}) & 0.0181 & 0.0132 & 0.0144 & 0.0189 & 0.0317 \\\hline
        \end{tabular}
    }
    \label{dbtable1}
\end{table}

We next took our simulated data set in the Poincar\'e ball and created a second data set by simply dividing the second coordinate of each point by 4. We can make essentially all of the same observations from these results, contained in Figure \ref{dbfig:imagetwo} and Table \ref{dbtable2}, as we did above from those of the first data set. In particular, the data-based isoquantile contours clearly conform to the shape of the data cloud better than the parameter-based ones, and they do so extremely well for smaller values of $\beta$. It is known in the Euclidean case that extreme geometric quantiles do not follow the shape very well, and this is certainly true for the parameter-based quantiles in Figure \ref{dbfig:dq7}; these issues are also apparent in the outermost contour of Figure \ref{dbfig:dq10}, but they are far less pronounced, perhaps indicating that data-based quantiles are more robust to these distortions than parameter-based quantiles, at least on hyperbolic spaces.

Because of the way that this data set was defined from the first, Table \ref{dbtable2} also supports the idea contained in Remark \ref{dbhope} that the approximations should have more trouble the more spread out the data are, and in fact the value in each cell of Table \ref{dbtable1} is larger than the value in the corresponding cell of Table \ref{dbtable2}.

\begin{figure}[!t]
	\centering
	\begin{subfigure}[t]{0.32\linewidth}
		\includegraphics[width=\linewidth]{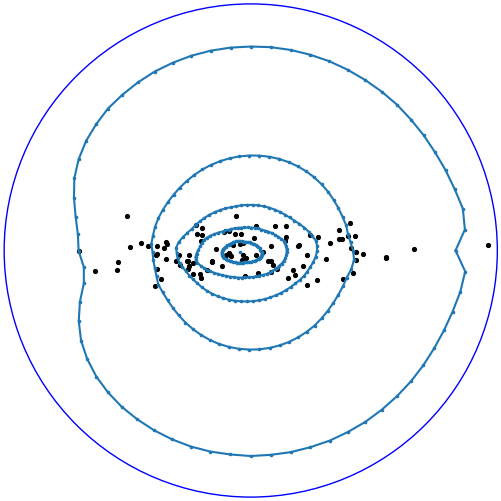}
		\caption{Parameter-based quantiles calculated with the true gradient.}
		\label{dbfig:dq7}
	\end{subfigure}
	\begin{subfigure}[t]{0.32\linewidth}
		\includegraphics[width=\linewidth]{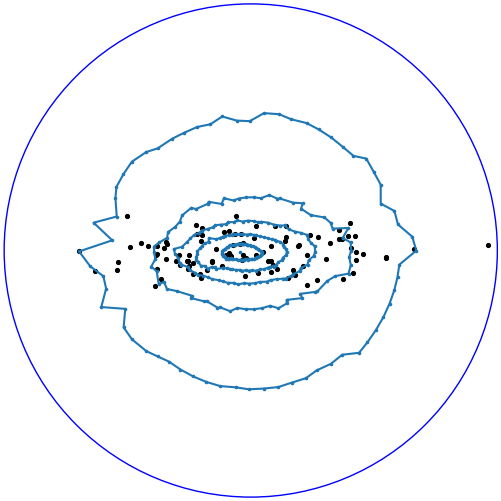}
		\caption{Parameter-based quantiles calculated with (\ref{dbapproxgrad}).}
		\label{dbfig:dq8}
	\end{subfigure}
	\begin{subfigure}[t]{0.32\linewidth}
		\includegraphics[width=\linewidth]{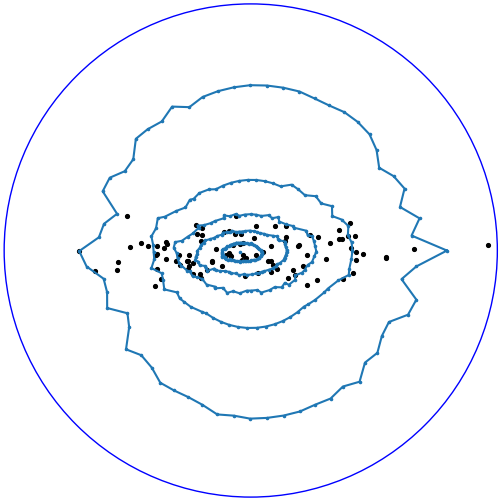}
		\caption{Parameter-based quantiles calculated with (\ref{dbbadgrad}).}
		\label{dbfig:dq9}
	\end{subfigure}

	\begin{subfigure}[t]{0.32\linewidth}
	    \includegraphics[width=\linewidth]{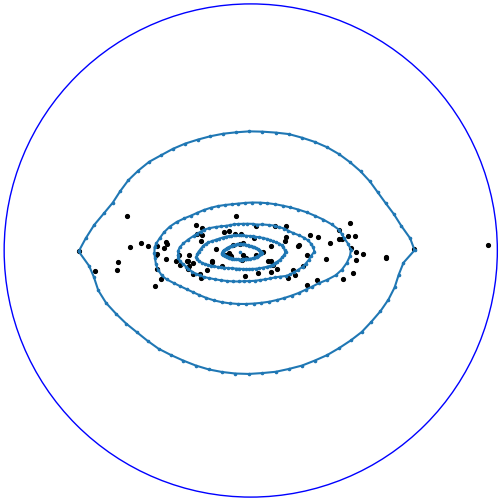}
	    \caption{Data-based quantiles calculated with the true gradient.}
	    \label{dbfig:dq10}
    \end{subfigure}
	    \begin{subfigure}[t]{0.32\linewidth}
		\includegraphics[width=\linewidth]{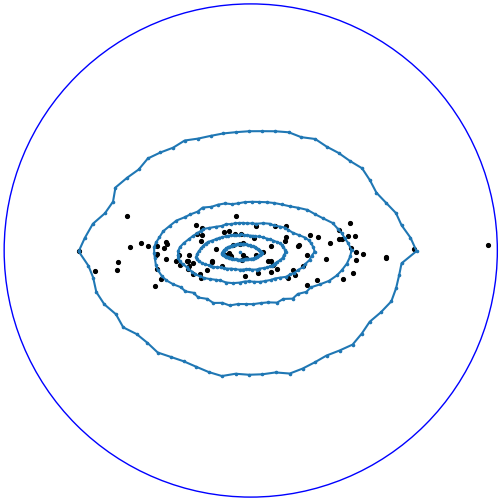}
		\caption{Data-based quantiles calculated with (\ref{dbapproxgrad}).}
		\label{dbfig:dq11}
	\end{subfigure}
	\begin{subfigure}[t]{0.32\linewidth}
		\includegraphics[width=\linewidth]{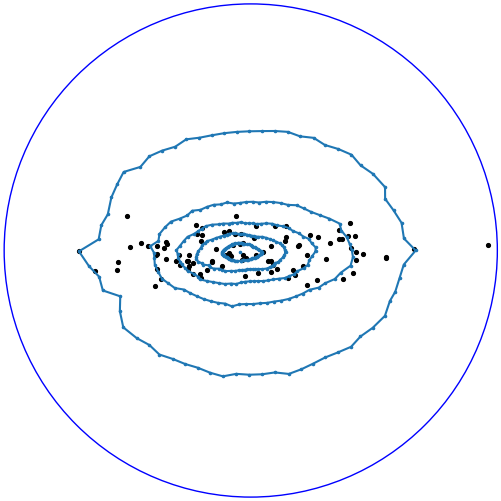}
		\caption{Data-based quantiles calculated with (\ref{dbbadgrad}).}
		\label{dbfig:dq12}
	\end{subfigure}
	\caption{The black points are the data points from the second simulated data set, the smaller blue points the computed quantiles, and the blue curves connecting them the isoquantile curves for $\beta\in\{0.2,0.4,0.6,0.8,0.98\}$. The top row shows attempted at computing parameter-based quantiles and the bottom row data-based quantiles. The results in the first column were computed using the true gradients, the second column using (\ref{dbapproxgrad}) and the third column (\ref{dbbadgrad}).}
	\label{dbfig:imagetwo}
\end{figure}

\begin{table}[!h]
    \centering
    \caption{For given $\beta$, means over 64 values of $\xi_{l,L}$ of the distance between the quantiles of the second simulated data set computed using the true gradient and each of (\ref{dbapproxgrad}) and (\ref{dbbadgrad}). The smallest value in each column is highlighted.}
    {\small
        \begin{tabular}{|c|c|c|c|c|c|c|}\cline{3-7}
            \multicolumn{2}{c|}{\multirow{2}{*}{}} & \multicolumn{5}{c|}{$\beta$} \\\cline{3-7}
            \multicolumn{2}{c|}{} & $0.2$ & $0.4$ & $0.6$ & $0.8$ & $0.98$ \\\hhline {|=|=|=|=|=|=|=|}
            \multirow{2}*{Parameter-based} & (\ref{dbapproxgrad}) & 0.0180 & 0.0396 & 0.0892 & 0.2214 & 0.9822 \\\cline{2-7}
            & (\ref{dbbadgrad}) & 0.0088 & 0.0229 & 0.0496 & 0.1319 & 0.7496 \\ \hhline {|=|=|=|=|=|=|=|}
            \multirow{2}*{Data-based} & (\ref{dbapproxgrad}) & \textbf{0.0018} & \textbf{0.0045} & \textbf{0.0066} & \textbf{0.0109} & \textbf{0.0210} \\\cline{2-7}
            & (\ref{dbbadgrad}) & 0.0085 & 0.0127 & 0.0208 & 0.0318 & 0.0520 \\\hline
        \end{tabular}
    }
    \label{dbtable2}
\end{table}


\subsubsection{Permutation tests with quantiles}\label{dbperm}

Regardless of how closely isoquantile contours follow the shape of a distribution, geometric quantiles give information about a distribution, and \cite{Koltchinskii1997} showed that in the Euclidean case, they completely characterize it. Imagining a measure of central tendency like the mean or median as giving a very rough, low-resolution image of a distribution, quantiles allow for arbitrarily fine resolutions determined by the number of quantiles used. Thus quantiles are useful in distinguishing distributions when the mean or median is not enough because they give more points for comparison.

Permutation tests are exact tests for the null hypothesis $H_0:F_X=F_Y$, where $F_X$ and $F_Y$ are the distributions underlying two data sets of size $N$ and $N'$, respectively. A common test statistic in permutation tests is $\lvert \hat m_N^X-\hat m_{N'}^Y\rvert$ in the case of real data, where $\hat m_N^X$ and $\hat m_{N'}^Y$ measure in some way the central tendency of the data sets; for simplicity, we will consider the sample medians. This statistic generalizes to $T_0=d(\hat m_N^X, \hat m_{N'}^Y)$ for metric space-valued data. The null hypothesis is rejected when $T_0$ is large. However, $T_0$ will have trouble rejecting $H_0$ even when $F_X\neq F_Y$ if the medians of the two distributions are equal, so we need more points of comparison. Letting $M$ be an Hadamard space and $(\beta_1,\xi_1),\ldots,(\beta_K,\xi_K)\in[0,1)\times \partial M$, we propose the statistic
\begin{equation*}
T_1=\sum_{k=1}^Kd(\hat{q}_N^X(\beta_k,\xi_k),\hat{q}_{N'}^Y(\beta_k,\xi_k)),
\end{equation*}
where $\hat{q}_N^X(\beta_k,\xi_k)$ and $\hat{q}_{N'}^Y(\beta_k,\xi_k)$ are sample $(\beta_k,\xi_k)$-quantiles, $k=1,\ldots,K$, for the respective data sets. Continuing the analogy from the previous paragraph, $T_1$ allows for comparisons at higher resolutions than $T_0$, and it should be able to distinguish distributions that differ at any of the $K$ quantiles, including at the median if any of the $\beta_k$'s is 0. As for the choice of quantile indices, in consideration of computational efficiency, they should be spaced apart to minimize redudancy.

We illustrate our point by testing $H_0$ with the data sets visualized in Figures \ref{dbfig:imageone} and \ref{dbfig:imagetwo}, $N=N'=100$. We use $K=5$ quantiles, indexed by $(0,\xi_{1,4}),(0.8,\xi_{1,4}),\ldots,(0.8,\xi_{4,4})$, where $\xi_{l,4}$, $l=1,\ldots,4$, are defined as in (\ref{dbxil}) for $L=4$. Note the first quantile is the median. This results in $p$-values of $0.432$ using $T_0$ and $0.01$ using $T_1$, calculated using 500 permutations. Thus $T_1$ succeeds in detecting the difference in the underlying distributions of these two data sets, whereas $T_0$ is an abject failure.

 \subsection{Real data experiments with diffusion tensor imaging} \label{dbreal}

Diffusion tensor imaging (DTI), first proposed by \cite{Basser1994}, allows measurement of the diffusion of water molecules in a voxel of a magnetic resonance imaging (MRI) scan of biological tissue. It has been crucial in identifying the structures and properties of white matter tracts in the brain. In DTI, the diffusion is modeled as a $3\times 3$ symmetric positive-definite matrix. In this section, we will first calculate and visualize data-based quantiles for DTI data, and then give another demonstration of the usefulness of geometric quantiles by measuring distributional characteristics.

 Denote the space of real symmetric $m\times m$ matrices by $\mathcal{S}_m$ and the space of real symmetric positive-definite (SPD) $m\times m$ matrices by $\mathcal{P}_m$. The former is an $m(m+1)/2$-dimensional vector space, and the latter can be considered a $m(m+1)/2$-dimensional Riemannian manifold on which the tangent space at each point is isomorphic to $\mathcal{S}_m$ and the Riemannian metric at $x\in \mathcal{P}_m$ is defined by the so-called trace, or affine invariant, metric: $\langle v_1,v_2\rangle=\text{tr}(x^{-1}v_1x^{-1}v_2)$, where $v_1,v_2\in T_x\mathcal{P}_m\cong\mathcal{S}_m$. In fact, this Riemannian manifold is complete and simply connected with sectional curvatures in $[-1/2,0]$ (see Proposition I.1 of \cite{Criscitiello2020}); therefore this is an Hadamard manifold. It is also a symmetric space (see Proposition 3.1(c) of \cite{Dolcetti2018}). Beyond diffusion tensor imaging, $\mathcal{P}_m$ is crucial to study because covariance (and precision) matrices are SPD matrices, and these can be random objects in their own right either as sample covariance matrices or as parameters in a Bayesian framework, in which they are often assigned the inverse-Wishart prior distribution; see \cite{Lee2018} for an example. For more details on this Riemannian manifold, such as expressions for the exponential maps and their inverses, parallel transport and the radial fields, consult \cite{Shin2024}. 

With eigendecomposition, an element of $\mathcal{P}_3$ can be visualized as an ellipsoid in three-dimensional space; the orientation of its mutually orthogonal axes is defined by an orthonormal basis of eigenvectors and their lengths by the eigenvalues. Coloring based on so-called fractional anisotropy, a positive real function of the eigenvalues, can also aid in visualization by encoding information about the ellipsoid; see \cite{Pajevic1999}. Briefly, red, green and blue represent left-right, anterior-posterior and superior-inferior orientations, respectively, while brightness encodes the eccentricity of the ellipsoid: darker hues are closer to spheres, while brighter hues are more elongated. 

Our data are visualized in this manner in Figure \ref{dbcc}. Both subfigures display the same data, the only difference being the background; both are included because some hues are more clearly seen against black and others against white. The data set, consisting of 90 data points, is extracted from an axial slice of the splenium of the corpus callosum of a subject. The data are available in the DIPY library for Python; see \url{https://docs.dipy.org/stable/examples_built/reconstruction/reconst_dti.html} for instructions on accessing the raw data and fitting them to the DTI model.

\begin{figure}[!t]
	\centering
	\begin{subfigure}[t]{0.49\linewidth}
		\includegraphics[width=\linewidth]{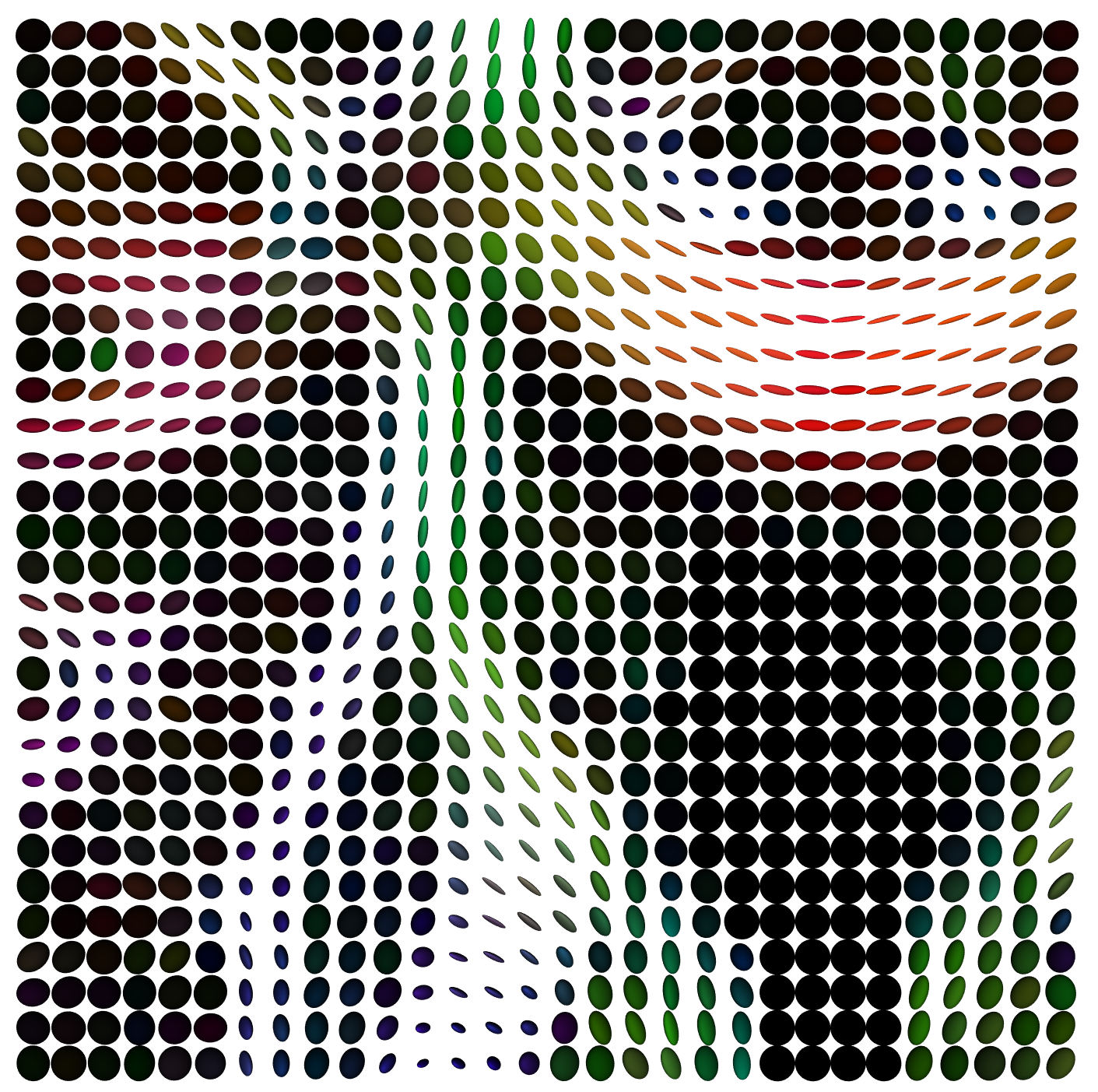}
  \caption{Data against a white background.}
  \label{dbcca}
	\end{subfigure}
	\begin{subfigure}[t]{0.49\linewidth}
		\includegraphics[width=\linewidth]{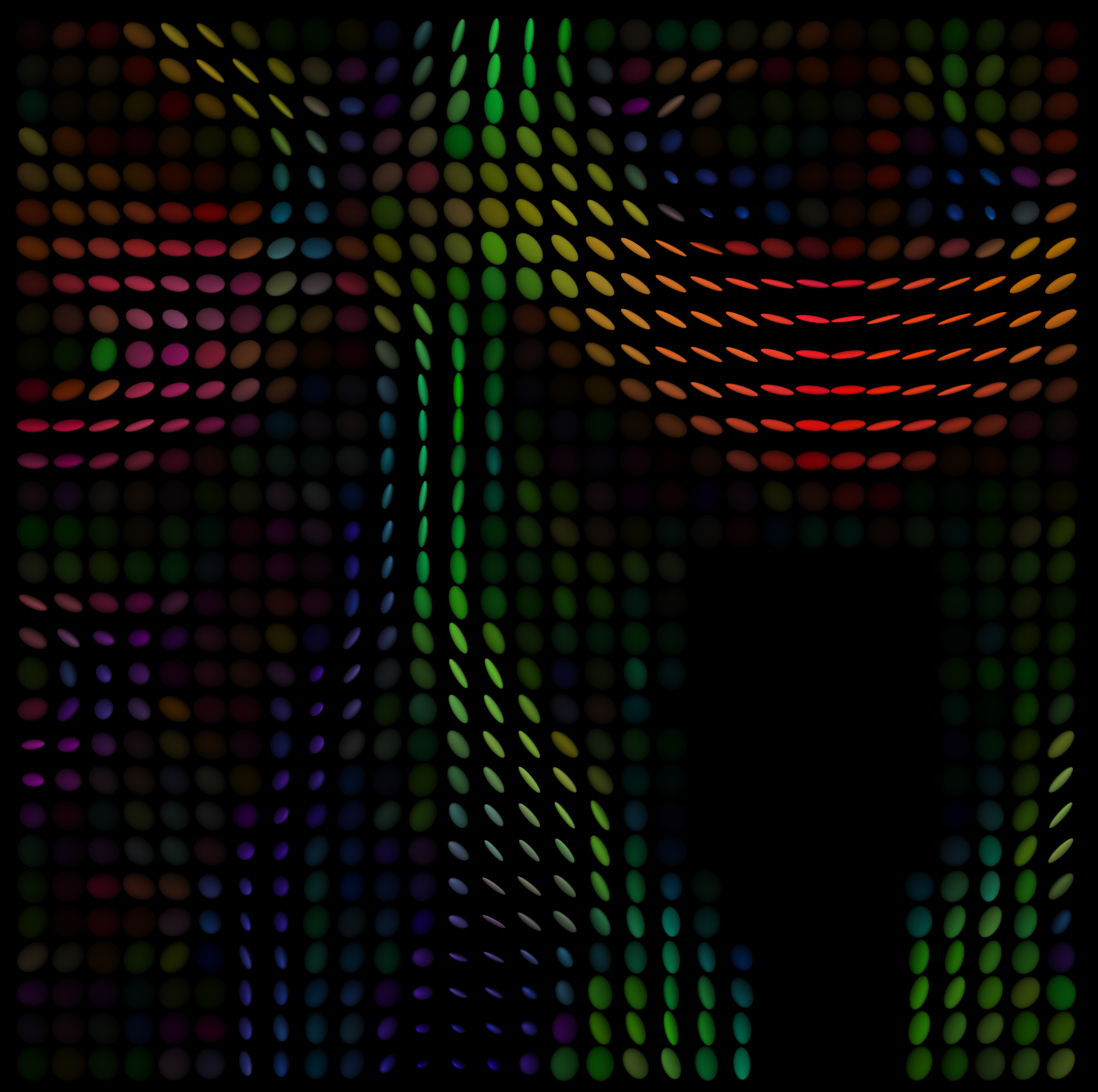}
  \caption{The same data against a black background.}
  \label{dbccb}
	\end{subfigure}
        \caption{Corpus callosum DTI data, visualized as ellipsoids in three-dimensional space.}
        \label{dbcc}
\end{figure}

 We again calculate quantiles for $\beta\in\{0,0.2,0.4,0.6,0.8,0.98\}$, and identifying $\xi\in\partial\mathcal{P}_3$ with $\xi_I\in T_I\mathcal{P}_3\cong\mathcal{S}_3$, we use 8 different values of $\xi$:
		\begin{align*}
			&\xi_1=\frac{1}{\sqrt{3}}\begin{bmatrix}1 & 0 & 0 \\ 0 & 1 & 0 \\ 0 & 0 & 1\end{bmatrix}, \xi_2=\frac{1}{\sqrt{2}}\begin{bmatrix}0 & 1 & 0 \\ 1 & 0 & 0 \\ 0 & 0 & 0\end{bmatrix},\xi_3=\frac{1}{\sqrt{2}}\begin{bmatrix}0 & 0 & 1 \\ 0 & 0 & 0 \\ 1 & 0 & 0\end{bmatrix}, \xi_4=\frac{1}{\sqrt{2}}\begin{bmatrix}0 & 0 & 0 \\ 0 & 0 & 1 \\ 0 & 1 & 0\end{bmatrix},
		\end{align*}
		$\xi_5=-\xi_1,\xi_6=-\xi_2,\xi_7=-\xi_3$, and $\xi_8=-\xi_4$; $\xi_1, \xi_2, \xi_3$, and $\xi_4$ are mutually orthonormal in $T_I\mathcal{P}_3$. 

  The curvature operator in (\ref{dbco}) is a self-adjoint linear operator from $T_p\mathcal{P}_m$ to itself, and so can be expressed as a real symmetric $m(m+1)/2\times m(m+1)/2$ matrix with respect to some orthonormal basis $\{w_1,\ldots,w_{m(m+1)/2}\}$ of $T_p\mathcal{P}_m$. This matrix can be obtained since an expression for the Riemannian curvature tensor of $\mathcal{P}_m$ is known from Proposition 3.1.2 \cite{Dolcetti2014} or Proposition 1.3(b) of \cite{Dolcetti2018}: the $(k,l)$-th entry is 
\begin{align*}
\langle CO_{\log_p(x)/d(p,x)}(w_k),w_l\rangle&=\bigg\langle R\bigg(\frac{\log_p(x)}{d(p,x)},w_k\bigg)\frac{\log_p(x)}{d(p,x)},w_l\bigg\rangle \\
&=\frac{1}{4}\text{Tr}\bigg(\bigg[p^{-1}\frac{\log_p(x)}{d(p,x)},p^{-1}w_k\bigg]\bigg[p^{-1}\frac{\log_p(x)}{d(p,x)},p^{-1}w_l\bigg]\bigg),
\end{align*}
where $[\cdot,\cdot]$ represents the matrix commutator. Then the eigendecomposition of this value will provide the $\kappa_i$s and $e_i$s in (\ref{dbsymgrad}), which we can use because $\mathcal{P}_m$ is a symmetric space. However, it seems to us that (\ref{dbapproxgrad}) should be an even better approximation of the gradient on $\mathcal{P}_m$ because it is ``flatter'' than hyperbolic space in the sense that its sectional curvatures are contained in $[-1/2,0]$ and these bounds are tight, and therefore some of its sectional curvatures are 0. Thus we will use this approximation in our descent algorithm.

The results are shown in Figure \ref{dbwholequantiles}. Each of the four subfigures shows the same configuration of ellipsoids. Figures \ref{dbwholea} and \ref{dbwholeb} differ only in background color for the same reason as in Figure \ref{dbcc}, while Figures \ref{dbwholec} and \ref{dbwholed} show the configuration from different angles; because the ellipsoids are three-dimensional, this is helpful in getting a better sense of their shapes. The central ellipsoid is the median ($\beta=0$), and in each of eight directions there lie five more ellipsoids. In Figures \ref{dbwholea} and \ref{dbwholeb}, starting from the top and moving counterclockwise, the ellipsoids in a given direction represent quantiles for $\xi_k$, $k=1,\ldots,8$ respectively, and moving outward from the center for a given $\xi_k$, the ellipsoids represent quantiles for $\beta=0.2,0.4,0.6,0.8,0.98$ respectively. 

 \begin{figure}[!t]
	\centering
	\begin{subfigure}[t]{0.49\linewidth}
		\includegraphics[width=\linewidth]{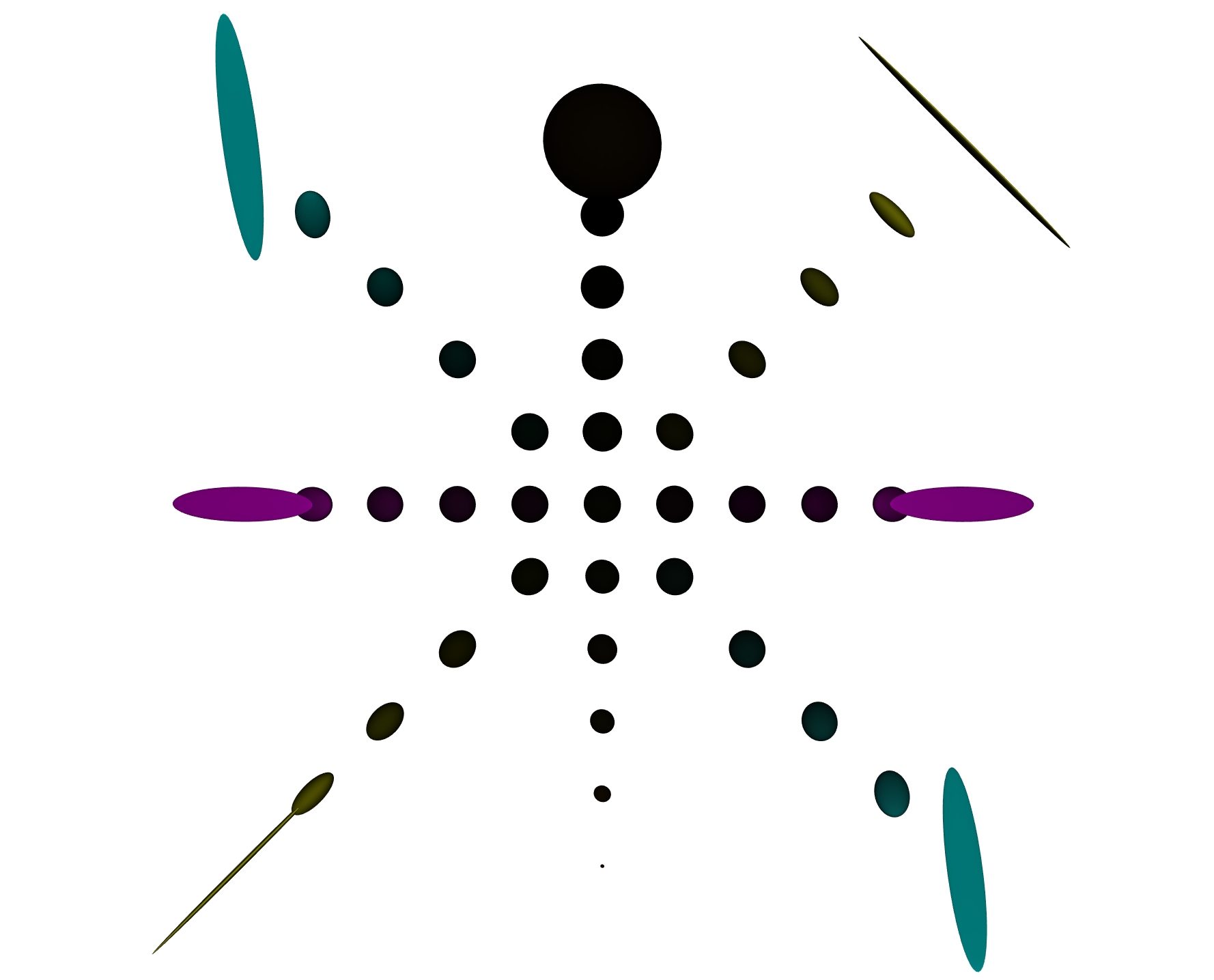}
  \caption{Quantiles arranged in the described configuration.}
  \label{dbwholea}
	\end{subfigure}
	\begin{subfigure}[t]{0.49\linewidth}
		\includegraphics[width=\linewidth]{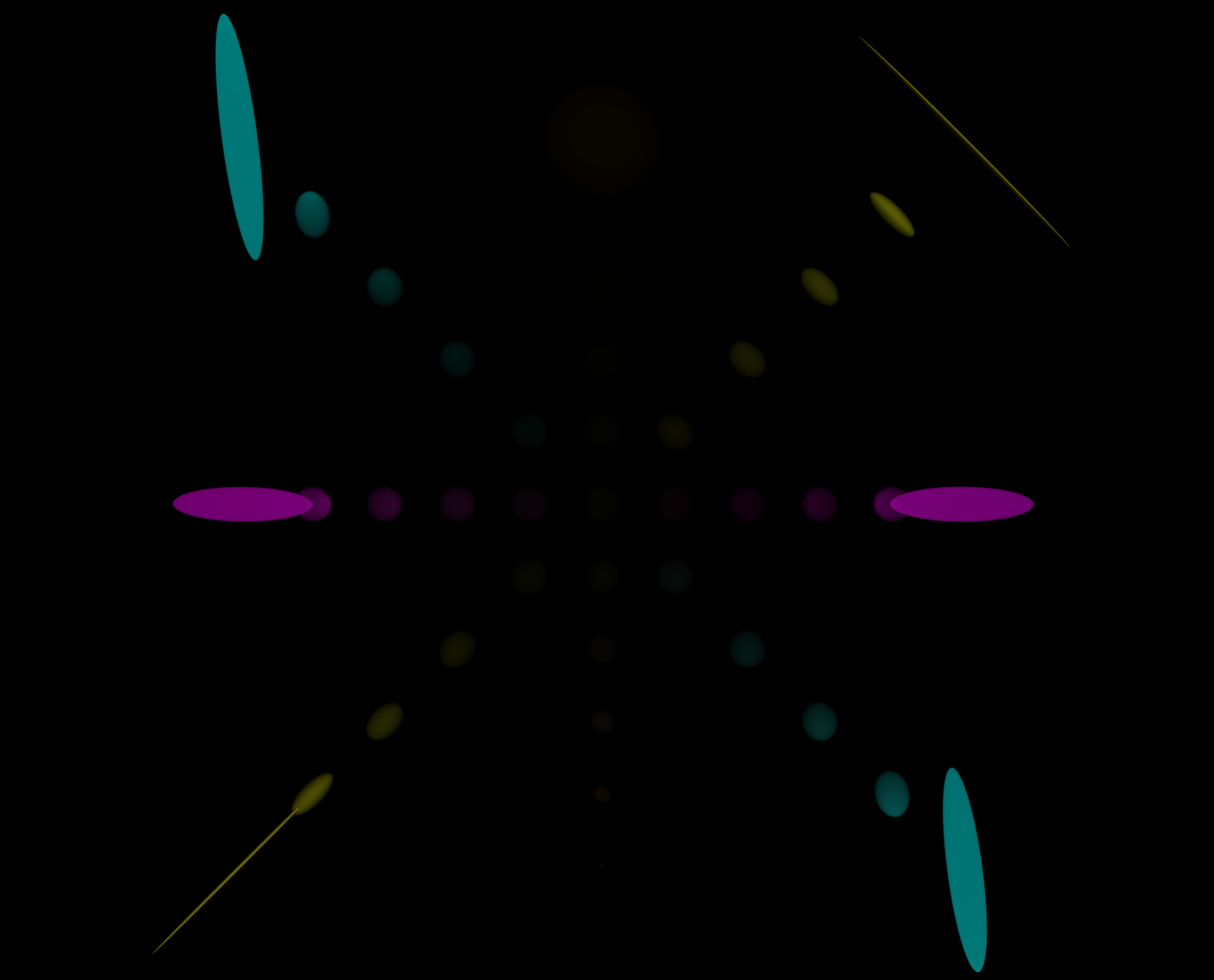}
  \caption{The same configuration against a black background.}
  \label{dbwholeb}
	\end{subfigure}
	\begin{subfigure}[t]{0.49\linewidth}
	    \includegraphics[width=\linewidth]{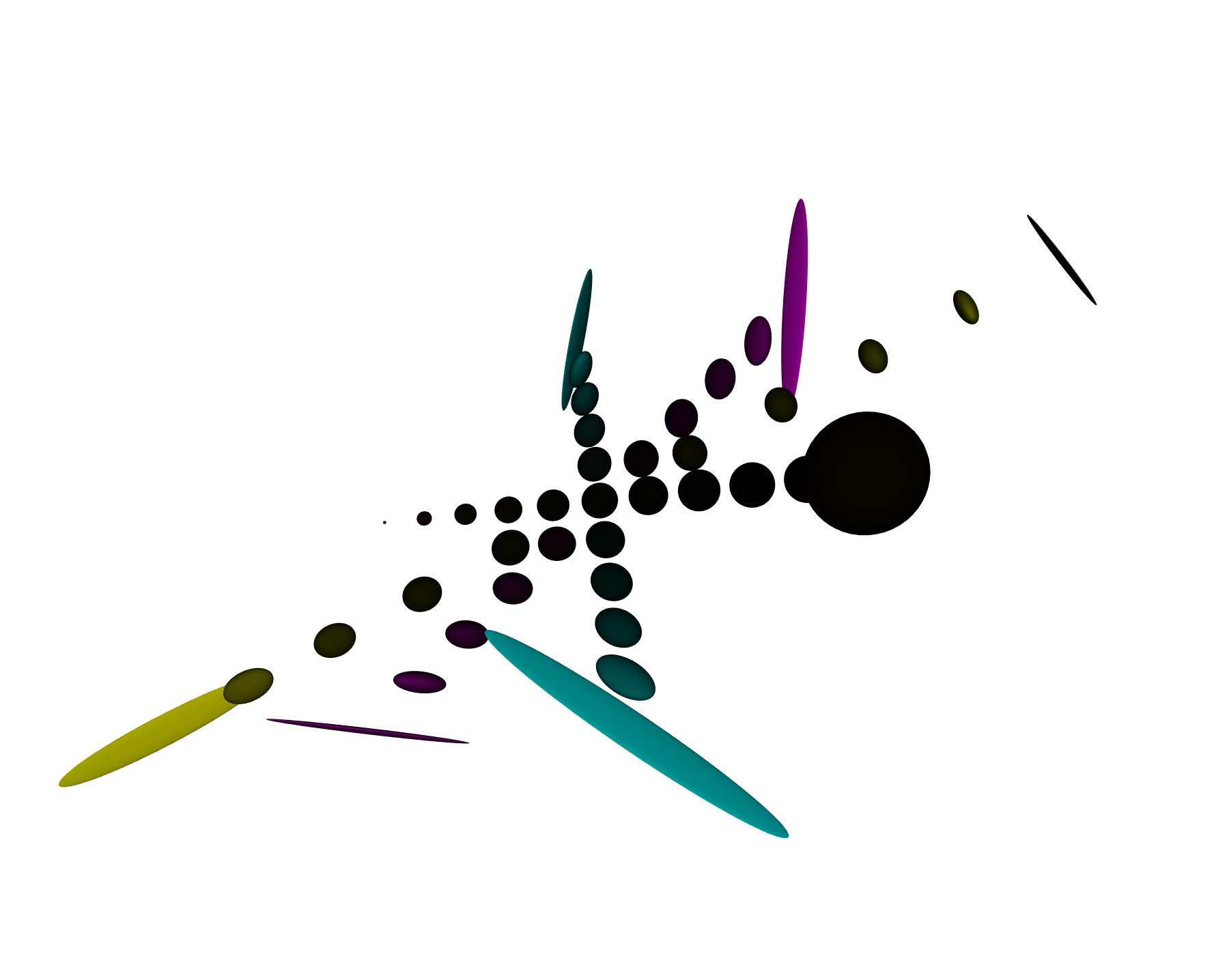}
      \caption{The same configuration from another angle.}
  \label{dbwholec}
    \end{subfigure}
	    \begin{subfigure}[t]{0.49\linewidth}
		\includegraphics[width=\linewidth]{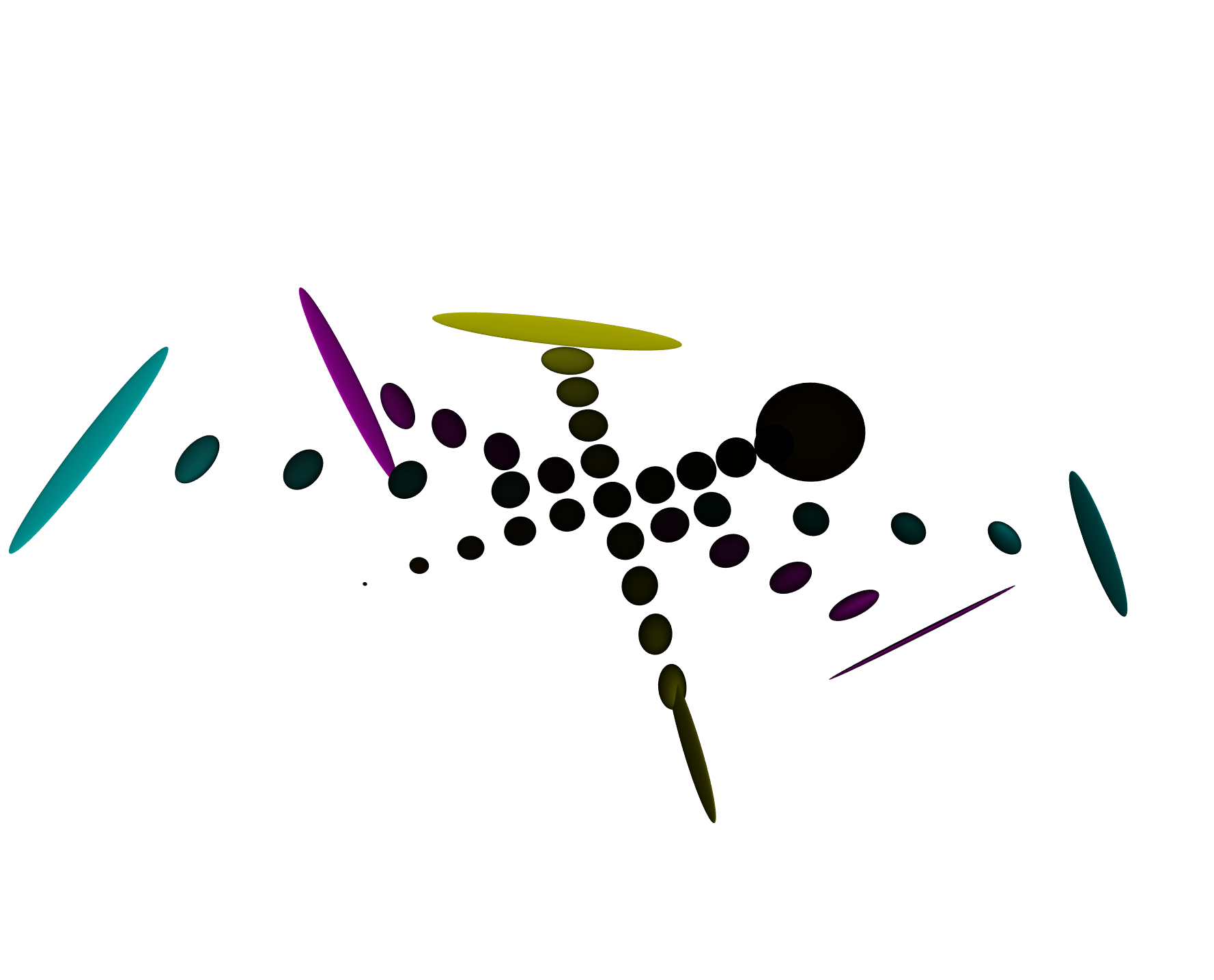}
  \caption{The same configuration from yet another angle.}
  \label{dbwholed}
	\end{subfigure}
	\caption{Quantiles for the data in Figure \ref{dbcc}, visualized as ellipsoids in three-dimensional space. The central ellipsoid is the median ($\beta=0$), and moving outward in each of eight directions, we get quantiles for $\beta=0.2,0.4,0.6,0.8,0.98$ respectively. In Figures \ref{dbwholea} and \ref{dbwholeb}, starting from the top and moving counterclockwise, the ellipsoids in a certain direction represent quantiles for $\xi_k$, $k=1,\ldots,8$ respectively.}
	\label{dbwholequantiles}
\end{figure}

Except for extreme values of $\beta$, the ellipsoids are quite close to spheres of a similar size. This makes sense, given the distribution of the original data seen in Figure \ref{dbcc}. For comparison, we also consider the patch of red ellipsoids, shown in Figure \ref{dbparta}, in the upper right quadrant of Figures \ref{dbcca} and \ref{dbccb}. Since this is a less spherical subset of the original data, its quantiles, calculated for the same values of $\beta$ and $\xi$, should look very different. These quantiles are visualized from three different angles in Figures \ref{dbpartb}, \ref{dbpartc} and \ref{dbpartd}, as in Figure \ref{dbwholequantiles}. The ellipsoids in Figure \ref{dbpartb} are arranged in the same way as those in Figure \ref{dbwholea} and \ref{dbwholeb}.
 
\begin{figure}[!t]
	\centering
	\begin{subfigure}[t]{0.49\linewidth}
		\includegraphics[width=\linewidth]{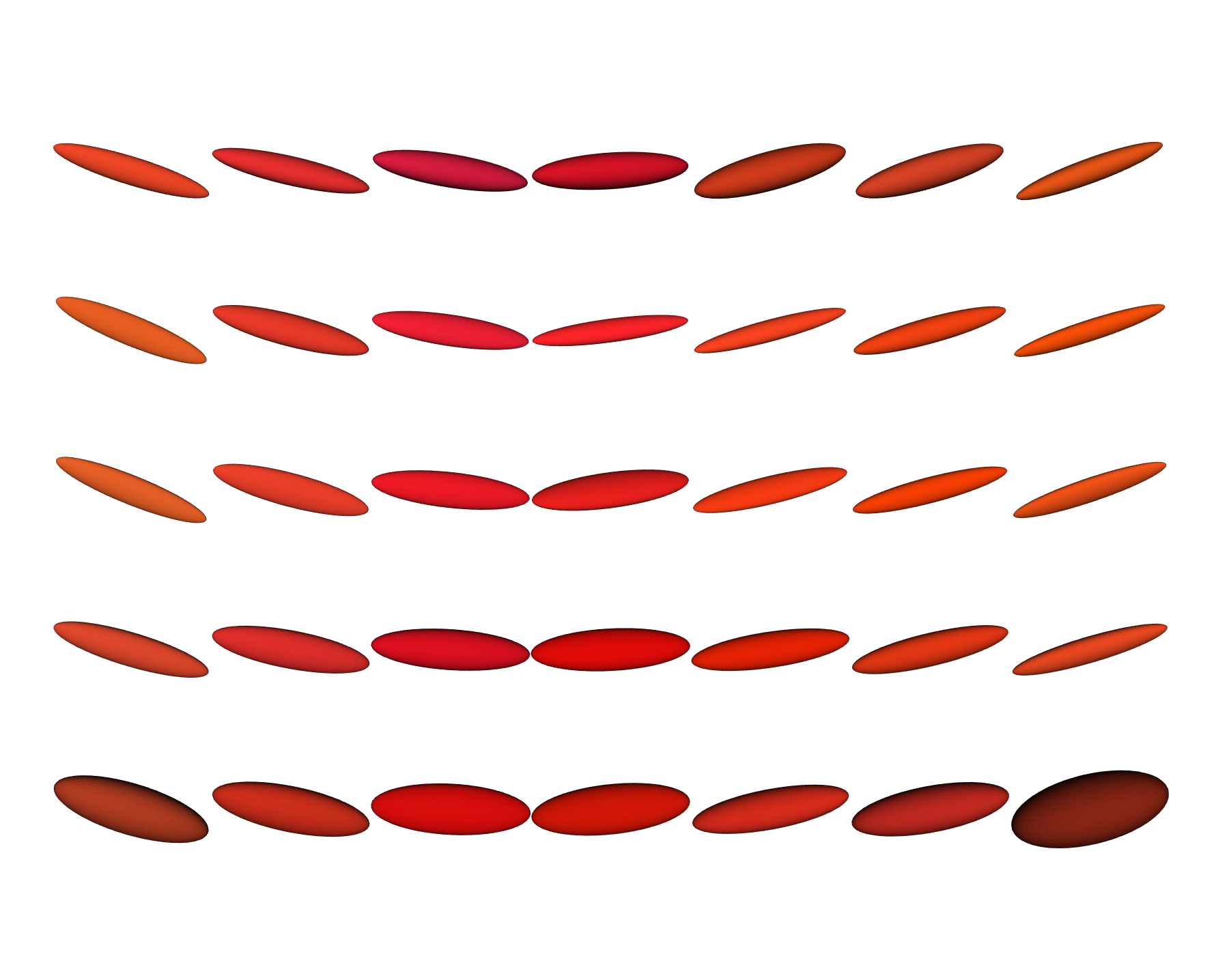}
      \caption{Subset of the original data.}
  \label{dbparta}
	\end{subfigure}
	\begin{subfigure}[t]{0.49\linewidth}
		\includegraphics[width=\linewidth]{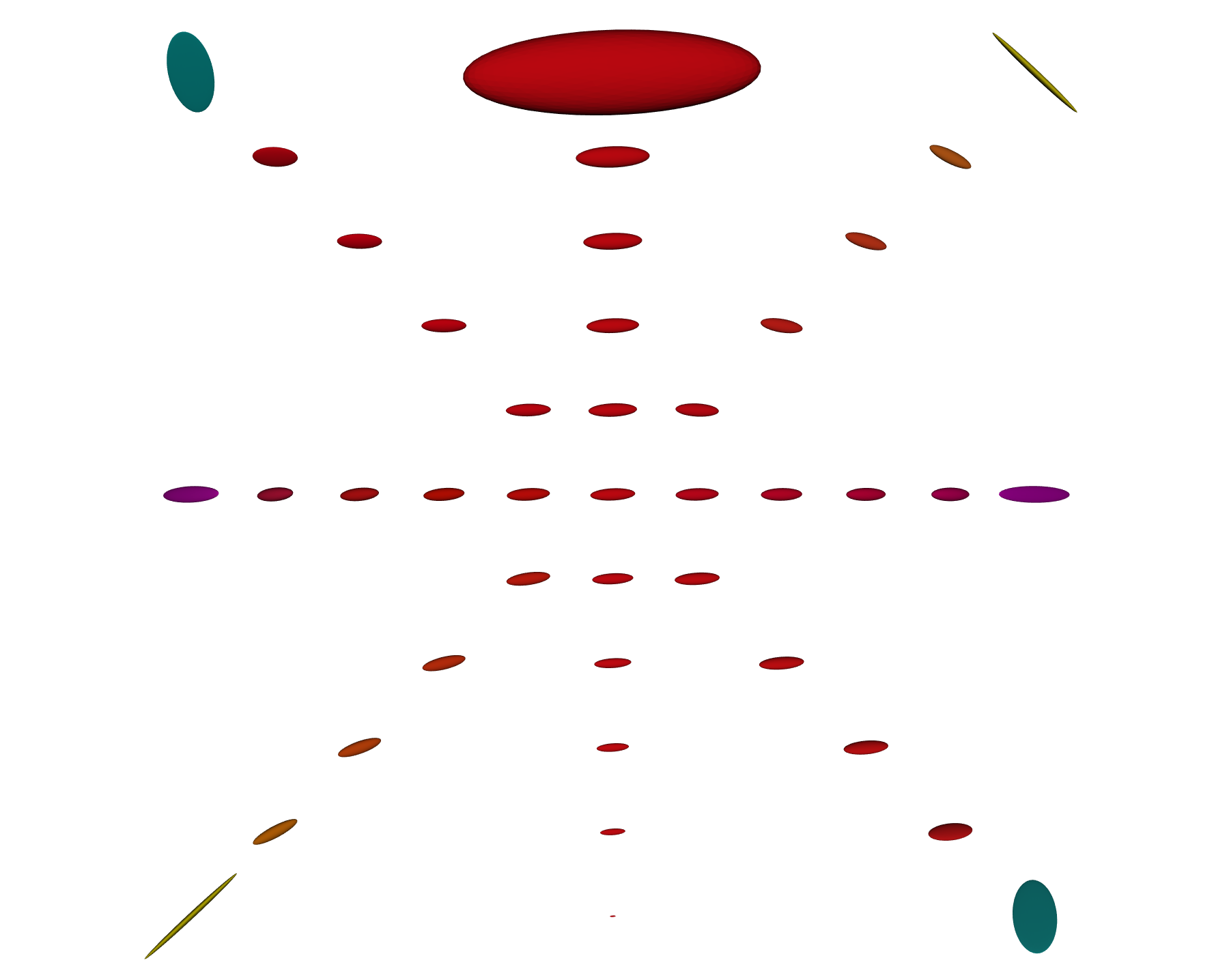}
      \caption{Quantiles arranged in the same configuration as in Figures \ref{dbwholea} and \ref{dbwholeb}.}
  \label{dbpartb}
	\end{subfigure}

	\begin{subfigure}[t]{0.49\linewidth}
	    \includegraphics[width=\linewidth]{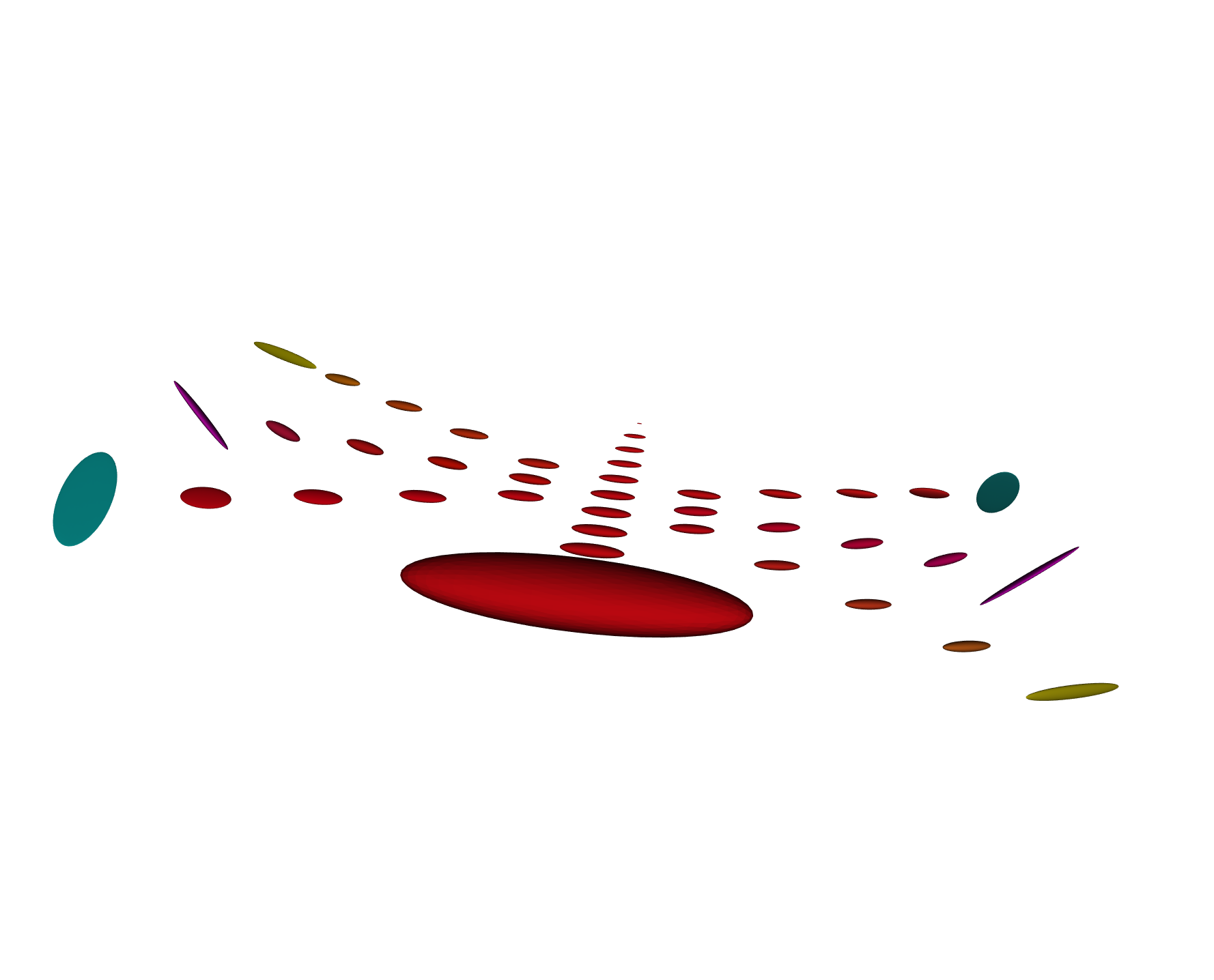}
      \caption{The same configuration from another angle.}
  \label{dbpartc}
    \end{subfigure}
	    \begin{subfigure}[t]{0.49\linewidth}
		\includegraphics[width=\linewidth]{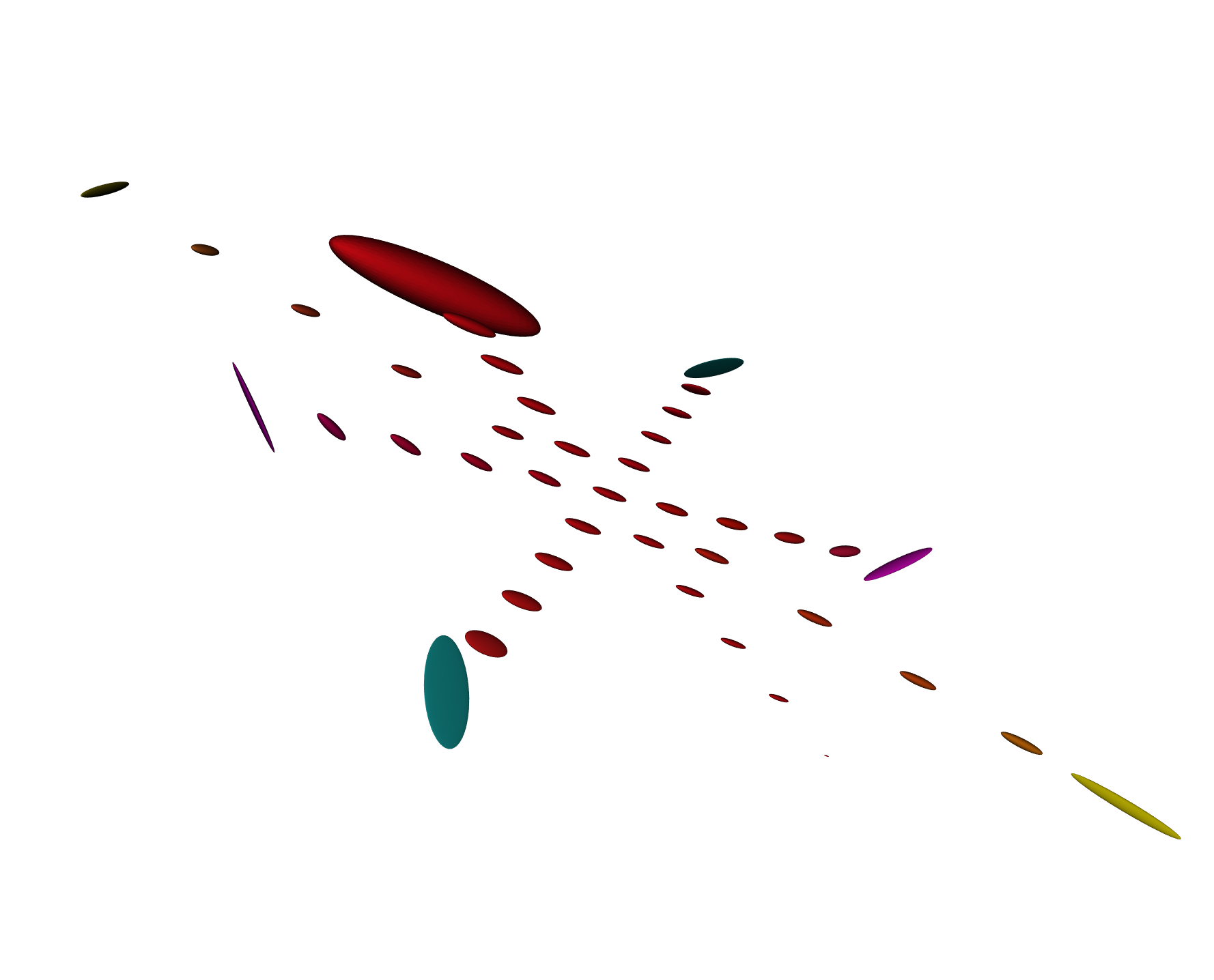}
      \caption{The same configuration from yet another angle.}
  \label{dbpartd}
	\end{subfigure}
	\caption{Subset of the original data and quantiles, visualized as ellipsoids in three-dimensional space.}
	\label{dbpartquantiles}
\end{figure}

\subsubsection{Measures of distributional characteristics}\label{dbmeasures}

In addition to the two dispersion measures in (\ref{dbdispersions}), the two skewness measures, two kurtosis measures and single spherical asymmetry measure of Section 6.2 of \cite{Shin2023} can also be adapted to data-based quantiles as
\begin{equation*}\begin{gathered}
\gamma_1(\beta):=\frac{\sup_{v\in S_{m}^{n-1}}\lVert \log_{m}(q(\beta,f_m(v)))+\log_{m}(q(\beta,f_m(-v)))\rVert}{\delta_1(\beta)},\\
\gamma_2(\beta):=\frac{\int_{S_{m}^{n-1}}\log_{m}(q(\beta,f_m(v))) dv/SA(n-1)}{\delta_2(\beta)},\\
\kappa_1(\beta,\beta'):=\frac{\delta_1(\beta')}{\delta_1(\beta)}, ~~
\kappa_2(\beta,\beta'):=\frac{\delta_2(\beta')}{\delta_2(\beta)}, ~~
\alpha(\beta):=\log\bigg(\frac{\sup_{\xi\in\partial M}d(m,q(\beta,\xi))}{\inf_{\xi\in\partial M}d(m,q(\beta,\xi))}\bigg).
\end{gathered}\end{equation*}
Then in the style of (\ref{dbdispersionsapprox}) and using notation of Section \ref{dbrobustness}, these quantities can be approximated with,
\begin{equation*}\begin{gathered}
\hat\gamma_{1,N}^\beta(\boldsymbol{X}):=\frac{\max_{k=1,\ldots,K}\lVert \log_{\hat m_N(\boldsymbol{X})}(\hat q_{k,N}(\boldsymbol{X})+\log_{\hat m_N(\boldsymbol{X})}(\hat p_{k,N}(\boldsymbol{X}))\rVert}{\hat \delta_{1,N}^\beta(\boldsymbol{X})},\\
\hat\gamma_{2,N}^\beta(\boldsymbol{X}):=\frac{(1/K)\sum_{k=1}^K\log_{\hat m_N(\boldsymbol{X})}(\hat q_{k,N}(\boldsymbol{X}))}{\hat \delta_{2,N}^\beta(\boldsymbol{X})},\\
\hat\kappa_{1,N}^{\beta,\beta'}(\boldsymbol{X}):=\frac{\delta_{1,N}^{\beta'}(\boldsymbol{X})}{\delta_{1,N}^\beta(\boldsymbol{X})}, ~~
\hat\kappa_{2,N}^{\beta,\beta'}(\boldsymbol{X}):=\frac{\delta_{2,N}^{\beta'}(\boldsymbol{X})}{\delta_{2,N}^\beta(\boldsymbol{X})}, ~~
\hat\alpha^\beta(\boldsymbol{X}):=\log\bigg(\frac{\max_{k=1,\ldots,K}d(\hat m_N(\boldsymbol{X}),q_k(\boldsymbol{X}))}{\min_{k=1,\ldots,K}d(\hat m_N(\boldsymbol{X}),q_k(\boldsymbol{X}))}\bigg),
\end{gathered}\end{equation*}
respectively. We calculate these quantities, as well as the dispersion measures in (\ref{dbdispersionsapprox}), with both non-extreme and extreme values of $\beta$ and $\beta'$: in the non-extreme case, $(\beta,\beta')=(0.2,0.8)$ for the kurtosis measures and $\beta=0.5$ otherwise, and the corresponding values in the extreme case are $(\beta,\beta')=(0.2,0.98)$ and $\beta=0.98$, respectively. Because $\mathcal{P}_3$ is 6-dimensional, we used a significantly larger value of $K$, 192, than in Section \ref{dbsims}, in which $M$ was only 2-dimensional. These 192 values of $\xi$ were chosen to be the image of 96 random points from the uniform distribution on $S_{\hat{m}_N(\mathbf{X})}^{n-1}$ and their antipodes under $f_{\hat{m}_N(\mathbf{X})}$; here again, $f_{\hat{m}_N(\mathbf{X})}$ is as defined in Section \ref{dbrobustness}. The measures were calculated using two different values of $\boldsymbol{X}$. In Table \ref{dbwholetable}, $\boldsymbol{X}$ represents the whole data set visualized in Figure \ref{dbwholea}, while in Table \ref{dbpartialtable}, it represents the partial data set from Figure \ref{dbparta}.

 \begin{table}[!h]
    \centering
    \caption{Estimates for the seven measures for the whole data set for various values of $\beta$ and $\beta'$.}
    {\small
\begin{tabular}{ |c|c|c|c|c|c|c| } 
\hline
$\hat\delta_1^{0.5}(\boldsymbol{X})$ & $\hat\delta_2^{0.5}(\boldsymbol{X})$ & $\hat\gamma_1^{0.5}(\boldsymbol{X})$ & $\lVert\hat\gamma_2^{0.5}(\boldsymbol{X})\rVert$ & $\hat\kappa_1^{0.2,0.8}(\boldsymbol{X})$ & $\hat\kappa_2^{0.2,0.8}(\boldsymbol{X})$ & $\hat\alpha^{0.5}(\boldsymbol{X})$\\
\hline
0.731 & 0.559 & 0.434 & 0.142 & 6.992 & 8.665 & 1.011 \\  \hline \hline
$\hat\delta_1^{0.98}(\boldsymbol{X})$ & $\hat\delta_2^{0.98}(\boldsymbol{X})$ & $\hat\gamma_1^{0.98}(\boldsymbol{X})$ & $\lVert\hat\gamma_2^{0.98}(\boldsymbol{X})\rVert$ & $\hat\kappa_1^{0.2,0.98}(\boldsymbol{X})$ & $\hat\kappa_2^{0.2,0.98}(\boldsymbol{X})$ & $\hat\alpha^{0.98}(\boldsymbol{X})$\\
\hline
6.589 & 5.982 & 0.280 & 0.064 & 26.501 & 33.782 & 0.710 \\ 
\hline
\end{tabular}
    }
    \label{dbwholetable}
\end{table}

 \begin{table}[!h]
    \centering
    \caption{Estimates for the seven measures for the partial data set for various values of $\beta$ and $\beta'$.}
    {\small
\begin{tabular}{ |c|c|c|c|c|c|c| } 
\hline
$\hat\delta_1^{0.5}(\boldsymbol{X})$ & $\hat\delta_2^{0.5}(\boldsymbol{X})$ & $\hat\gamma_1^{0.5}(\boldsymbol{X})$ & $\lVert\hat\gamma_2^{0.5}(\boldsymbol{X})\rVert$ & $\hat\kappa_1^{0.2,0.8}(\boldsymbol{X})$ & $\hat\kappa_2^{0.2,0.8}(\boldsymbol{X})$ & $\hat\alpha^{0.5}(\boldsymbol{X})$\\
\hline
1.068 & 0.861 & 0.121 & 0.034 & 4.558 & 5.698 & 0.513 \\  \hline \hline
$\hat\delta_1^{0.98}(\boldsymbol{X})$ & $\hat\delta_2^{0.98}(\boldsymbol{X})$ & $\hat\gamma_1^{0.98}(\boldsymbol{X})$ & $\lVert\hat\gamma_2^{0.98}(\boldsymbol{X})\rVert$ & $\hat\kappa_1^{0.2,0.98}(\boldsymbol{X})$ & $\hat\kappa_2^{0.2,0.98}(\boldsymbol{X})$ & $\hat\alpha^{0.98}(\boldsymbol{X})$\\
\hline
6.530 & 5.391 & 0.042 & 0.010 & 15.167 & 16.880 & 0.406 \\ 
\hline
\end{tabular}
    }
    \label{dbpartialtable}
\end{table}

The measures for both non-extreme and extreme values of $\beta$ suggest that the whole data set is more skewed and spherically asymmetric than the partial data set. However, according to $\hat \delta_1^{0.5}$ and $\hat \delta_2^{0.5}$, the whole data set is more dispersed than the partial data set, while $\hat \delta_1^{0.98}$ and $\hat \delta_2^{0.98}$ suggest the opposite. What do we make of this? The measures can be thought of as evaluating a characteristic based on the central bulk or the peripheral extremes of a distribution, depending on whether extreme or non-extreme quantiles, respectively, are used. Thus the seemingly contradictory conclusions make sense when comparing Figures \ref{dbcca} and \ref{dbparta}. On one hand, many of the ellipsoids in Figure \ref{dbcca} are more or less spherical, black and of similar sizes, meaning that much of the data represented is concentrated around a single point, while this degree of uniformity in a large portion of the data is not at all visible in Figure \ref{dbparta}; thus it make sense that $\hat \delta_1^{0.5}$ and $\hat \delta_2^{0.5}$ are greater for the whole data set. On the other hand, Figure \ref{dbcca} shows more variety of color and shape at the extremes than Figure \ref{dbparta}, in which all of the ellipsoids of roughly similar color and shape, and so it also makes sense that the order is reversed for $\hat \delta_1^{0.98}$ and $\hat \delta_2^{0.98}$. These conclusions, at least for the extreme case, are evident in Figures \ref{dbwholea} and \ref{dbpartb}: With the possible exception of the ones at the top and the bottom of the images, the ellipsoids for $\beta=0.98$ are visibly less extreme than the corresponding ellipsoids in the latter than in the former. Finally, given these conclusions about the dispersion and how $\hat\kappa(\beta,\beta')$ is calculated, the whole data set should have greater kurtosis than the partial data set, as is confirmed by Tables \ref{dbwholetable} and \ref{dbpartialtable}.

 \section{Discussion}\label{dbdiscussion}

This paper defines a new data-based notion of quantiles on Hadamard spaces as an alternative to the parameter-based one suggested in \cite{Shin2023}. We find that these new quantiles can be taken much further than the parameter-based quantiles. Stronger versions of asymptotic properties like strong consistency and joint asymptotic normality are demonstrated, and there are also several advantages related to robustness, extreme quantiles and the gradient of the loss functions. Through simulations, we showed empirically that an approximate gradient descent algorithm works quite well at computing data-based quantiles in contrast to parameter-based quantiles and illustrated how to test whether different data sets come from the same distribution using quantiles and permutation tests; finally, we computed data-based quantiles and measures of distributional characteristics for real DTI data.

The introduction notes several potential applications for quantiles on Hadamard spaces besides measures of distributional characteristics, and there are many other avenues of future research available just by generalizing existing work with multivariate geometric quantiles to Hadamard space settings; in particular, we hope to work on quantile regression models on Hadamard spaces. Research into these areas could be done using both types of quantiles, comparing the pros and cons of the two perspectives. Similarly, both parameter-based and data-based versions of expectiles and other $M$-quantiles could be defined on Hadamard spaces.

\section*{Acknowledgments} 
		
  This research was supported by the National Research Foundation of Korea (NRF) funded by the Korean government (2021R1A2C1091357).

 \begin{appendix}
		\section*{Appendix}
		
		\section{Proofs} \label{dbproofs}

  \subsection{Proof of the propositions in Section \ref{dbdef}}\label{dbproof_measurable}

  	\begin{proof}[Proof of Proposition \ref{dbmeasurable}]

If $x\neq p$, Proposition II.9.2(1) of \cite{Bridson1999} guarantees that the map $x\mapsto \angle_x(p,\xi)$ is upper semi-continuous. Since upper semi-continuous functions between spaces endowed with their Borel $\sigma$-algebras are measurable, the map $x\mapsto \rho(x,p;\beta,\xi)$ is a measurable map on $M\backslash\{p\}$, which is a measurable set, and 0 on $\{p\}$; it is therefore measurable on $M$.
	\end{proof}

\begin{proof}[Proof of Proposition \ref{dbequivariance}]
		
		When $x=p$, $\rho(g(x),g(p);\beta,g\xi)=0=\sigma\rho(x,p;\beta,\xi)$ since $g$ is bijective. The first part of the proof of Proposition 3.2 in \cite{Shin2023} demonstrates that scaled isometries preserve angles between geodesic, so when $x\neq p$,
		\begin{eqnarray*}
\rho(g(x),g(p);\beta,g\xi)&=&d'(g(p),g(x))-\beta d'(g(p),g(x))\mathrm{cos}(\angle_{g(x)}(g(p),g\xi)) \\
&=&\sigma d(p,x)-\beta \sigma d(p,x)\mathrm{cos}(\angle_X(p,\xi))  \\
&=&\sigma[d(p,x)-\beta d(p,x)\mathrm{cos}(\angle_X(p,\xi))] \\
&=&\sigma\rho(x,p;\beta,\xi),
		\end{eqnarray*}
		from which the conclusion follows since $g$ is bijective.
	\end{proof}

 \begin{proof}[Proof of Proposition \ref{dbnew}]
     At $p\neq x$, $d(\cdot,x)$ is continuous since $\lvert d(p,x)-d(p',x)\rvert\leq d(p,p')$, and $\cos(\angle_x(\cdot,\xi))$ is continuous thanks to Proposition II.9.2(1) in \cite{Bridson1999}. Since $(1-\beta)d(p,x)\leq\rho(x,p;\beta,\xi)\leq(1+\beta)d(p,x)$ and the limits of both sides as $p$ approaches $x$ are 0, the map is also continuous at $p=x$.
 \end{proof}


  \begin{proof}[Proof of Proposition \ref{dbbasic}]
			(a) For $\beta\in[0,1)$, the steps of the proof are completely analogous to those in the proof of Proposition 3.4(a) in \cite{Shin2023}, except that $p\mapsto\rho(x,p;\beta,\xi)$ is known to be continuous for all Hadamard spaces thanks to Proposition \ref{dbnew}, and therefore the result holds for all Hadamard spaces. Now for general $\beta\in[0,\infty)$, $G^{\beta,\xi}=G^{0,\xi}+2\beta(G^{0.5,\xi}-G^{0,\xi})$, so the conclusion follows.
			
			(b) It can be proved that the quantile set is nonempty, closed and bounded in exactly the same manner as in the proof of Proposition 3.4(b) in \cite{Shin2023}, except that this conclusion holds for all Hadamard spaces thanks to Proposition \ref{dbbasic}(a). Then by a generalization of the Hopf-Rinow theorem to length spaces (see for example Proposition I.3.7 in \cite{Bridson1999}), the local compactness of $M$ implies that the quantile set is compact.
		\end{proof}

\subsection{Proofs for the results in Section \ref{dbsc}}\label{dbproof_sulln}


\begin{proof}[Proof of Lemma~\ref{dbsulln}]
			$G^{\beta,\xi}$ is finite on all $M$ by Proposition \ref{dbbasic}(a). Denote by $C(L)$ the vector space of continuous functions on $L$ equipped with the uniform (sup) norm (i.e., for $f\in C(L)$, $\lVert f\rVert_{\sup}:=\sup_{p\in L}\lvert f(p)\rvert$), which is a Banach space; because $L$ is a compact metric space, the Stone-Weierstrass theorem implies that $C(L)$ is separable. Then, the restrictions $\hat{G}^{\beta,\xi}_N|_L$ and $G^{\beta,\xi}|_L$ to $L$ are in $C(L)$ by Proposition \ref{dbbasic}(a). In the rest of this proof we equip $C(L)$ with its Borel $\sigma$-algebra.

   The metric function is continuous everywhere and by Proposition II.9.2 of \cite{Bridson1999}, the map $x\mapsto \angle_x(p,\xi)$, and hence $x\mapsto -\beta d(p,x)\cos(\angle_x(p,\xi))$, is upper semi-continuous on $M\backslash\{p\}$ for fixed $p$. When $x\rightarrow p$, $\lvert-\beta d(p,x)\cos(\angle_x(p,\xi))\rvert\leq \beta d(p,x)\rightarrow 0$, so the map is continuous, and hence upper semi-continuous, when $x=p$ too. Therefore $x\mapsto \rho(x,p;\beta,\xi)$ is upper semi-continuous on $M$. So if $\rho(x_0,p;\beta,\xi)<r$ for real $r$, there exists some open neighborhood $U\subset$ of $x_0$ such that $x\in U$ implies $\rho(x,p;\beta,\xi)<r$. Then the preimage of $E_{p,r}:=$ under the map $x\mapsto\rho(x,\cdot;\beta,\xi)$ from $M$ to $C(L)$ contains $U$, and is thus open and hence measurable. Thus if we can show the collection of sets $\{E_{p,r}\}$ generates the open sets on $C(L)$, the aforementioned map into $C(L)$ is measurable.

   Denoting by $B_{p,r,\epsilon}$ the set $\{f\in C(L):r-\epsilon\leq f(p)\leq r+\epsilon\}$ for $\epsilon>0$, 
   \begin{align*}
   B_{p,r,\epsilon}=(C(L)\backslash E_{p,r-\epsilon})\bigcup(\bigcap_{m=1}^\infty E_{p,r+\epsilon+1/m}).
   \end{align*}
   Because $L$ is compact, it has a dense countable subset $L_D$. Then for $f_0\in C(L)$,
   \begin{align*}
       \{f\in C(L):\lVert f-f_0\rVert_{\sup}\leq\epsilon\}=\bigcap_{p\in L_D}B_{p,f_0(p),\epsilon}
   \end{align*}
   by the continuity of the elements of $C(L)$. This shows that $\{E_{p,r}\}$ generates the closed balls in $C(L)$. An open ball is a countable union of expanding concentric closed balls, and it is known that the open balls generate the open sets in spaces that are separable, or equivalently on metric spaces, second-countable or Lindel\"of. But it is also known that since $L$ is a compact metric space, $C(L)$ must be separable. Therefore $\rho(X,\cdot;\beta,\xi)|_L$, and hence each $\hat{G}^{\beta,\xi}_N|_L$, is a $C(L)$-valued random element. The rest of this proof proceeds identically to the corresponding part of the proof for Lemma 4.1 of \cite{Shin2023}.
		\end{proof}

  \begin{proof}[Proof of Theorem~\ref{dbslln}]
      The proof proceeds indentically to that of Theorem 4.1 in \cite{Shin2023}, except that thanks to Proposition \ref{dbbasic}, Lemma \ref{dbsulln} and a generalization of the Hopf-Rinow theorem (see for example Proposition I.3.7 in \cite{Bridson1999}) according to which closed and bounded sets in complete and locally compact length spaces are compact, the result holds for all locally compact Hadamard spaces. 
  \end{proof}

     \subsection{Proofs for the results in Section~\ref{dbasymp}}\label{dbproof_qgrad}


		\begin{proof}[Proof of Theorem~\ref{dbqgrad}]
			Take an $h\in T_{q}M$ such that $\lVert h\rVert_g=1$ and let $\gamma_{q,h}(t)=\mathrm{exp}_{q}(th)$. We suppress the $q$ and $h$ parameters in the notation for $\gamma_{q,h}$. $\rho$ is differentiable as a function of its second argument on $M\backslash\{x\}$. So, if $x\neq q$,
			\begin{equation} \label{dbyes}
			\lim_{t\rightarrow 0+}\frac{\rho(x,\gamma(t);\beta,\xi)-\rho(x,\gamma(0);\beta,\xi)}{t}=d\rho(x,\cdot;\beta,\xi)_{q}(h)=\langle\nabla \rho(x,q;\beta,\xi),h\rangle_g.
			\end{equation}
			If $x=q$, 
			\begin{equation}\begin{aligned} \label{dbno}
			\lim_{t\rightarrow 0+}\frac{\rho(x,\gamma(t);\beta,\xi)-\rho(x,\gamma(0);\beta,\xi)}{t}
			&=\lim_{t\rightarrow 0+}\frac{\lVert th\rVert_g-\langle \beta\xi_x,\log_{x}(\gamma(t))\rangle_g+0}{t}   \\
			&=\lim_{t\rightarrow 0+}\frac{t-\langle \beta\xi_x,t\gamma^\prime(0)\rangle_g}{t}   \\
			&=1-\langle \beta\xi_{q},h\rangle_g.
			\end{aligned}\end{equation}
			
			For all $t$,
			\begin{align*}
			\frac{\lvert\rho(x,\gamma(t);\beta,\xi)-\rho(x,\gamma(0);\beta,\xi)\rvert}{t} &\leq\frac{\vert d(x,\gamma(t))-d(x,\gamma(0))\rvert+\beta\lvert\langle\xi_x,\log_x(\gamma(t))-\log_x(\gamma(0))\rangle_g\rvert}{t} \\
   &\leq\frac{d(\gamma(t),\gamma(0))+\beta d(\gamma(t),\gamma(0))}{t} \\
   &=1+\beta
			\end{align*}
   by the Cartan-Hadamard theorem. So using the bounded convergence theorem, (\ref{dbyes}) and (\ref{dbno}), the rest of the proof follows in an identical manner to the corresponding part of the proof of Theorem 3.1 of \cite{Shin2023}. Notably, unlike in that theorem, compactness of the support is not needed to use the bounded convergence theorem here.
		\end{proof}

  Several lemmas are needed.

\begin{lemma}\label{dblipschitz}
    Let $M$ be an $n$-dimensional $(n\geq2)$ Hadamard manifold, $(\beta,\xi)\in[0,1)\times\partial M$ and $\psi$ a chart defined on some open $V\subset M$ such that $\psi(V)$ is bounded in the Euclidean metric. Then there exists some positive number $\kappa$ for which 
    \begin{equation*}
    \lvert\rho(x,p_1;\beta,\xi)-\rho(x,p_2;\beta,\xi)\rvert\leq \kappa\lVert\psi(p_1)-\psi(p_2)\rVert_2
    \end{equation*}
    for all $x\in M$ and $p_1,p_2\in V$.
\end{lemma}
\begin{proof}
The first part of this proof. in which we show the existence of a positive upper semi-continuous function $l:\psi(V)\times\psi(V)\rightarrow\mathbb{R}$ such that $d(\psi^{-1}(y),\psi^{-1}(z))=l(y,z)\lVert z-y\rVert_2$, is similar to that of Lemma B.1(a) of \cite{Shin2023}, but with crucial differences. Fix some $y'\in \psi(V)$. By the openness of $\psi(V)\subset\mathbb{R}^n$, for all $y$ and $z$ sufficiently close in the Euclidean metric to $y'$, the image of the Euclidean straight line $c_{y,z}$ defined on $[0,1]$ by $c_{y,z}(t)=y+t(z-y)$ is contained in $\psi(V)$. Denoting by $\zeta_{\max}(y)$ the largest eigenvalue of $g_y^\psi$, the $n\times n$ positive definite matrix representation of the Riemannian metric at $y\in \psi(V)$, $\zeta_{\max}$ is continuous because eigenvalues are the roots of polynomials whose coefficients are continuous functions of the entries of a matrix. Then define $\zeta_{\sup}(y,z):=\sup_{t\in[0,1]}\zeta_{\max}(c_{y,z}(t))$. The length of $c_{y,z}$ according to the Riemannian metric must be at least as large as the length of the Riemannian geodesic from $\psi^{-1}(y)$ to $\psi^{-1}(z)$, so
			\begin{align*}
			d(\psi^{-1}(y),\psi^{-1}(z))&\leq\int_{0}^1\sqrt{\dot c_{y,z}(t)^Tg_{c_{y,z}(t)}\dot c_{y,z}(t)} dt \\
   &=\int_{0}^1\sqrt{(z-y)^Tg_{c_{y,z}(t)}^\psi(z-y)} dt \\
   &\leq\sqrt{\zeta_{\sup}(y,z)}\int_{0}^1\lVert z-y\rVert_2dt \\
   &=\sqrt{\zeta_{\sup}(y,z)}\lVert z-y\rVert_2
			\end{align*}
			or
			\begin{align*}
			\frac{d(\psi^{-1}(y),\psi^{-1}(z))}{\lVert z-y\rVert_2}\leq\sqrt{\zeta_{\sup}(y,z)}
			\end{align*}
			if $y\neq z$. Recall that $y'$ is some fixed point in $\psi(V)$. By the continuity of $\zeta_{\max}$, there exists for any $\epsilon>0$ some $\delta>0$ for which $\lVert y-y'\rVert_2<\delta$ implies $\lvert \sqrt{\zeta_{\max}(y)}-\sqrt{\zeta_{\max}(y')}\rvert<\epsilon$. Because the image of $c_{y,z}$ is compact, $\zeta_{\sup}(y,z)=\zeta_{\max}(c_{y,z}(t))>0$ for some $t\in[0,1]$. Therefore, if $\lVert y-y'\rVert_2<\delta$ and $\lVert z-y'\rVert_2<\delta$, $\lVert c_{y,z}(t)-y'\rVert_2<\delta$ and $\lvert \sqrt{\zeta_{\sup}(y,z)}-\sqrt{\zeta_{\max}(y')}\rvert<\epsilon$. This means that 
			\begin{equation}\begin{aligned} \label{dblowersemi}
			\lim\sup_{(y,z)\rightarrow(y',y')}\frac{d(\psi^{-1}(y),\psi^{-1}(z))}{\lVert z-y\rVert_2}\leq\lim_{(y,z)\rightarrow(y',y')}\sqrt{\zeta_{\sup}(y,z)}=\sqrt{\zeta_{\max}(y')}.
			\end{aligned}\end{equation}
			
			Then, $l:\psi(V)\times\psi(V)\rightarrow\mathbb{R}$ defined by  
			\begin{align*}
			l(y,z)=\begin{cases}
			\frac{d(\psi^{-1}(y),\psi^{-1}(z))}{\lVert z-y\rVert_2}&~~\mbox{if $y\neq z$} \\ 
			\sqrt{\zeta_{\max}(y)} &~~ \mbox{if $y=z$}
			\end{cases}
			\end{align*}
			is continuous if $y\neq z$, while upper semi-continuity at $(y,y)\in \psi(V)\times \psi(V)$ follows from (\ref{dblowersemi}) and the continuity of $\sqrt{\zeta_{\max}(y)}$ on $\psi(V)$; thus $l:\psi(V)\times\psi(V)\rightarrow\mathbb{R}$ is a positive upper semi-continuous function.

    Then
    \begin{align*}
        \lvert\rho(x,p_1;\beta_1,\xi_k)-\rho(x,p_2;\beta_1,\xi_k)\rvert&\leq\lvert d(x,p_1)-d(x,p_2)\rvert+\beta\lvert\langle\xi_x,\log_x(p_1)-\log_x(p_2)\rangle\rvert \\
        &\leq (1+\beta)d(p_1,p_2) \\
        &=l(\psi(p_1),\psi(p_2))\lVert\psi(p_1)-\psi(p_2)\rVert_2
    \end{align*}
    for all $x,p_1,p_2\in V$ by the triangle inequality, Cauchy-Schwarz inequality and Cartan-Hadamard theorem. Since $\psi(V)$ is bounded, it is contained in some compact set and therefore $l$ attains a maximum, which we call $\kappa$, on $\psi(V)\times \psi(V)$.
\end{proof}

\begin{lemma}\label{dbtwo}
    Let $M$ be an $n$-dimensional ($n\geq 2$) Hadamard manifold. Then for each $k=1,...,K$ and any $p\in M$, $E[\lVert\Psi_k(X,p)\rVert_2^2]<\infty$.
\end{lemma}
\begin{proof}
    Fix $k$ and $p\in M\cong \phi(M)$. For any $x\neq p$ and $h\in T_pM\cong \mathbb{R}^n$, 
    \begin{equation*}
        \lim_{t\rightarrow 0}\frac{\rho_k(x,p+th)-\rho_k(x,p)}{\lVert th\rVert_2}=\langle\Psi_k(x,p),h\rangle_2.
    \end{equation*}
    Then letting $h=\Psi_k(x,p)$ and applying Lemma \ref{dblipschitz} with $(\beta,\xi)=(\beta_k,\xi_k)$, $V$ an open neighborhood of $p$ such that $\phi(V)$ is bounded in the Euclidean metric and $\psi=\phi|_V$, $\lVert\Psi_k(x,p)\rVert_2^2\leq\kappa$ if $x\neq p$. Thus 
    \begin{equation}\label{dbneed}
        \lVert\Psi_k(x,p)\rVert_2\leq\sqrt\kappa+\lVert\Psi_k(p,p)\rVert
    \end{equation} 
    for all $x$ and so $E[\lVert\Psi_k(X,p)\rVert_2]\leq\sqrt\kappa+\lVert\Psi_k(p,p)\rVert$.
\end{proof}

\begin{proof}[Proof of Theorem~\ref{dbclt1}] 
With just a handful of differences, the proof follows that of Theorem 4.2 of \cite{Shin2023}. We list those differences here. 

The core of the proof requires showing that all of the assumptions of Theorem 3 and its corollary in Section 4 of \cite{Huber1967} are satisfied. In the parameter-based case, Assumption (N-2) was included as part of (I) in the statement of the theorem, but in the data-based case, Assumption (N-2) holds by Theorem \ref{dbqgrad} and the existence of the density in a neighborhood of $q_k$. Similarly, Assumption (N-4) held by hypothesis in the parameter-based case, but here it is a consequence of (I), as shown in Lemma \ref{dbtwo}.

The conditions need to be renumbered: all mentions of (III), (IV) and (V) in the proof of Theorem 4.2 of \cite{Shin2023} should be relabeled as (II), (III) and (IV), respectively, here.
  
  In the proof of Lemma A.2, $R_k(x,p)$ is now defined to be the Riemannian gradient of $-\langle\beta_k(\xi_k)_x,\log_x(p)\rangle_g$ (instead of $\langle\beta_k(\xi_k)_p,\log_p(x)\rangle_g$) as a function of $p$, and $R_k$ is continuous on $M\times M$ by Proposition 3.3(b) of \cite{Shin2023}. Then, $\Psi_k(x,\tau)=-g_\tau\log_\tau(x)/\lVert\log_\tau(x)\rVert_g+g_\tau R_k(x,\tau)$ when $x\neq\tau$, and since $R_k(x,p)=-d(\log_x)_p^\dagger\beta_k(\xi_k)_x$ by the proof of Proposition \ref{dbgrad}, $\Psi_k(x,\tau)=g_\tau R_k(x,\tau)$ when $x=\tau$. 

  Finally, $T_k:M\times M\rightarrow \mathbb{R}$ is defined so that $T_k(x,q)$ is the Frobenius norm of the $n\times n$ Euclidean Hessian matrix of $-\langle\beta\xi_x,\log_x(q)\rangle_g$ (instead of $\langle\beta_k(\xi_k)_p,\log_p(x)\rangle_g$) as a function of $q$, and it is continuous on all of $M\times M$ by Proposition 3.3(b) of \cite{Shin2023}. This is where the $C^2$ condition was needed in the original proof of Theorem 4.2 of \cite{Shin2023}, but that condition is not necessary here.
\end{proof}

\begin{proof}[Proof of Corollary \ref{dbbdd}]
			There exists a sufficiently large neighborhood $U$ of $q_k$  which is bounded in the Riemannian metric and for which $P(X\not\in \bar{U})=0$, so condition (II) in Theorem \ref{dbclt1} holds.
		\end{proof}

Proposition \ref{dbsupp} can be derived in the same way as its counterpart in \cite{Shin2023}.

To prove Theorem \ref{dbclt2} we require Lemma B.1 from \cite{Shin2023}, which we repeat here for completeness.

\begin{lemma} \label{dblem}
			Let $M$ be an $n$-dimensional $(n\geq2)$ Hadamard manifold.
			\begin{itemize}
				\item[(a)] There exists a positive upper semi-continuous function $L:M\times M\rightarrow \mathbb{R}$ that satisfies
				\begin{align*}
				\lVert x-p\rVert_2=L(x,p)\lVert\log_p(x)\rVert_g.
				\end{align*}
				
				\item[(b)] There exists a positive upper semi-continuous function $S:M\times M\rightarrow \mathbb{R}$ that satisfies
				\begin{align*}
				\lvert D_{r'}\Psi^r(x,p;0,\xi)\rvert&\leq\frac{S(x,p)}{\lVert x-p\rVert_2},
				\end{align*}
				for all $r,r^\prime=1,\ldots,n$ whenever $x\neq p$.
			\end{itemize}
		\end{lemma}

\begin{lemma} \label{dbdiff}
		Let $M$ satisfy all the conditions stated in Theorem \ref{dbclt2} with the possible exception of (III). Then for each $k=1,\ldots,K$, the map $F_k:M\rightarrow \mathbb{R}$ defined by $p\mapsto E[\rho_k(X,p)]$ is 
  
  (a) differentiable on $Q_k$ from (I), on which its Euclidean gradient at $q$ is $E[\Psi_k(X,q)]$, and 
  
  (b) twice continuously differentiable in some open neighborhood of $q_k$, on which its Euclidean Hessian matrix at $q$ is $\Lambda_q^k$, the symmetric $n\times n$ matrix whose $(r,r')$-entry is $E[D_{r'}\Psi_k^r(X,q)]$.
	\end{lemma}

 \begin{proof}

(a) Fix $k$ and $p\in M\cong\phi(M)$. As in the proof of Lemma \ref{dbtwo}, apply Lemma \ref{dblipschitz} with $(\beta,\xi)=(\beta_k,\xi_k)$, $V$ an open neighborhood of $p$ such that $\phi(V)$ is bounded in the Euclidean metric and $\psi=\phi|_V$ to get
		\begin{equation*}\begin{aligned}
		&\bigg\lvert\frac{\rho_k(X,\tau)-\rho_k(X,p)-\Psi_k(X,p)^T(\tau-p)}{\lVert \tau-p\rVert_2}\bigg\rvert \\
		&\leq\frac{\lvert \rho_k(X,\tau)-\rho_k(X,p)\rvert}{\lVert \tau-p\rVert_2}+\frac{\lVert\Psi_k(X,p)\rVert_2\lVert \tau-p\rVert_2}{\lVert \tau-p\rVert_2} \\
		&\leq\kappa+\sqrt\kappa+\lVert\Psi_k(p,p)\rVert_2
		\end{aligned}\end{equation*}
for all $X\in M$ and $\tau\neq p$, where we have also used the Cauchy-Schwarz inequality and (\ref{dbneed}). The expression in the first line also converges to 0 as $\tau\rightarrow p$ if $X\neq p$, which is almost surely the case if $p=q\in Q_k$ by (I), so we can use the bounded convergence theorem to get
		\begin{align*}
		&\lim_{\tau\rightarrow q}\frac{F_k(\tau)-F_k(q)-E[\Psi_k(X,q)]^T(\tau-q)}{\lVert \tau-q\rVert_2} \\
		&=E\bigg[\lim_{\tau\rightarrow q}\frac{\rho_k(X,\tau)-\rho_k(X,q)-\Psi(X,q)^T(\tau-q)}{\lVert \tau-b\rVert_2}\bigg] \\
		&=0
		\end{align*}
		for all $q\in Q_k$, proving the desired result.

(b) Fix $k$. The map $T_k:M\times M\rightarrow \mathbb{R}$ defined so that $T_k(x,q)$ is the Frobenius norm of the $n\times n$ Euclidean Hessian matrix of $-\langle\beta\xi_x,\log_x(q)\rangle_g$ as a function of $q$, is continuous on all of $M\times M$ by Proposition 3.3(b) of \cite{Shin2023}, and
			\begin{align*}
			\lvert D_r'\Psi_k^r(x,q)\rvert&\leq\frac{S(x,q)}{\lVert x-q\rVert_2}+T_k(x,q),
			\end{align*}
			for all $r,r'=1,\ldots,n$ by Lemma \ref{dblem}(b).

This next part of the proof is inspired by the section labeled (i) in the proof of Theorem 4.2 of \cite{Shin2023}, with some adjustments. For any $p\in M$ and $t_0>0$, denote $\{p+tw:t\in[0,t_0],w\in S^{n-1}\}$ by $p+[0,t_0]\times S^{n-1}$, and do similarly with open and half-open intervals. If (IIa) is true for a certain value of $d_1$, it is also true for any smaller positive value of $d_1$. Therefore, assume without loss of generality that $q_k+[0,2d_1]\times S^{n-1}$ is contained in $Q_k\cap U$. Fix a $q$ satisfying $\lVert q-q_k\rVert<d_1$; there exists a $d_q>0$ such that $q+[0,d_q]\times S^{n-1}\subset q_k+[0,d_1]\times S^{n-1}$, and therefore
\begin{equation}\label{dbnewthree}
   E\bigg[\sup_{\tau:\lVert \tau-q\rVert_2\leq d_q}\lvert D_{r'}\Psi_k^r(X,\tau)\rvert^\alpha;X\not\in \bar{U}\bigg]<\infty
   \end{equation}
for all $r,r'=1,\ldots, n$ by (IIa). Also, $q+[0,2d_q]\times S^{n-1}$ is contained in $Q_k\cap U$.

Let $j$ be a positive integer such that $\alpha':=1+1/j\in(1,\alpha)$. (\ref{dbnewthree}), Lemma \ref{dblem}, and the boundedness of the density, which we will call $f$, on $Q_k$ imply
			\begin{align*}
			B_1:&=\max_{r,r':1\leq r,r'\leq n}E\bigg[\sup_{\tau:\lVert \tau-q\rVert\leq d_q}\lvert D_{r'}\Psi_k^r(X,\tau)\rvert^{\alpha'};X\not\in \bar{U}\bigg]<\infty, \\
			B_2:&=\sup_{(x,\tau):\lVert x-q\rVert_2\geq 2d_q, x\in \bar{U},\lVert \tau-q\rVert_2\leq d_q} \lvert D_{r'}\Psi_k^r(x,\tau)\rvert<\infty, \\
			B_3:&=\sup_{(x,\tau):\lVert x-q\rVert_2\leq 2d_q,\lVert \tau-q\rVert_2\leq d_q} S(x,\tau)<\infty, \\
			B_4:&=\sup_{(x,\tau):\lVert x-q\rVert_2\leq 2d_q,\lVert \tau-q\rVert_2\leq d_q} T_k(x,\tau)<\infty, \\
			B_5:&=\sup_{x:\lVert x-q\rVert_2\leq 2d_q} f(x)<\infty,
			\end{align*}
			since upper semi-continuous functions attain their maxima on compact sets. For non-negative real numbers $a$ and $b$ and positive integer $j$, $a+b\leq(a^{1/j}+b^{1/j})^j$ by the binomial theorem, so $(a+b)^{\alpha'-1}\leq a^{\alpha'-1}+b^{\alpha'-1}$. Therefore $(a+b)^{\alpha'}=(a+b)(a+b)^{\alpha'-1}\leq(a+b)(a^{\alpha'-1}+b^{\alpha'-1})=a^{\alpha'}+ab^{\alpha'-1}+a^{\alpha'-1}b+b^{\alpha'}$; so 
			\begin{align*}
			&\bigg(\frac{S(x,\tau)}{\lVert X-\tau\rVert_2}+T_k(x,\tau)\bigg)^{\alpha'} \\
   &\leq\frac{S(x,\tau)^{\alpha'}}{\lVert X-\tau\rVert_2^{\alpha'}}+\frac{S(x,\tau)T_k(x,\tau)^{\alpha'-1}}{\lVert X-\tau\rVert_2}+\frac{S(x,\tau)^{\alpha'-1}T_k(x,\tau)}{\lVert X-\tau\rVert_2^{\alpha'-1}}+T_k(x,\tau)^{\alpha'}
			\end{align*}
			when $x\neq \tau$, which we use below. 
   
   Then for any $w\in S^{n-1}$,
			\begin{equation}\begin{aligned} \label{dblong}
			&E\bigg[\sup_{\tau:\tau\in q+[0,d_q]\times w}\lvert D_{r'}\Psi_k^r(X,\tau)\rvert^{\alpha'};X\not\in q+[0,d_q]\times w\bigg] \\
			&=E\bigg[\sup_{\tau:\tau\in q+[0,d_q]\times w}\lvert D_{r'}\Psi_k^r(X,\tau)\rvert^{\alpha'};X\not\in q+[0,d_q]\times w, X\not\in \bar{U}\bigg] \\
			&+E\bigg[\sup_{\tau:\tau\in q+[0,d_q]\times w}\lvert D_{r'}\Psi_k^r(X,\tau)\rvert^{\alpha'};X\not\in q+[0,d_q]\times w, \lVert X-q\rVert_2\geq 2d_q, X\in \bar{U}\bigg] \\
			&+E\bigg[\sup_{\tau:\tau\in q+[0,d_q]\times w}\lvert D_{r'}\Psi_k^r(X,\tau)\rvert^{\alpha'};X\not\in q+[0,d_q]\times w, \lVert X-q\rVert_2< 2d_q\bigg] \\
			&\leq E\bigg[\sup_{\tau:\lVert \tau-q\rVert\leq d_q}\lvert D_{r'}\Psi_k^r(X,\tau)\rvert^{\alpha'};X \not\in \bar{U}\bigg] \\
			&+E\bigg[\sup_{\tau:\lVert \tau-q\rVert\leq d_q}\lvert D_{r'}\Psi_k^r(X,\tau)\rvert^{\alpha'};\lVert X-q\rVert_2\geq 2d_q, X\in \bar{U}\bigg] \\
			&+E\bigg[\sup_{\tau:\tau\in q+[0,d_q]\times w}
			\bigg(\frac{S(X,\tau)}{\lVert X-\tau\rVert_2}+T_k(X,\tau)\bigg)^{\alpha'};X\not\in q+[0,d_q]\times w,  \\
   &\qquad\lVert X-q\rVert_2\leq 2d_q\bigg] \\
			&\leq B_1+B_2^{\alpha'} \\
   &+B_3^{\alpha'}E\bigg[\sup_{\tau:\tau\in q+[0,d_q]\times w}\frac{1}{\lVert X-q\rVert_2^{\alpha'}};X\not\in q+[0,d_q]\times w, \lVert X-q\rVert_2\leq 2d_q\bigg] \\ 
			&+B_3B_4^{\alpha'-1}E\bigg[\sup_{\tau:\tau\in q+[0,d_q]\times w}\frac{1}{\lVert X-\tau\rVert_2};X\not\in q+[0,d_q]\times w, \lVert X-q\rVert_2\leq 2d_q\bigg] \\
			&+B_3^{\alpha'-1}B_4E\bigg[\sup_{\tau:\tau\in q+[0,d_q]\times w}\frac{1}{\lVert X-\tau\rVert_2^{\alpha'-1}};X\not\in q+[0,d_q]\times w, \lVert X-q\rVert_2\leq 2d_q\bigg]+B_4^{\alpha'}.
			\end{aligned}\end{equation}
			Now for $\nu>0$,         
			\begin{equation}\begin{aligned} \label{dblonger}
			&E\bigg[\sup_{\tau:\tau\in q+[0,d_q]\times w}\frac{1}{\lVert X-\tau\rVert_2^\nu};X\not\in q+[0,d_q]\times w, \lVert X-q\rVert_2\leq 2d_q\bigg]  \\
			&\leq B_5\int I(x\not\in q+[0,d_q]\times w, \lVert x-q\rVert_2\leq 2d_q)\sup_{\tau:\tau\in q+[0,d_q]\times w}\frac{1}{\lVert x-\tau\rVert_2^\nu} dx \\
			&=B_5\int I(y\not\in [0,d_q]\times (1,0,\ldots,0)^T\subset\mathbb{R}^n, \lVert y\rVert_2\leq 2d_q) \\
   &\qquad\cdot\sup_{\tau:\tau\in[0,d_q]\times (1,0,\ldots,0)^T\subset\mathbb{R}^n}\frac{1}{\lVert y-q\rVert_2^\nu} dy \\
			&=B_5\bigg(\int I(y^1<0, \lVert(y_1,\ldots,y_n)^T)\rVert_2\leq 2d_q)\frac{1}{\lVert(y_1,\ldots,y_n)^T)\rVert_2^\nu}dy^1\cdots dy^n \\
			&\qquad+\int I(0\leq y^1< d_q, \lVert(y_1,\ldots,y_n)^T)\rVert_2\leq 2d_q)\frac{1}{\lVert(0,y_2,\ldots,y_n)^T)\rVert_2^\nu}dy^1\cdots dy^n \\
			&\qquad+\int I(y^1\geq d_q, \lVert(y_1,\ldots,y_n)^T)\rVert_2\leq 2d_q)\frac{1}{\lVert(y_1-d_q,y_2,\ldots,y_n)^T\rVert_2^\nu}dy^1\cdots dy^n\bigg) \\
			&\leq B_5\bigg(\int I(y^1<0, \lVert(y_1,\ldots,y_n)^T)\rVert_2\leq 2d_q)\frac{1}{\lVert(y_1,\ldots,y_n)^T)\rVert_2^\nu}dy^1\cdots dy^n \\
			&\qquad+\int I(0\leq y^1< d_q, \lVert(0,y_2,\ldots,y_n)^T)\rVert_2\leq 2d_q)\frac{1}{\lVert(0,y_2,\ldots,y_n)^T)\rVert_2^\nu}dy^1\cdots dy^n \\
			&\qquad+\int I(y^1\geq d_q, \lVert(y_1-d_q,\ldots,y_n)^T)\rVert_2\leq 2d_q) \\
   &\qquad\cdot\frac{1}{\lVert(y_1-d_q,y_2,\ldots,y_n)^T\rVert_2^\nu}dy^1\cdots dy^n\bigg) \\
			&=B_5\bigg(\int I(y^1<0, \lVert(y_1,\ldots,y_n)^T)\rVert_2\leq 2d_q)\frac{1}{\lVert(y_1,\ldots,y_n)^T)\rVert_2^\nu}dy^1\cdots dy^n \\
			&\qquad+\int I(0\leq y^1< d_q, \lVert(0,y_2,\ldots,y_n)^T)\rVert_2\leq 2d_q)\frac{1}{\lVert(0,y_2,\ldots,y_n)^T)\rVert_2^\nu}dy^1\cdots dy^n \\
			&\qquad+\int I(y^1\geq 0, \lVert(y_1,\ldots,y_n)^T)\rVert_2\leq 2d_q)\frac{1}{\lVert(y_1,y_2,\ldots,y_n)^T\rVert_2^\nu}dy^1\cdots dy^n\bigg) \\
			&=B_5\bigg(\int I(\lVert(y_1,\ldots,y_n)^T)\rVert_2\leq 2d_q)\frac{1}{\lVert(y_1,\ldots,y_n)^T)\rVert_2^\nu}dy^1\cdots dy^n \\
			&\qquad+\int I(0\leq y^1< d_q, \lVert(0,y_2,\ldots,y_n)^T)\rVert_2\leq 2d_q)\frac{1}{\lVert(0,y_2,\ldots,y_n)^T\rVert_2^\nu}dy^1\cdots dy^n\bigg), 
			\end{aligned}\end{equation}
			where the substitution of $x$ with $y=(y^1,\ldots,y^n)^T$ in the third line signifies a shift of $\mathbb{R}^n$ by $-q$ and then rotation so that $w$ aligns with $(1,0,\ldots,0)^T\in\mathbb{R}^n$, and the substitution in the second-to-last equality comes from shifting $y_1$ by $-d_q$. In the standard $n$-dimensional spherical coordinate substitution of $y=(y^1,\ldots,y^n)^T$ with $(R,\theta_1,\ldots,\theta_{n-1})^T$, $R=\lVert y\rVert_2$, $y^1=R\sin(\theta_1)\cdots\sin(\theta_{n-1})$ and $y^r=R\sin(\theta_1)\cdots\sin(\theta_{n-r})\cos(\theta_{n-r+1})$ for $r=2,\ldots,n$, giving $dy^1\cdots dy^n=R^{n-1}\sin^{n-2}(\theta_1)\cdots\sin(\theta_{n-2})dRd\theta_1\ldots d\theta_{n-1}$. Therefore,
			\begin{equation}\begin{aligned} \label{dbshort1}
			&\int I(\lVert(y_1,\ldots,y_n)^T)\rVert_2\leq 2d_q)\frac{1}{\lVert(y_1,\ldots,y_n)^T)\rVert_2^\nu}dy^1\ldots dy^n  \\
			&=\int_0^{2\pi}\int_0^\pi\ldots\int_0^\pi\int_0^{2d_q} \frac{R^{n-1}\sin^{n-2}(\theta_1)\ldots\sin(\theta_{n-2})}{R^\nu}dRd\theta_1\ldots d\theta_{n-1}  \\
			&\leq \int_0^{2\pi}\int_0^\pi\ldots\int_0^\pi\int_0^{2d_q} \frac{R^{n-1}}{R^\nu}dRd\theta_1\ldots d\theta_{n-1}  \\
			&=\frac{(2d_q)^{n-\nu}}{n-\nu}\times \pi^{n-2}\times 2\pi \\
			&<\infty
			\end{aligned}\end{equation}
			if $n>\nu$. We now substitute $y=(y^1,\ldots,y^n)^T$ with $(y^1,R',\theta_1',\ldots,\theta_{n-2}')^T$, defined by $R'=\lVert (0,y^2,\ldots,y^n)^T\rVert_2$, $y^2=R'\sin(\theta_1')\ldots\sin(\theta_{n-2}')$ and $y^r=R'\sin(\theta_1')\ldots\sin(\theta_{n-1-r}')\cos(\theta_{n-r}')$ for $r=3,\ldots,n$; this substitution just expresses $(y_2,\ldots,y_n)$ in standard $(n-1)$-dimensional spherical coordinates. Then, $dy^1\ldots dy^n=(R')^{n-2}\sin^{n-3}(\theta_1')\ldots\sin(\theta_{n-3}')dy^1dR'd\theta_1'\ldots d\theta_{n-2}'$, and therefore,
			\begin{equation}\begin{aligned} \label{dbshort2}
			&\int I(0\leq y^1< d_q, \lVert(0,y_2,\ldots,y_n)^T\rVert_2\leq 2d_q)\frac{1}{\lVert(0,y_2,\ldots,y_n)^T)\rVert_2^\nu}dy^1\ldots dy^n  \\
			&=\int_0^{2\pi}\int_0^\pi\ldots\int_0^\pi\int_0^{2d_q}\int_0^{d_q} \frac{(R')^{n-2}\sin^{n-3}(\theta_1')\ldots\sin(\theta_{n-3}')}{(R')^\nu}dy^1dR'd\theta_1'\ldots d\theta_{n-2}'  \\
			&\leq\int_0^{2\pi}\int_0^\pi\ldots\int_0^\pi\int_0^{2d_q}\int_0^{d_q} \frac{(R')^{n-2}}{(R')^\nu}dy^1dR'd\theta_1'\ldots d\theta_{n-2}'  \\
			&= d_q\times\frac{(2d_q)^{n-1-\nu}}{n-1-\nu}\times\pi^{n-3}\times 2\pi \\
			&<\infty
			\end{aligned}\end{equation}
			if $n>\nu+1$. Therefore, (\ref{dbshort1}) and (\ref{dbshort2}) hold for each $\nu\in\{\alpha'-1,1,\alpha'\}$ when $n\geq 3$. Since (\ref{dblong}), (\ref{dblonger}), (\ref{dbshort1}) and (\ref{dbshort2}) do not depend on $w$, 
   \begin{equation}\label{dbh}
   \sup_{w\in S^{n-1}}E\bigg[\sup_{\tau:\tau\in q+[0,d_q]\times w}\lvert D_{r'}\Psi_k^r(X,\tau)\rvert^{\alpha'};X\not\in q+[0,d_q]\times w\bigg]<\infty
   \end{equation}
   for all $r,r'=1,\ldots,n$.

   Define $h^q_{r,r'}$ to be a map on $(0,\infty)\times S^{n-1}\times M$ defined by 
   \begin{equation*}
   h^q_{r,r'}(\eta,w,x)=\sup_{\tau:\tau\in q+[0,\eta]\times w}\lvert D_{r'}\Psi_k^r(x,\tau)-D_{r'}\Psi_k^r(x,q)\rvert I(x\not\in q+[0,\eta]\times w),
   \end{equation*}
   where $q+[0,\eta]\times w:=\{q+sw:0\leq s\leq\eta\}$.
			
			In this paragraph, we demonstrate the continuity of $E[h^q_{r,r'}(\eta,w,X)]$ as a function of $w$ for sufficiently small $\eta$, and subsequently the uniform convergence of the map $w\mapsto E[h^q_{r,r'}(\eta,w,X)]$ on $S^{n-1}$ to 0 as $\eta\rightarrow 0$. Fix $\eta\in(0,d_q)$, $w_0\in S^{n-1}$ and $x\in M\backslash\{q+[0,\eta]\times w_0\}$, and consider some closed neighborhood $W\subset S^{n-1}$ of $w_0$ in $S^{n-1}$ for which $x\not\in q+[0,\eta]\times W:=\{q+sw:s\in[0,\eta],w\in W\}$. Defining $H(x,\tau):=\lvert D_{r'}\Psi_k^r(x,\tau)-D_{r'}\Psi_k^r(x,q)\rvert$, $H(x,\cdot)$ is a continuous function on $q+[0,\eta]\times W$, so it is uniformly continuous on this set. Therefore, for any $\epsilon>0$, there exists some $\delta>0$ for which $\lVert \tau-\tau'\rVert_2<\delta$ and $\tau,\tau'\in q+[0,\eta]\times W$ implies $\lvert H(x,\tau)-H(x,\tau')\rvert<\epsilon$. Let $W^*\subset S^{n-1}$ be the open geodesic ball centered on $w_0$ of radius $\delta/\eta$, where the radius is measured in terms of the standard geodesic distance on $S^{n-1}$. Then, if $w\in W^*\cap W$ and $s\in[0,\eta]$, $\lVert sw-sw_0\rVert_2< s\delta/\eta\leq \delta$, so $H(x,q+sw)\in(H(x,q+sw_0)-\epsilon,H(x,q+sw_0)+\epsilon)$. Therefore, for all $s\in[0,\eta]$, $H(x,q+sw)<\sup_{\tau:q+[0,\eta]\times w_0}H(x,\tau)+\epsilon$, so 
			\begin{equation}\begin{aligned} \label{dbupper}
			h^q_{r,r'}(\eta,w,x)\leq h^q_{r,r'}(\eta,w_0,x)+\epsilon.
			\end{aligned}\end{equation}
			By the compactness of $q+[0,\eta]\times w_0$, there exists some $s_0\in[0,\eta]$ for which $H(x,q+s_0w_0)=\sup_{\tau:q+[0,\eta]\times w_0}H(x,\tau)=h^q_{r,r'}(\eta,w_0,x)$, so we can also say that
			\begin{equation}\begin{aligned} \label{dblower}
			h^q_{r,r'}(\eta,w,x)\geq H(x,q+s_0w)>H(x,q+s_0w_0)-\epsilon=h^q_{r,r'}(\eta,w_0,x)-\epsilon
			\end{aligned}\end{equation}
			if $w\in W^*\cap W$; (\ref{dbupper}) and (\ref{dblower}) imply that $h^q_{r,r'}(\eta,w,x)\rightarrow h^q_{r,r'}(\eta,w_0,x)$ as $w\rightarrow w_0$. We defined $q$ in such a way that $Q_k$ is a neighborhood of $q$, so  $\eta\in(0,d_q)$ can be chosen such that $\{z:\lVert z-q\rVert_2<\eta\}$ is contained in $Q_k$, where absolute continuity holds, in which case $P(X\in q+[0,\eta]\times W)=0$ and $h^q_{r,r'}(\eta,w,X)\rightarrow h^q_{r,r'}(\eta,w_0,X)$ as $w\rightarrow w_0$ almost surely. In addition, (\ref{dbh}) provides the $L^{\alpha'}$-boundedness of $\{h(\eta,w,X)\}_{w\in S^{n-1}}$ since $h(\eta,w,X)\leq \sup_{\tau:\tau\in q+[0,d_q]\times w}\lvert D_{r'}\Psi_k^r(X,\tau)\rvert I(X\not\in q+[0,d_q]\times w)$, implying uniform integrability. These two results imply that $E[h^q_{r,r'}(\eta,w,X)]\rightarrow E[h^q_{r,r'}(\eta,w_0,X)]$ as $w\rightarrow w_0$ by Vitali's convergence theorem, and so $E[h^q_{r,r'}(\eta,\cdot,X)]$ is continuous on $S^{n-1}$, as desired. In addition, $S^{n-1}$ is compact and for a fixed $w_0$, $E[h^q_{r,r'}(\eta,w_0,X)]$ decreases monotonically as $\eta\rightarrow 0$, and $h^q_{r,r'}(\eta,w_0,x)\rightarrow 0$ as $\eta\rightarrow 0$ if $x\neq q$, implying $\lim_{\eta\rightarrow 0}E[h^q_{r,r'}(\eta,w_0,X)]=0$ by (I), (\ref{dbh}) and the dominated convergence theorem. Therefore, by Dini's theorem, $E[h^q_{r,r'}(\eta,\cdot,X)]$ converges uniformly on $S^{n-1}$ to 0 as $\eta\rightarrow 0$.
			
			\sloppy For each $r,r'=1,\ldots,n$, $E[D_{r'}\Psi_k^r(X,q)]$ is finite by (\ref{dbh}) since $D_{r'}\Psi_k^r(x,q)\leq \sup_{\tau:\tau\in q+[0,d_q]\times w}\lvert D_{r'}\Psi_k^r(x,\tau)\rvert I(x\not\in q+[0,d_q]\times w)$ if $x\not\in q+[0,d_q]\times w$ and $P(X\not\in q+[0,d_q]\times w)=1$ by (I). Additionally, since the $n\times n$ matrix $D\Psi_k(x,q)$ is symmetric, $\Lambda_q^k$ is symmetric.
			
			For any $\tau\neq q$ such that $\lVert \tau-q\rVert_2<d_q$, the image of $c_{q,\tau}:[0,1]\rightarrow M$, the straight line segment in Euclidean space for which $c_{q,\tau}(0)=q$ and $c_{q,\tau}(1)=\tau$, is contained in $Q_k\subset M$. By the mean value theorem, for each $r=1,\ldots,n$, there exists some (stochastic) $t_r\in[0,1]$ for which, 
			\begin{align*}
			\Psi_k^r(x,\tau)&=\Psi_k^r(x,q)+(D_1\Psi_k^r(x,c_{q,\tau}(t_r)),\ldots,D_n\Psi_k^r(x,c_{q,\tau}(t_r)))c_{q,\tau}'(t_r) \\
			&=(D_1\Psi_k^r(x,c_{q,\tau}(t_r)),\ldots,D_n\Psi_k^r(x,c_{q,\tau}(t_r)))(\tau-q)
			\end{align*}
			if $x$ is not in the image of $c_{q,\tau}$, which is a set of probability 0 by (I).
			
			Therefore
			\begin{equation}\begin{aligned} \label{dbhess}
			&\lVert \lambda_k(\tau)-\lambda_k(q)-\Lambda_q^k(\tau-q) \rVert_2 \\
			&=\lVert E[\Psi_k(X,\tau)-\Psi_k(X,q);X\not\in q+[0,1]\times(\tau-q)]  \\
			&\qquad-E[D\Psi_k(X,q)(\tau-q);X\not\in q+[0,1]\times(\tau-q)]\rVert_2  \\
			&=\big\lVert E\big[[D_{r'}\Psi_k^r(X,c_{q,\tau}(t_r))-D_{r'}\Psi_k^r(X,q)]_{r,r'=1,\ldots,n}(\tau-q); \\
   &\qquad X\not\in q+[0,1]\times(\tau-q)\big]\big\rVert_2  \\
   &= E\big[\big\lVert[D_{r'}\Psi_k^r(X,c_{q,\tau}(t_r))-D_{r'}\Psi_k^r(X,q)]_{r,r'=1,\ldots,n}(\tau-q)\big\rVert_2; \\
   &\qquad X\not\in q+[0,1]\times(\tau-q)\big]  \\
			&\leq E\Bigg[\bigg\lVert\bigg[h^q_{r,r'}\bigg(\lVert \tau-q\rVert_2,\frac{\tau-q}{\lVert \tau-q\rVert_2},X\bigg)\bigg]_{r,r'=1,\ldots,n}\bigg\rVert_F\lVert \tau-q\Vert_2\Bigg]  \\
			&\leq E\bigg[\sum_{r'=1}^n\sum_{r=1}^n h^q_{r,r'}\bigg(\lVert \tau-q\rVert_2,\frac{\tau-q}{\lVert \tau-q\rVert_2},X\bigg)\lVert \tau-q\Vert_2\bigg]
			\end{aligned}\end{equation}
			Here, the last inequality follows because the $L_2$ norm is dominated by the $L_1$ norm, and the second-to-last because for any matrix $A$ and vector $v$ that are conformable, $\lVert Av\rVert_2\leq \lVert A\rVert_F\lVert v\rVert_2$ thanks to the Cauchy-Schwarz inequality and the fact that each component of $Av$ is the dot product of $v$ and a row in $A$.
(\ref{dbhess}) and the uniform convergence to 0 of each $E[h^q_{r,r'}(\eta,\cdot,X)]$ as a function on $S^{n-1}$ as $\eta\rightarrow 0$ imply that 
     \begin{align*}
			\frac{\lVert \lambda_k(\tau)-\lambda_k(q)-\Lambda_q^k(\tau-q) \rVert_2}{\lVert \tau-q\rVert_2}\leq\sum_{r'=1}^n\sum_{r=1}^n E\bigg[h^q_{r,r'}\bigg(\lVert \tau-q\rVert_2,\frac{\tau-q}{\lVert \tau-q\rVert_2},X\bigg)\bigg]\longrightarrow 0
			\end{align*}
			as $\tau\rightarrow q$. Hence, $\lambda_k$ is differentiable at any $q$ in $q_k+[0,d_1)\times S^{n-1}$ since $q$ was defined so that $\lVert q-q_k\rVert_2<d_1$, at which its derivative is the symmetric $\Lambda_q^k$.

   To prove the continuity of $q\rightarrow \Lambda_q^k$ in $q_k+[0,d_1)\times S^{n-1}$, fix $r$ and $r'$. Observe that since (I) implies that $P(X\not\in q+[0,d_q]\times w)=1$ for any $w\in S^{n-1}$, (\ref{dbh}) implies that $\{\lvert D_{r'}\Psi_k^r(X,\tau)\rvert^{\alpha'}\}_{\tau:\lVert \tau-q\rVert_2\leq d_q}$ is $L^{\alpha'}$-bounded and hence uniformly integrable. The continuity of $\tau\rightarrow D_{r'}\Psi_k^r(x,\tau)$ for any $x\in M$ whenever $\tau\neq x$ implies that $D_{r'}\Psi_k^r(X,\tau)\rightarrow D_{r'}\Psi_k^r(X,q)$ as $\tau\rightarrow q$ if $X\neq q$, which is almost surely the case by (I). Therefore we can apply Vitali's theorem to show that $E[D_{r'}\Psi_k^r(X,\tau)]\rightarrow E[D_{r'}\Psi_k^r(X,q)]$ as $\tau\rightarrow q$ for each $r,r'=1,\ldots, n$, and hence $q\rightarrow \Lambda_q^k$ is continuous in $q_k+[0,d_1)\times S^{n-1}$.
 \end{proof}

\begin{proof}[Proof of Theorem~\ref{dbclt2}] 
Consider the $NK$-dimensional Riemannian manifold $M^K$ with the induced product metric, the $M^K$-valued random element $\boldsymbol{X}:=(X,...,X)$ and its i.i.d. samples $\boldsymbol{X}_1:=(X_1,...,X_1),\ldots,\boldsymbol{X}_N:=(X_N,...,X_N)$, and the loss function $\boldsymbol{\rho}:M^K\rightarrow \mathbb{R}$ defined by $\boldsymbol{\rho}(\boldsymbol{x},\boldsymbol{p})=\sum_{k=1}^K \rho_k(\pi_k(\boldsymbol{x}),\pi_k(\boldsymbol{p}))$, where $\pi_k:M^k\rightarrow M$ is the projection map defined by $\pi_k(p_1,\ldots,p_K)=p_k$, $k=1,\ldots,K$. Our goal is to apply Theorem 2.11 of \cite{Eltzner2019} to this setting, which requires showing the validity of Assumptions 2.2, 2.3, 2.4 and 2.6 of that paper.

Clearly the minimizer of the population loss function $\boldsymbol{F}:M^K\rightarrow \mathbb{R}$ defined by $\boldsymbol{F}(\boldsymbol{p})=E[\boldsymbol{\rho}(\boldsymbol{X},\boldsymbol{p})]=\sum_{k=1}^K F_k(\pi_k(\boldsymbol{X}),\pi_k(\boldsymbol{p}))$ is unique and equal to $\boldsymbol{q}:=(q_1,\ldots,q_K)$ by (I) and the measurable selections of the minimizer sets of the sample loss function $\boldsymbol{p}\mapsto \sum_{i=1}^N \boldsymbol{\rho}(\boldsymbol{X}_i,\boldsymbol{p})$ are precisely the ordered $K$-tuples of measurable selections from the sample $(\beta_k,\xi_k)$-quantiles, $k=1,\ldots,K$, defined by $X_1,\ldots,X_N$. Thus by the consistency of sample quantiles demonstrated in Theorem \ref{dbslln}, Assumption 2.2 of \cite{Eltzner2019} is satisfied. The inverse exponential map at $\boldsymbol{q}\in M^K$, denoted by $\boldsymbol{\log}_{\boldsymbol{q}}:M^K\rightarrow T_{\boldsymbol{q}}M^K\cong T_{q_1}M\times \ldots\times T_{q_K}M\cong \mathbb{R}^{Kn}$, satisfies $\boldsymbol{\log}_{\boldsymbol{q}}(\boldsymbol{p})=(\log_{\pi_1(\boldsymbol{q})}(\pi_1(\boldsymbol{p}))^T,\ldots,\log_{\pi_k(\boldsymbol{q})}(\pi_k(\boldsymbol{p}))^T)^T$ and is smooth on all of $M^K$, and so Assumption 2.3 of \cite{Eltzner2019}, which requires a local manifold structure around $\boldsymbol{q}$, is valid. We will specifically consider $\boldsymbol{\log}_{\boldsymbol{q}}$ restricted to $\boldsymbol{V}:=V_1\times\ldots\times V_K$, where $V_k$ is the open geodesic ball in $M$ of some finite radius centered at $q_k$ for each $k=1,\ldots,L$.

The map $\boldsymbol{p}\mapsto \boldsymbol{\rho}(\boldsymbol{x},\boldsymbol{p})$ is differentiable everywhere except at $\boldsymbol{p}=\boldsymbol{x}$, and thus since $P(\boldsymbol{X}=\boldsymbol{q})=0$ thanks to (I), Assumption 2.4(i) of \cite{Eltzner2019}, which requires almost sure differentiability at $\boldsymbol{q}$, is satisfied. For each $k=1,\ldots,L$, let $\psi=\log_{q_k}$ and $V$ be $V_k$ in Lemma \ref{dblipschitz}, and call the resulting positive number guaranteed to exist by that lemma $\kappa_k$. Then for all $\boldsymbol{x}\in M^K$ and $\boldsymbol{p}_1,\boldsymbol{p}_2\in \boldsymbol{V}:=V_1\times\ldots\times V_K$,
\begin{equation*}\begin{aligned}
    \lvert\boldsymbol{\rho}(\boldsymbol{x},\boldsymbol{p}_1)-\boldsymbol{\rho}(\boldsymbol{x},\boldsymbol{p}_2)\rvert&\leq\sum_{k=1}^K\lvert\rho(\pi_k(\boldsymbol{x}),\pi_k(\boldsymbol{p}_1);\beta_k,\xi_k)-\rho(\pi_k(\boldsymbol{x}),\pi_k(\boldsymbol{p}_2);\beta_k,\xi_k)\rvert \\
    &\leq \sum_{k=1}^K\kappa_k\lVert\log_{q_k}(\pi_k(\boldsymbol{p}_1))-\log_{q_k}(\pi_k(\boldsymbol{p}_2))\rVert_2 \\
    &\leq K \big(\max_{k=1,\ldots,K}\kappa_K\big)\lVert\boldsymbol{\log}_{\boldsymbol{q}}(\boldsymbol{p}_1)-\boldsymbol{\log}_{\boldsymbol{q}}(\boldsymbol{p}_2)\rVert_2,
\end{aligned}\end{equation*}
so Assumption 2.4(ii) of \cite{Eltzner2019}, which requires almost surely local Lipschitz continuity at $\boldsymbol{q}$, is also valid with $\dot\rho=K \max_{k=1,\ldots,K}\kappa_K$.

Let $\boldsymbol{\exp}_{\boldsymbol{q}}$ be the exponential map at $\boldsymbol{q}$ of $M^K$, the inverse of $\boldsymbol{\log}_{\boldsymbol{q}}$. Because smooth charts preserve twice continuous differentiability, Lemma \ref{dbdiff}(b) shows that $\boldsymbol{F}\circ\boldsymbol{\exp}_{\boldsymbol{q}}$ defined on $\boldsymbol{\log}_{\boldsymbol{q}}(\boldsymbol{V})$ is twice continuously differentiable in some neighborhood of $0$ in $\boldsymbol{\log}_{\boldsymbol{q}}(\boldsymbol{V})$. The first derivative of $\boldsymbol{F}\circ\boldsymbol{\exp}_{\boldsymbol{q}}$ at 0 vanishes because Theorem \ref{dbqgrad}, (\ref{dbequiv}), Lemma \ref{dbdiff}(a) and (I) imply that the first derivative of $\boldsymbol{F}$ at $\boldsymbol{q}$ is 0.

Label the components of the exponential map so that  $\boldsymbol{\exp}_{\boldsymbol{q}}=(\boldsymbol{\exp}_{\boldsymbol{q}}^1,\ldots,\boldsymbol{\exp}_{\boldsymbol{q}}^{Kn}):\boldsymbol{\log}_{\boldsymbol{q}}(\boldsymbol{V})\rightarrow \boldsymbol{V}\cong \phi(V_1)\times\ldots\times\phi(V_K)\subset\mathbb{R}^{Kn}$. By the chain rule,
		\begin{align*}
		\frac{\partial (\boldsymbol{F}\circ\boldsymbol{\exp}_{\boldsymbol{q}})}{\partial y^r}(0)=\sum_{j=1}^{Kn}\frac{\partial \boldsymbol{F}}{\partial y^j}(\boldsymbol{q})\frac{\partial \boldsymbol{\exp}_{\boldsymbol{q}}^j}{\partial y^r}(0)
		\end{align*}
		and
		\begin{align*} 
		&\frac{\partial^2 (\boldsymbol{F}\circ\boldsymbol{\exp}_{\boldsymbol{q}})}{\partial y^{r'}\partial y^r}(0) \nonumber\\
		&=\sum_{j=1}^{Kn}\bigg[\bigg(\sum_{j'=1}^{Kn}\frac{\partial^2 \boldsymbol{F}}{\partial y^{j'}\partial y^j}(\boldsymbol{q})\frac{\partial \boldsymbol{\exp}_{\boldsymbol{q}}^{j'}}{\partial y^r}(0)\bigg)\frac{\partial \boldsymbol{\exp}_{\boldsymbol{q}}^j}{\partial y^{r'}}(0)+\frac{\partial \boldsymbol{F}}{\partial y^j}(\boldsymbol{q})\frac{\partial^2 \boldsymbol{\exp}_{\boldsymbol{q}}^j}{\partial y^{r'}\partial y^r}(0)\bigg].
		\end{align*}
		By the previous paragraph, each $(\partial \boldsymbol{F}/\partial y^j)(\boldsymbol{q})=0$; the Euclidean Hessian matrix of $\boldsymbol{F}\circ\boldsymbol{\exp}_{\boldsymbol{q}}$ at 0 is, by Lemma \ref{dbdiff}(b), the symmetric $Kn\times Kn$ matrix $\boldsymbol{D}^T \Lambda \boldsymbol{D}$, where $\Lambda$ is as defined in the statement of the theorem and $\boldsymbol{D}$ is the $Kn\times Kn$ Jacobian matrix at 0 of $\boldsymbol{\exp}_{\boldsymbol{q}}$.

Thus there exists a special orthogonal $Kn\times Kn$ matrix $\boldsymbol{R}$ so that 
\begin{equation}\label{dbdiag}
\boldsymbol{R}(\boldsymbol{D}^T \Lambda \boldsymbol{D})\boldsymbol{R}^T=\text{diag}(L_1,\ldots,L_{Kn}),
\end{equation}
a diagonal matrix of the eigenvalues of $\boldsymbol{D}^T \Lambda \boldsymbol{D}$. Therefore the multivariate version of Taylor's theorem implies that 
\begin{align*}
    \boldsymbol{F}\circ\boldsymbol{\exp}_{\boldsymbol{q}}(\boldsymbol{v})&=\boldsymbol{F}\circ\boldsymbol{\exp}_{\boldsymbol{q}}(0)+0+\frac{1}{2}\boldsymbol{v}^T\boldsymbol{D}^T\Lambda \boldsymbol{D}\boldsymbol{v}+o(\lVert\boldsymbol{v}\rVert_2^2) \\
    &=\boldsymbol{F}\circ\boldsymbol{\exp}_{\boldsymbol{q}}(0)+\frac{1}{2}\boldsymbol{v}^T\boldsymbol{R}^T\boldsymbol{R}\boldsymbol{D}^T\Lambda \boldsymbol{D}\boldsymbol{R}^T\boldsymbol{R}\boldsymbol{v}+o(\lVert\boldsymbol{v}\rVert_2^2) \\
    &=\boldsymbol{F}\circ\boldsymbol{\exp}_{\boldsymbol{q}}(0)+\sum_{j=1}^{Kn}\frac{L_j}{2}(\boldsymbol{R}\boldsymbol{v})_j^2+o(\lVert\boldsymbol{v}\rVert_2^2),
\end{align*}
where $(\boldsymbol{R}\boldsymbol{v})_j$ is the $j$th entry of the $Kn$-vector $\boldsymbol{R}\boldsymbol{v}$, $j=1,\ldots,Kn$. Thus letting $r=2$ in Assumption 2.6 of \cite{Eltzner2019}, which requires sufficient smoothness at $0$ of $\boldsymbol{F}\circ\boldsymbol{\exp}_{\boldsymbol{q}}$, all of the assumptions necessary for the smeary central limit theorem of that paper are satisfied.

Thus defining 
\begin{align*}
    \hat{\boldsymbol{v}}_N':=\begin{cases} \boldsymbol{\log}_{\boldsymbol{q}}((\hat q_{1,N}^T,\ldots,\hat q_{K,N}^T)^T)&\text{ if $\hat q_{1,N}\in V_1,\ldots,\hat q_{K,N}\in V_K$,} \\ 0&\text{ otherwise}\end{cases}
\end{align*}
in $\boldsymbol{\log}_{\boldsymbol{q}}(\boldsymbol{V})$, we can apply Theorem 2.11 of \cite{Eltzner2019}, keeping in mind our correction in Remark \ref{dbmistake}, to say that
\begin{align*}
    \sqrt{N}\boldsymbol{R}\hat{\boldsymbol{v}}_N'\rightsquigarrow N\bigg(0,\frac{1}{4}\text{diag}\bigg(\frac{L_1}{2},\ldots,\frac{L_{Kn}}{2}\bigg)^{-1}\boldsymbol{R}\boldsymbol{C}\boldsymbol{R}^T\text{diag}\bigg(\frac{L_1}{2},\ldots,\frac{L_{Kn}}{2}\bigg)^{-1}\bigg)
\end{align*}
\sloppy as $N\rightarrow\infty$, where $\boldsymbol{C}$ is the covariance matrix of the Euclidean gradient of $\boldsymbol{v}\mapsto\boldsymbol{\rho}(\boldsymbol{X},\boldsymbol{\exp}_{\boldsymbol{q}}(\boldsymbol{v}))$ at $\boldsymbol{v}=0$. This gradient is $\boldsymbol{D}^T(\Psi_1(X,q_1)^T,\ldots,\Psi_K(X,q_K)^T)^T$ by the chain rule since the gradient of $\boldsymbol{p}\mapsto\boldsymbol{\rho}(\boldsymbol{x},\boldsymbol{p})$ is $(\Psi_1(X,q_1)^T,\ldots,\Psi_K(X,q_K)^T)^T$; therefore $\boldsymbol{C}=\boldsymbol{D}^T\Sigma\boldsymbol{D}$. This and (\ref{dbdiag}) mean that the covariance matrix of the above limiting distribution is $\boldsymbol{R}\boldsymbol{D}^{-1}\Lambda^{-1}(\boldsymbol{D}^T)^{-1}\boldsymbol{R}^T\boldsymbol{R}\boldsymbol{D}^T\Sigma\boldsymbol{D}\boldsymbol{R}^T\boldsymbol{R}\boldsymbol{D}^{-1}\Lambda^{-1}(\boldsymbol{D}^T)^{-1}\boldsymbol{R}^T=\boldsymbol{R}\boldsymbol{D}^{-1}\Lambda^{-1}\Sigma\Lambda^{-1}(\boldsymbol{D}^T)^{-1}\boldsymbol{R}^T$, recalling that charts are diffeomorphism and hence $D$ is invertible. Premultiplication by $\boldsymbol{R}^T$ gives $\sqrt{N}\hat{\boldsymbol{v}}_N'\rightsquigarrow N(0,\boldsymbol{D}^{-1}\Lambda^{-1}\Sigma\Lambda^{-1}(\boldsymbol{D}^T)^{-1})$ as $N\rightarrow\infty$, and Lemma 3.2 of \cite{Eltzner2019} shows that the 0-smeariness of $\{\hat{\boldsymbol{v}}_N'\}$ with limiting distribution $N(0,\boldsymbol{D}^{-1}\Lambda^{-1}\Sigma\Lambda^{-1}(\boldsymbol{D}^T)^{-1})$ implies the 0-smeariness of $\{\boldsymbol{\exp}_{\boldsymbol{q}}(\boldsymbol{\hat{\boldsymbol{v}}_N'})-(q_1^T,\ldots,q_K^T)^T\}$ with limiting distribution $N(0,\boldsymbol{D}\boldsymbol{D}^{-1}\Lambda^{-1}\Sigma\Lambda^{-1}(\boldsymbol{D}^T)^{-1}\boldsymbol{D}^T)$ because the Jacobian matrix of $\boldsymbol{\exp}_{\boldsymbol{q}}$ at 0 is $\boldsymbol{D}$; that is, $\sqrt{N}[\boldsymbol{\exp}_{\boldsymbol{q}}(\hat{\boldsymbol{v}}_N')-(q_1^T,\ldots,q_K^T)^T]\rightsquigarrow N\big(0,\Lambda^{-1}\Sigma\Lambda^{-1})$ as $N\rightarrow\infty$.

Finally, since $P(\boldsymbol{\exp}_{\boldsymbol{q}}(\hat{\boldsymbol{v}}_N')=(\hat q_{1,N}^T,\ldots,\hat q_{K,N}^T)^T)\geq P(\hat q_{1,N}\in V_1,\ldots,\hat q_{K,N}\in V_K)\rightarrow 1$ as $N\rightarrow\infty$ by the consistency of sample quantiles shown in Theorem \ref{dbslln}, $(\hat q_{1,N}^T,\ldots,\hat q_{K,N}^T)^T=\boldsymbol{\exp}_{\boldsymbol{q}}(\hat{\boldsymbol{q}}_N')+o_p(1)$ and so $\sqrt{N}[(\hat q_{1,N}^T,\ldots,\hat q_{K,N}^T)^T-(q_1^T,\ldots,q_K^T)^T]\rightsquigarrow N(0,\Lambda^{-1}\Sigma\Lambda^{-1})$ as $N\rightarrow\infty$, completing the proof.
\end{proof}

 \subsection{Proofs for the results in Section~\ref{dbrobustness}}\label{dbproof_robustness}

We require the following lemma.

 \begin{lemma}\label{dbinfbdp}
    The breakdown point of $\hat q_N(\beta,\xi)$ is at least $(\lceil(1-\beta)N/2\rceil)/N$.
\end{lemma}

\begin{proof}
    Letting $N':=\lceil(1-\beta)N/2\rceil-1$, then 
    \begin{equation}\label{dbint}
    \alpha:=(1-\beta)N/2-N'\in(0,1].
    \end{equation}
    We will corrupt $N'$ of the data points $(X_1,\ldots,X_N)$ and let $\hat q_\xi$ be a sample $(\beta,\xi)$-quantile of this data set. Without loss of generality, let the first $N'$ elements be corrupted and denote a sample $(\beta,\xi)$-quantile of the corrupted set $(Y_1,\ldots,Y_{N'},X_{N'+1},\ldots,X_N)$ by $\hat q'_\xi$. We can do this without loss of generality because sample quantile sets do not depend on the order of the data points and we will treat all sample $(\beta,\xi)$-quantile functions in this proof. Since
    \begin{equation*}
        \langle \xi_{Y_i},\log_{Y_i}(\hat q'_\xi)\rangle-\langle \xi_{Y_i},\log_{Y_i}(\hat q_\xi)\rangle\leq d(\hat q_\xi,\hat q'_\xi)
    \end{equation*}
    by the Cauchy-Schwarz inequality and the Cartan-Hadamard theorem,
    \begin{equation}\begin{aligned}\label{dbcorrupt}
        \rho(Y_i,\hat q'_\xi;\beta,\xi)&=d(\hat q'_\xi,Y_i)-\beta\langle\xi_{Y_i},\log_{Y_i}(\hat q'_\xi) \rangle \\
        &\geq d(\hat q_\xi,Y_i)-d(\hat q_\xi,\hat q'_\xi)-\beta d(\hat q_\xi,\hat q'_\xi)-\beta\langle \xi_{Y_i},\log_{Y_i}(\hat q_\xi)\rangle\\
        &=\rho(Y_i,\hat q_\xi;\beta,\xi)-(1+\beta)d(\hat q_\xi,\hat q'_\xi)
    \end{aligned}\end{equation}
    by the triangle inequality. Let $l_\xi:=\max_{i=1,\ldots,N}\rho(X_i,\hat q_\xi;\beta,\xi)$, $R_\xi:=l_\xi/(1-\beta)$ and $B(\hat q_\xi,r)$ be the closed geodesic ball of radius $r$ centered at $\hat q_\xi$ for any $r>0$. Now suppose 
    \begin{equation} \label{dbsuppose1}
    d(\hat q_\xi,\hat q'_\xi)>2R_\xi.
    \end{equation}
    Then $X_i\in B(\hat q_\xi,R_\xi)\subset B(\hat q_\xi,2R_\xi)$ and $\hat q'_\xi\not\in B(\hat q_\xi,2R_\xi)$, so denoting the intersection between $B(\hat q_\xi,2R_\xi)$ and the geodesic segment connecting $X_i$ and $\hat q'_\xi$ by $V_i$, $d(V_i,X_i)\geq R$ and therefore
    \begin{equation}\begin{aligned}\label{dbnoncorrupt}
        \rho(X_i,\hat q_\xi;\beta,\xi)&=\frac{d(\hat q'_\xi,X_i)}{d(V_i,X_i)}(d(V_i,X_i)-\beta\langle\xi_{X_i},\log_{X_i}(V_i)\rangle) \\
        &\geq\frac{d(\hat q'_\xi,X_i)}{d(V_i,X_i)}(1-\beta)d(V_i,X_i) \\
        &=(1-\beta)d(V_i,X_i)+\bigg(\frac{d(\hat q'_\xi,X_i)}{d(V_i,X_i)}-1\bigg)(1-\beta)d(V_i,X_i) \\
        &\geq(1-\beta)R_\xi+(1-\beta)d(\hat q'_\xi,V_i) \\
        &\geq l_\xi+(1-\beta)\inf_{Z\in B(\hat q_\xi,2R_\xi)}d(\hat q'_\xi,Z) \\
        &=\rho(X_i,\hat q_\xi;\beta,\xi)+(1-\beta)(d(\hat q_\xi,\hat q'_\xi)-2R_\xi)
    \end{aligned}\end{equation}
    Therefore by (\ref{dbint}), (\ref{dbcorrupt}) and (\ref{dbnoncorrupt}),
    \begin{equation*}\begin{aligned}
        &\sum_{i=1}^{N'}\rho(Y_i,\hat q'_\xi;\beta,\xi)+\sum_{i=N'+1}^N\rho(X_i,\hat q'_\xi;\beta,\xi) \\
        &\geq \sum_{i=1}^{N'}\rho(Y_i,\hat q_\xi;\beta,\xi)+\sum_{i=N'+1}^N\rho(X_i,\hat q_\xi;\beta,\xi)-N'(1+\beta)d(\hat q_\xi,\hat q'_\xi)+(N-N')(1-\beta)(d(\hat q_\xi,\hat q'_\xi)-2R)\\
        &=\sum_{i=1}^{N'}\rho(Y_i,\hat q_\xi;\beta,\xi)+\sum_{i=N'+1}^N\rho(X_i,\hat q_\xi;\beta,\xi)-\bigg[\frac{(1-\beta)N}{2}-\alpha\bigg](1+\beta)d(\hat q_\xi,\hat q'_\xi) \\
        &\qquad+\bigg[\frac{(1+\beta)N}{2}+\alpha\bigg](1-\beta)(d(\hat q_\xi,\hat q'_\xi)-2R_\xi)\\
        &=\sum_{i=1}^{N'}\rho(Y_i,\hat q_\xi;\beta,\xi)+\sum_{i=N'+1}^N\rho(X_i,q;\beta,\xi)+2\alpha d(\hat q_\xi,\hat q'_\xi)-[(1-\beta^2)N+2\alpha(1-\beta)]R_\xi. \\
    \end{aligned}\end{equation*}
    By the definition of $\hat q'_\xi$, this implies $d(\hat q_\xi,\hat q'_\xi)\leq[(1-\beta^2)N/(2\alpha)+(1-\beta)]R_\xi$. Alternatively, if (\ref{dbsuppose1}) is not true then $d(\hat q_\xi,\hat q'_\xi)\leq 2R_\xi$, so regardless of the precise values of $y_1,\ldots,y_N$, 
    \begin{equation}\label{dbupper}
        d(\hat q_\xi,\hat q'_\xi) \leq R_\xi\max\{2, (1-\beta^2)N/(2\alpha)+(1-\beta)\}.
    \end{equation}
    Therefore the breakdown point is at least $(N'+1)/N$.
    \end{proof}
    
    \begin{proof}[Proof of Theorem~\ref{dbbdp}]
   (a) and (b) As in the proof of Lemma \ref{dbinfbdp}, $\hat q_\xi$ is a sample $(\beta,\xi)$-quantile of the data set $(X_1,\ldots,X_N)$. Suppose $N''\leq N$ points are corrupted; specifically, we will let $Z_1(t)=\ldots=Z_{N''}(t)=\exp_{\hat q_\xi}(t\beta\xi_{\hat q_\xi})$ for $t>0$. Then denote an element of the $(\beta,\xi)$-quantile set of $\{Z_1(t),\ldots,Z_{N''}(t),X_{N''+1},\ldots,X_N\}$ by $\hat q''_\xi(t)$. Then the triangle inequality implies 
    \begin{equation*}
    t-d(\hat q_\xi,X_i)=d(\hat q_\xi,Z_1(t))-d(\hat q_\xi,X_i)\leq d(X_i,Z_1(t))\leq d(\hat q_\xi,Z_1(t))+d(\hat q_\xi,X_i)=t+d(\hat q_\xi,X_i),
    \end{equation*}
    implying $\lim_{t\rightarrow\infty}d(X_i,Z_1(t))/t=1$ and $\lim_{t\rightarrow\infty}\log_{X_i}(Z_1(t))/t=\xi_{X_i}$ by Proposition 3.3(a) of \cite{Shin2023}. Hence
    \begin{equation}\begin{aligned}\label{dbbdpfirst}
        &\frac{1}{t}\bigg(\sum_{i=1}^{N''}\rho(Z_i(t),Z_1(t);\beta,\xi)+\sum_{i=N''+1}^{N}\rho(X_i,Z_1(t);\beta,\xi)\bigg) \\
        &=\sum_{i=N''+1}^{N}\bigg(\frac{d(X_i,Z_1(t))}{t}-\bigg\langle\beta\xi_{X_i},\frac{\log_{X_i}(Z_1(t))}{t}\bigg\rangle\bigg) \\
        &\rightarrow \sum_{i=N''+1}^{N} (1-\langle\beta\xi_{X_i},\xi_{X_i}\rangle) \\
        &=(N-N'')(1-\beta) 
    \end{aligned}\end{equation}
    as $t\rightarrow\infty$.
    
    Now assuming that the breakdown point is greater than $N''/N$, $C:=\sup_{t>0}d(\hat q_\xi,\hat q''_\xi(t))<\infty$. Then
    \begin{equation*}
        t-C\leq d(\hat q_\xi,Z_1(t))-d(\hat q_\xi,\hat q''_\xi(t))\leq d(Z_1(t),\hat q''_\xi(t))\leq d(\hat q_\xi,Z_1(t))+d(\hat q_\xi,\hat q''_\xi(t))\leq t+C
    \end{equation*}
    implies $\lim_{t\rightarrow\infty}d(Z_1(t),\hat q''_\xi(t))/t=1$, 
    \begin{equation*}
        0\leq\lvert\langle\beta\xi_{Z_1(t)},\log_{Z_1(t)}(\hat q_\xi)-\log_{Z_1(t)}(\hat q''_\xi(t))\rangle\rvert\leq \beta d(\hat q_\xi,\hat q''_\xi(t))\leq \beta C,
    \end{equation*}
    which follows from the Cauchy-Schwarz inequality and the Cartan-Hadamard theorem, implies $\lim_{t\rightarrow\infty}\langle\beta\xi_{Z_1(t)},\log_{Z_1(t)}(\hat q_\xi)-\log_{Z_1(t)}(\hat q''_\xi(t))\rangle/t=0$, and for each $i$,
    \begin{equation*}
        0\leq \rho(X_i,\hat q''_\xi(t);\beta,\xi)\leq (1+\beta)d(X_i,\hat q''_\xi(t))\leq (1+\beta)(d(X_i,\hat q_\xi)+C)
    \end{equation*}    
    implies $\lim_{t\rightarrow\infty}\rho(X_i,\hat q''_\xi(t);\beta,\xi)/t=0$. Thus
    \begin{equation}\begin{aligned}\label{dbbdpsecond}
        &\frac{1}{t}\bigg(\sum_{i=1}^{N''}\rho(Z_i(t),\hat q''_\xi(t);\beta,\xi)+\sum_{i=N''+1}^{N}\rho(X_i,\hat q''_\xi(t);\beta,\xi)\bigg) \\
        &=N''\bigg(\frac{d(Z_1(t),\hat q''_\xi(t))}{t}-\frac{\langle\beta\xi_{Z_1(t)},\log_{Z_1(t)}(\hat q_\xi)\rangle}{t} \\
        &\qquad+\frac{\langle\beta\xi_{Z_1(t)},\log_{Z_1(t)}(\hat q_\xi)-\log_{Z_1(t)}(\hat q''_\xi(t))}{t}\bigg)+\sum_{i=N''+1}^{N}\frac{\rho(X_i,\hat q''_\xi(t);\beta,\xi)}{t} \\
        &=N''\bigg(\frac{d(Z_1(t),\hat q''_\xi(t))}{t}-\frac{\langle\beta\xi_{Z_1(t)},-t\xi_{Z_1(t)}\rangle}{t} \\
        &\qquad+\frac{\langle\beta\xi_{Z_1(t)},\log_{Z_1(t)}(\hat q_\xi)-\log_{Z_1(t)}(\hat q''_\xi(t))}{t}\bigg)+\sum_{i=N''+1}^{N}\frac{\rho(X_i,\hat q''_\xi(t);\beta,\xi)}{t} \\
        &\rightarrow N''(1-\langle\beta\xi_{Z_1(t)},-\xi_{Z_1(t)}\rangle) \\
        &=N''(1+\beta) 
    \end{aligned}\end{equation}
as $t\rightarrow\infty$. 

If $(1-\beta)N/2\not\in\mathbb{Z}$, let $N''=\lceil(1-\beta)N/2\rceil$. On the other hand, if $(1-\beta)N/2\in\mathbb{Z}$, let $N''=(1-\beta)N/2+1$. In either case $N''>(1-\beta)N/2$, so $(N-N'')(1-\beta)<(1-\beta^2)N/2$ and $N''(1+\beta)>(1-\beta^2)N/2$, and thus (\ref{dbbdpfirst}) and (\ref{dbbdpsecond}) imply that for all sufficiently large $t$, 
\begin{align*}
    &\sum_{i=1}^{N''}\rho(Z_i(t),Z_1(t);\beta,\xi)+\sum_{i=N''+1}^{N}\rho(X_i,Z_1(t);\beta,\xi) \\
    &< \sum_{i=1}^{N''}\rho(Z_i(t),\hat q''_\xi(t);\beta,\xi)+\sum_{i=N''+1}^{N}\rho(X_i,\hat q''_\xi(t);\beta,\xi),
\end{align*}
contradicting the definition of $\hat q''_\xi(t)$. Thus the assumption that the breakdown point is greater than $N''/N$ is wrong, proving the desired result in light of Lemma \ref{dbinfbdp}.

(c) If $\beta\neq 0$ and $N$ is such that $(1-\beta)N/2\in\mathbb{Z}$, let $M=\mathbb{R}$. Traditionally, quantiles in $\mathbb{R}$ are indexed by some $\tau\in(0,1)$. There are two elements in $\partial \mathbb{R}$, one, which we will call $\xi_1$, corresponding to indices greater than $1/2$ and the other, which we will call $\xi_2$, to indices less than $1/2$.

    A classic quantile function maps a sample to the smallest element of the sample quantile set; such an element exists because quantile sets are compact by Proposition \ref{dbbasic}(b). Noting that $N-(1-\beta)N/2=(1+\beta)N/2$ is an integer, if $\xi=\xi_1$, the chosen sample quantile is also the $((1+\beta)N/2)$th smallest element of the sample since the sample quantile set in this case would be the closed interval between the $((1+\beta)N/2)$th and $((1+\beta)N/2+1)$th smallest elements. Thus if $(1-\beta)N/2$ or fewer elements are corrupted, the chosen sample quantile cannot be greater than the largest of the uncorrupted elements or smaller than the $(\beta N)$th smallest, and thus it is contained in $[\min\{X_1,\ldots,X_N\},\max\{X_1,\ldots,X_N\}]$ for all $(X_1,\dots,X_N)\in M^N$. Therefore the breakdown point at any sample must be $(\lceil(1-\beta)N/2\rceil+1)/N$ in light of Theorem \ref{dbbdp}(b). On the other hand, if $\xi=\xi_2$, the chosen sample quantile is the $((1-\beta)N/2)$th smallest element of the sample so sending any $(1-\beta)N/2$ points in the sample to $-\infty$ also sends the chosen sample quantile to $-\infty$. Thus the breakdown point at any sample must be $(1-\beta)N/2\in\mathbb{Z}$. It can similarly be argued that if we instead choose the largest element of the sample quantile set, the breakdown point is $(1-\beta)N/2\in\mathbb{Z}$ if $\xi=\xi_1$ and $(\lceil(1-\beta)N/2\rceil+1)/N$ if $\xi=\xi_2$.

    When $\beta=0$, use the sample median function that chooses the sample median $m$ of smallest magnitude $\lvert m\rvert$, and in the event of two such sample median, chooses the negative one. Suppose we corrupt $N/2$ or fewer of the elements. There are 3 cases to consider: (i) all of the uncorrupted elements are non-negative, (ii) all of the uncorrupted elements are non-positive, and (iii) some of the uncorrupted elements are negative and some are positive. The median set of a sample is the closed interval between the $(N/2)$th and $(N/2+1)$th smallest elements, so in case (i), the chosen sample median cannot be greater than the largest of the uncorrupted elements or less than 0, in case (ii), it cannot be greater than 0 or less than the smallest of the uncorrupted elements, and in case (iii), it cannot be greater than the largest of the uncorrupted elements or less than the smallest. Therefore in every case, it is contained in $[\min\{0,X_1,\ldots,X_N\},\max\{0,X_1,\ldots,X_N\}]$ and the breakdown point is $(\lceil N/2\rceil+1)/N$ for all $(X_1,\ldots,X_N)\in M^N$. On the other hand, the sample median function that chooses the smallest sample median as in the previous paragraph is translation equivariant, and hence the breakdown point at any sample must be $\lfloor (N+1)/2\rfloor/N=(\lceil N/2\rceil)/N$ by Theorem 2.1 in \cite{Lopuhaa1991}.
\end{proof}

\begin{proof}[Proof of Theorem \ref{dbbdp2}]
Fix $\beta\in(0,1]$ and some $y\in M$. Consider the data set $(X_1,\ldots,X_N)$. Letting $r_1=(1/N)\sum_{i=1}^N(1+\beta)d(X_i,y)$ and $r_2=\max_{i=1,\ldots,N}d(X_i,y)$, for any point $y^*$ satisfying $d(y,y*)>r:=r_1/(1-\beta)+r_2$,
\begin{align*}
    \frac{1}{N}\sum_{i=1}^N\rho(X_i,y^*;\beta,\xi)&\geq \frac{1}{N}\sum_{i=1}^N(1-\beta)d(y^*,X_i) \\
    &\geq \frac{1}{N}\sum_{i=1}^N(1-\beta)(d(y^*,y)-d(X_i,y)) \\
    &>r_1;
\end{align*}
therefore $y^*$ cannot be a sample $(\beta,\xi)$-quantile for any $\xi\in\partial M$; that is, all sample $(\beta,\xi)$-quantiles are contained in the closed geodesic ball of radius $r$ centered at $y$. 

Let $N'$, $\hat q_\xi$, $R_\xi$, $\hat q'_\xi$ and $\alpha$ be as in the proof of Lemma \ref{dbinfbdp} and corrupt the first $N'$ elements of the data set. Then for any $\xi\in\partial M$, 
    \begin{equation}\begin{aligned} \label{dbr}
    R:=\frac{(1+\beta)(r+r_2)}{1-\beta}&\geq\max_{i=1,\ldots,N}\frac{(1+\beta)(d(\hat q_\xi,X_1)+d(X_i,X_1))}{1-\beta} \\
    &\geq \max_{i=1,\ldots,N}\frac{(1+\beta)d(\hat q_\xi,X_i)}{1-\beta} \\
    &\geq R_\xi.
    \end{aligned}\end{equation}
    Thanks to (\ref{dbupper}) in the proof of Lemma \ref{dbinfbdp} and (\ref{dbr}),
    \begin{align*}
    &\lvert d(\hat q_\xi,\hat q'_\xi)\rvert\leq 2R\max\bigg\{2, \frac{(1-\beta^2)N}{2\alpha}+(1-\beta)\bigg\}.
    \end{align*}
    Crucially, $R,r_1,r_2$ and $r$ are independent of $\xi$ and the specific corruption of the $N'$ data points. Thus by the triangle inequality and the Cartan-Hadamard theorem, and denoting the original sample $(X_1,\ldots,X_N)$ by $\boldsymbol{X}_N$ and the corrupted sample by $\boldsymbol{X}_N'$,
    \begin{align*}
        &\lVert\log_{\hat m_N(\boldsymbol{X}_N)}(\hat q_{k,N}(\boldsymbol{X}_N'))-\log_{\hat m_N(\boldsymbol{X}_N)}(\hat p_{k,N}(\boldsymbol{X}_N'))\rVert \\
        &\leq d(\hat q_{k,N}(\boldsymbol{X}_N'),\hat p_{k,N}(\boldsymbol{X}_N')) \\
        &\leq d(\hat q_{k,N}(\boldsymbol{X}_N'),\hat q_{k,N}(\boldsymbol{X}_N))+d(\hat q_{k,N}(\boldsymbol{X}_N),y)+d(y,\hat p_{k,N}(\boldsymbol{X}_N))+d(\hat p_{k,N}(\boldsymbol{X}_N),\hat p_{k,N}(\boldsymbol{X}_N')) \\
        &\leq 4R\max\bigg\{2, \frac{(1-\beta^2)N}{2\alpha}+(1-\beta)\bigg\}+2r.
    \end{align*}
     for each $k=1,\ldots,K$, and this upper bound is finite. Therefore since $\hat \delta_{1,N}^\beta$ and $\hat \delta_{2,N}^\beta$ are finite calculated at the original sample, the breakdown points are everywhere at least $(\lceil(1-\beta)N/2\rceil)/N$.
\end{proof}

\subsection{Proof of Theorem~\ref{dbext}}\label{dbproof_ext}

\begin{proof}
(a) Fixing $p,x\in M$, $\alpha_{x,p}(t):=\langle\xi_x,\log_x(\gamma_p^\xi(t))\rangle$ idenfities the projection of $\log_x(\gamma_p^\xi(t))$ onto the image $\log_x(\gamma_x^\xi(\mathbb{R}))$; that is, $\log_x(\gamma_x^\xi(\alpha(t)))$ is the point in this image which minimizes the distance to $\log_x(\gamma_p^\xi(t))$. Therefore for $t\geq 0$,
\begin{equation}\begin{aligned} \label{dbmany}
0\leq d(\gamma_p^\xi(t),x)^2-\alpha(t)^2\leq\lVert\log_x(\gamma_p^\xi(t))-\log_x(\gamma_x^\xi(t))\rVert^2\leq d(\gamma_p^\xi(t),\gamma_x^\xi(t))^2\leq B^2,
\end{aligned}\end{equation}
for some constant $B$; the third inequality results from the Cartan-Hadamard theorem and the fourth from the fact that $\gamma_p^\xi$ and $\gamma_x^\xi$ are asymptotic. Because $d(\gamma_p^\xi(t),x)\geq \lvert d(p,x)-d(p,\gamma_p^\xi(t))\rvert\geq t-d(p,x)\rightarrow\infty$ as $t\rightarrow\infty$, (\ref{dbmany}) implies that 
\begin{equation}\begin{aligned} \label{dbalpha}
\langle\xi_x,\log_x(\gamma_p^\xi(t))\rangle\rightarrow\infty
\end{aligned}\end{equation}
too; that is, $\alpha_{X,p}(t)$ converges to $\infty$ (almost) surely. Assumption 2.1 and Proposition \ref{dbbasic}(a) imply that $\alpha_{X,p}(t)=\rho(X,\gamma_p^\xi(t);1,\xi)-\rho(X,\gamma_p^\xi(t);0,\xi)$ is integrable, so fixing $\epsilon>0$, 
\begin{align*}
    &\lim\inf_{t\rightarrow\infty}E[\alpha_{X,p}(t)] \\
&=\lim\inf_{t\rightarrow\infty}(E[\alpha_{X,p}(t);\alpha_{X,p}(t)\leq\epsilon]+E[\alpha_{X,p}(t);\alpha_{X,p}(t)>\epsilon]) \\
&\geq \epsilon\lim_{t\rightarrow\infty}P(\alpha_{X,p}(t)>\epsilon) \\
&=\epsilon;
\end{align*}
in the third line, the limit is 1 since $\alpha_{X,p}(t)\rightarrow \infty$ in probability. Since the choice of $\epsilon$ was arbitrary,
\begin{equation}\begin{aligned} \label{dbfirstlim}
    \lim_{t\rightarrow\infty} E[\alpha_{X,p}(t)]=\infty.
\end{aligned}\end{equation}

By (\ref{dbmany}), binomial expansion, and (\ref{dbalpha}), for $t\geq 0$.
\begin{align*}
    \rho(x,p;1,\xi)\rangle\leq& \alpha_{x,p}(t)\bigg(1+\frac{B^2}{\alpha_{x,p}(t)^2}\bigg)^{1/2}-\alpha_{x,p}(t) \\
    =&-\alpha_{x,p}a(t)+\alpha_{x,p}(t)\bigg(1+\frac{B^2}{2\alpha_{x,p}(t)^2}-\frac{B^4}{8\alpha_{x,p}(t)^4}+\ldots\bigg) \\
    =&\frac{B^2}{2\alpha_{x,p}(t)}-\frac{B^4}{8\alpha_{x,p}(t)^3}+\ldots \\
    \rightarrow& 0
\end{align*}
as $t\rightarrow\infty$; that is, $\rho(X,p;1,\xi)$ converges to 0 (almost) surely. Then fixing $\epsilon>0$, 
\begin{align*}
&\lim\sup_{t\rightarrow\infty}G^{1,\xi}(\gamma_p^\xi(t)) \\
&=\lim\sup_{t\rightarrow\infty} (E[\rho(X,\gamma_p^\xi(t);1,\xi);\rho(X,\gamma_p^\xi(t);1,\xi)\leq\epsilon]+E[\rho(X,\gamma_p^\xi(t);1,\xi);\rho(X,\gamma_p^\xi(t);1,\xi)>\epsilon]) \\
&\leq \epsilon\lim_{t\rightarrow\infty}P(\rho(X,\gamma_p^\xi(t);1,\xi)\leq\epsilon) +\lim_{t\rightarrow\infty}E[\rho(X,\gamma_p^\xi(t);1,\xi);\rho(X,\gamma_p^\xi(t);1,\xi)>\epsilon]\\
&=\epsilon;
\end{align*}
in the third line, the first limit is 1 since $\rho(X,p;1,\xi)\rightarrow 0$ in probability, and the second is 0 by the dominated convergence theorem since $\rho(X,p;1,\xi)\rightarrow 0$ almost surely. Since the choice of $\epsilon$ was arbitrary,
\begin{equation}\begin{aligned} \label{dbsecondlim}
    \lim_{t\rightarrow\infty}G^{1,\xi}(\gamma_p^\xi(t))=0.
\end{aligned}\end{equation}

If $\beta=1$, the assumption about the support guarantees that $G^{\beta,\xi}(p)>0$ for all $p\in M$, and so (\ref{dbsecondlim}) implies that $G^{\beta,\xi}$ does not achieve a minimum. On the other hand, if $\beta>1$, (\ref{dbfirstlim}) and (\ref{dbsecondlim}) imply $G^{\beta,\xi}(\gamma_p^\xi(t))=G^{1,\xi}(\gamma_p^\xi(t))-(\beta-1)E[\alpha_{X,p}(t)]\rightarrow -\infty$ as $t\rightarrow\infty$, again impying that $G^{\beta,\xi}$ does not achieve a minimum. This proves one direction; the other follows immediately from Proposition \ref{dbbasic}(b). 

(b) Suppose the contrary. Then there is some subsequence $q_{m'}$ of $\{q_m\}$ which is bounded, and therefore there exists a further subsequence $\{q_{m''}\}$ of $\{q_{m'}\}$ which converges in $M$. So we can assume without loss of generality that $q_m$ converges to some $q_\infty\in M$. The map on $[0,\infty)\times M$ defined by $(\beta,p)\mapsto G^{\beta,\xi}(p)$ is continuous because 
\begin{align*}
    &\lvert G^{\beta,\xi}(p)-G^{\beta',\xi}(p')\rvert \\
    \leq&\lvert E[d(p,X)-d(p',X)]\rvert+\lvert E[\langle \xi_X,\beta\log_X(p)-\beta'\log_X(p)\rangle]\rvert \\
    \leq& E[d(p,p')]+\beta' \lvert E[\langle \xi_X,\log_X(p)-\log_X(p')\rangle]\rvert+\lvert\beta-\beta'\rvert E[\langle \xi_X,\log_X(p)\rangle] \\
    \leq& d(p,p')+\beta'd(p,p')+\lvert\beta-\beta'\rvert E[d(p,X)] \\
    \rightarrow& 0
\end{align*}
as $(\beta',p')\rightarrow (\beta,p)$; the third inequality is a consequence of the Cauchy-Schwarz inequality and the Cartan-Hadamard theorem. Therefore since $G^{\beta_{m},\xi}(q_{m})\leq G^{\beta_{m},\xi}(p)$ for all $p\in M$, letting $m\rightarrow \infty$ shows that $G^{1,\xi}(q_\infty)\leq G^{1,\xi}(p)$ for all $p\in M$, contradicting (a).

(c) Fix $\epsilon>0$. By (\ref{dbsecondlim}) in the proof of (a), there exists some $p\in M$ for which $0<G^{1,\xi}(p)<\epsilon/2$. There also exists some integer $m_0$ for which $m\geq m_0$ implies $\lvert G^{\beta_m,\xi}(p)-G^{1,\xi}(p)\rvert=(1-\beta_m)E[\langle \xi_X,\log_X(p)\rangle]<\epsilon/2$ since $\beta_m\rightarrow 1$. Therefore $0\leq G^{\beta_m,\xi}(q_m)\leq G^{\beta_m,\xi}(p)<\epsilon$ for all $m\geq m_0$, and $E[d(q_m,X)-\beta_m\langle \xi_X,\log_X(q_m)\rangle]\rightarrow 0$. Since the integrand is non-negative, this is equivalent to $d(q_m,X)-\beta_m\langle \xi_X,\log_X(q_m)\rangle\rightarrow 0$ in mean, and therefore in probability. By (b), $1/d(q_m,X)\rightarrow 0$ (almost) surely, and therefore in probability. Multiplying these two together gives $1-\beta_m\langle\xi_X,\log_X(q_m)/d(q_m,X)\rangle\xrightarrow{p} 0$. Rearranging and multiplying by $1/\beta_m\rightarrow 1$ gives 
\begin{align*}
\langle\xi_X,\log_X(q_m)/d(q_m,X)\rangle\xrightarrow{p} 1.
\end{align*}
Since the left-hand side is bounded above by 1 by the Cauchy-Schwarz inequality, this convergence also holds in mean by Vitali's convergence theorem. Finally, $\lVert \log_X(q_m)/d(q_m,X)-\xi_X\rVert^2=2-2\langle\xi_X,\log_X(q_m)/d(q_m,X)\rangle\xrightarrow{L^1} 0$, from which the desired conclusion follows.
        \end{proof}

 \subsection{Proofs for the results in Section~\ref{dbhyp}}\label{dbproof_hyp}

  \begin{proof}[Proof of Proposition~\ref{dbgrad}]
     (a) Fix a geodesic $\gamma$ that satisfies $\gamma(0)=p$. By the smoothness of the map $(p,x)\mapsto (x,\log_x(p)$ on $M\times M$, the map $p\mapsto\langle\xi_x,\log_x(p)\rangle$ is smooth on all of $M$, so the inner product of its gradient at $p$ and $\gamma'(0)$ is $\lim_{t\rightarrow 0}-\beta\langle \xi_x,\log_x(\gamma(t))-\log_x(\gamma(0))\rangle/t$, the absolute value of which is less than or equal to $\lim\sup_{t\rightarrow 0}\beta d(\gamma(0),\gamma(t))/t=\beta\lVert\gamma'(0)\rVert$
     by the Cauchy-Schwarz inequality and Cartan-Hadamard theorem. Since this is true for all geodesics which pass through $p$ at $t=0$, the norm of the gradient is less than or equal to $\beta$. Then since the map $p\mapsto d(p,x)$ is also smooth on $M\backslash\{x\}$, and its gradient at $p\neq x$ is known to be $\log_p(x)/d(p,x)$, whose norm is 1, the conclusion follows.

     (b) For some interval $I$ containing , take any smooth curve $\alpha:I\rightarrow M$ satisfying $\alpha(0)=p$. Then
     \begin{align*}
         \frac{d}{dt}\langle \xi_x,\log_x(\alpha(t))\rangle\bigg|_{t=0}=\langle\xi_x,d(\log_x)_p\alpha'(0)\rangle=\langle d(\log_x)_p^\dagger\xi_x,\alpha'(0)\rangle.
     \end{align*}
     The gradient of with respect to $p$ of $d(x,p)$ is known to be $\log_p(x)/d(p,x)$ when $p\neq x$, from which the first equality follows. The gradient of $d(x,p)$ can also be expressed as $d(\log_x)_p^\dagger\log_x(p)/d(p,x)$ since
     \begin{align*}
         \frac{d}{dt}d(\alpha(t),x)\bigg|_{t=0}&=\langle \log_x(p),\log_x(p)\rangle^{1/2} \\
         &=\frac{1}{2}\frac{\langle \log_x(p),d(\log_x)_p\alpha'(0)\rangle+\langle d(\log_x)_p\alpha'(0),\log_x(p)\rangle}{d(p,x)} \\
         &=\bigg\langle d(\log_x)_p^\dagger \frac{\log_x(p)}{d(p,x)},\alpha'(0)\rangle\bigg\rangle,
     \end{align*}
     from which the second equality follows.
 \end{proof}

  \begin{proof}[Proof of Theorem \ref{dbgradient}]

			Define a smooth path $c:I\rightarrow M$, where $I\subset \mathbb{R}$ is an open interval containing 0, such that $c(0)=p$, and a family of geodesic $\{\gamma_s\}_{s\in I}$ by $\gamma_s(t)=\exp_x(t\cdot\log_x(c(s)))$. Note that $J(t):=(\partial/\partial s) \gamma_s(t)|_{s=0}$ is a Jacobi field along $\gamma_0$, and $\log_x(c(s))=\gamma_s'(0)=(\partial/\partial t)\gamma_s(t)|_{t=0}$. Then by the symmetry of the covariant derivative,
   \begin{equation}\begin{aligned} \label{dbsymmet}
       \frac{D}{ds}\log_x(c(s))\bigg|_{s=0}=\frac{D}{\partial s}\bigg(\frac{\partial}{\partial t}\gamma_s(t)\bigg|_{t=0}\bigg)\bigg|_{s=0}=\frac{D}{\partial t}\bigg(\frac{\partial}{\partial s}\gamma_s(t)\bigg|_{s=0}\bigg)\bigg|_{t=0}=\frac{D}{dt}J(t)\bigg|_{t=0}.
   \end{aligned}\end{equation}
   As mentioned, Jacobi fields can be explicitly calculated on locally symmetric spaces. Take $e_i$ and $g_i$ as defined prior to the theorem. Each $e_i$ can be extended by parallel transport to a vector field $W_i$ along $\gamma_0$ taking value $W_i(t)$ at $\gamma_0(t)$ such that $W_i(1)=e_i$. Because $e_i$ and $\kappa_i$ are eigenvectors and eigenvalues, respectively, of $CO_{e_1}$, $CO_{e_1}(e_i)=\kappa_ie_i$, and then $R(W
_1(t),W_i(t))W_1(t)=\kappa_iW_i(t)$ since $R$ is parallel under parallel transport on locally symmetric spaces. Then $J(t)=\sum_{i=1}^nJ_i(t)W_i(t)$ for some smooth real-valued functions $J_i$, and it must satisfy the Jacobi equation:
   \begin{align*}
\frac{D^2}{dt^2}J(t)+R(\gamma_0'(t),J(t))\gamma_0'(t)&=\sum_{i=1}^n\{J_i''(t)W_i(t)+J_i(t)d(p,x)^2R(W_1(t),W_i(t))W_1(t)\} \\
&=\sum_{i=1}^n\{J_i''(t)+J_i(t)d(p,x)^2\kappa_i\}W_i(t) \\
&=0,
   \end{align*}
implying for each $i$ that $J_i(t)=A_if_i(t)+B_ig_i$(t), where
\begin{align*}
    f_i(t)=\begin{cases}
        \frac{1}{d(p,x)\sqrt{-\kappa_i}}\cosh(d(p,x)\sqrt{-\kappa_i}t) &\text{if $\kappa_i<0$,} \\
        1 &\text{if $\kappa_i=0$,}
    \end{cases}
\end{align*}
and $A_i,B_i$ are constants. Because $J(0)=(d/ds)\gamma_s(0)|_{s=0}=(d/dt)x=0$, we know that each $A_i=0$. On the other hand, $J(1)=(d/ds)\gamma_s(1)|_{s=0}=(d/ds)c(s)|_{s=0}=c'(0)$, so each $B_i=\langle W_i(1),c'(0)\rangle/g_i(1)$, which is $\langle \log_p(x)/d(p,x),c'(0)\rangle$ when $i=1$. Then
\begin{align*}
    \frac{D}{Dt}J(t)=\sum_{i=1}^n\frac{(d/dt)g_i(t)}{g_i(1)}\langle e_i,c'(0)\rangle W_i(t).
\end{align*}
By (\ref{dbsymmet}), the derivative of $\langle \xi_x,\log_x(c(s))\rangle$ at $s=0$ is therefore
\begin{align*}
    \bigg\langle \xi_x,\frac{D}{ds}\log_x(c(s))\bigg|_{s=0}\bigg\rangle=\bigg\langle\sum_{i=1}^n \frac{(d/dt)g_i(t)|_{t=0}}{g_i(1)}\langle \xi_x,W_i(0)\rangle e_i,c'(0)\bigg\rangle,
\end{align*}
completing the proof.

		\end{proof}

      \begin{proof}[Proof of Corollary \ref{dbhypgrad}]
         In (\ref{dbsymgrad}), each $\kappa_i=\kappa$ except for $\kappa_1$ which is 0, so 
\begin{align*}
		&\nabla \rho(x,p;\beta,\xi)   \\
  &=-\frac{\log_p(x)}{d(p,x)}-\beta\bigg(\bigg\langle\Gamma_{x\rightarrow p}(\xi_x),\frac{\log_p(x)}{d(p,x)}\bigg\rangle\frac{\log_p(x)}{d(p,x)} \\
  &\qquad+\frac{d(p,x)\sqrt{-\kappa}}{\sinh(d(p,x)\sqrt{-\kappa})}\bigg(\Gamma_{x\rightarrow p}(\xi_x)-\bigg\langle\Gamma_{x\rightarrow p}(\xi_x),\frac{\log_p(x)}{d(p,x)}\bigg\rangle\frac{\log_p(x)}{d(p,x)}\bigg) \bigg) \\
  &=-\frac{\log_p(x)}{d(p,x)}-\beta\bigg(\frac{d(p,x)\sqrt{-\kappa}}{\sinh(d(p,x)\sqrt{-\kappa})}\Gamma_{x\rightarrow p}(\xi_x) \\
  &\qquad+\bigg(\frac{d(p,x)\sqrt{-\kappa}}{\sinh(d(p,x)\sqrt{-\kappa})}-1\bigg)\bigg\langle\xi_x,\frac{\log_x(p)}{d(p,x)}\bigg\rangle\frac{\log_p(x)}{d(p,x)}\bigg).
		\end{align*}
     \end{proof}

  \end{appendix}

 \bibliographystyle{apalike}
\bibliography{references}

\end{document}